\documentclass[a4paper,10pt]{article}

\usepackage{braket} 
\usepackage{titlesec}
\usepackage{bm} 
\usepackage{cancel} 
\usepackage{enumerate}
\usepackage{amssymb, amsmath, amsthm}
\usepackage{multirow}
\usepackage{mathrsfs} 
\usepackage{breqn} 
\usepackage{fancyhdr} 
\usepackage{lastpage} 
\usepackage{extramarks} 
\usepackage{changepage}
\usepackage{graphicx}
\usepackage{empheq}
\usepackage{caption,subcaption}
\usepackage{xfrac} 
\usepackage{graphics,setspace} 
\usepackage{pgfplots}
\usepackage{relsize}
\usepackage{floatrow}
\usepackage{bbm} 
\usepackage{nth}
\usepackage{color}
\usepackage{listings}
\usepackage{tikz}
\usepackage{hyperref}
\usetikzlibrary{spy}

\newfloatcommand{capbtabbox}{table}[][\FBwidth]

\DeclareMathOperator{\RE}{Re} 
\DeclareMathOperator{\IM}{Im} 

\DeclareMathOperator{\sn}{sn} 
\DeclareMathOperator{\cn}{cn} 
\DeclareMathOperator{\dn}{dn} 
\AtBeginDocument{\renewcommand{\d}{\textrm{d}}}

\newcommand{\cc}{^*}

\newcommand{\pd}[2]{\frac{\partial #1}{\partial #2}}

\newcommand{\DD}[3][]{\frac{\d^{#1} #2}{\d #3^{#1}}}

\newcommand{\D}[2][]{\frac{\d^{#1}}{\d #2^{#1}}}

\newcommand{\abs}[1]{\left| #1 \right|}
\newcommand{\norm}[1]{\left| \left| #1 \right| \right|}

\newcommand{\R}{\mathbb{R}} 

\newcommand{\N}{\mathbb{N}} 
\newcommand{\Z}{\mathbb{Z}} 
\newcommand{\C}{\mathbb{C}} 
 
\newcommand{\cB}{\mathcal{B}} 
\newcommand{\cL}{\mathcal{L}} %
\newcommand{\cK}{\mathcal{K}} %

\newcommand{\lt}{\left}
\newcommand{\rt}{\right}
\newcommand{\Langle}{\lt\langle}
\newcommand{\Rangle}{\rt\rangle}
\newcommand{\ip}[1]{\Langle #1 \Rangle}

\newcommand{\eps}{\epsilon}

\newcommand{\lam}{\lambda}

\newcommand{\0}[1]{\mathcal{O}\lt(#1\rt)}
\newcommand{\bigo}[1]{\mathcal{O}} 

\renewcommand{\L}{\mathcal{L}}

\newcommand{\eqnref}{\eqref} 
\let\oldchi\chi
\renewcommand{\chi}{\raisebox{1pt}{$\oldchi$}}
\newcommand{\ie}{\textit{i.e.},~}
\newcommand{\eg}{\textit{e.g.},~}
\newcommand{\vs}{\textit{vs.}~}

\newenvironment{definition}[1][Definition]{\begin{trivlist}
\item[\hskip \labelsep {\bfseries #1}]}{\end{trivlist}}

\newtheorem{theorem}{Theorem}
\newtheorem{lemma}{Lemma}[theorem]

\theoremstyle{definition}
\newtheorem{remark}{Remark}[theorem]

\usepackage[utf8]{inputenc}
\usepackage{epstopdf}
\usepackage{enumitem}
\usepackage[title]{appendix}
\graphicspath{{images/}}
\usetikzlibrary{arrows}
\usetikzlibrary{decorations.markings}
\newcommand{\te}[1]{#1}
\newcommand{\tee}[1]{#1}
\newcommand{\teee}[1]{#1}
\usepackage{ulem}

\title{The orbital stability of elliptic solutions of the Focusing Nonlinear
Schr\"odinger Equation}
\author{Bernard Deconinck and Jeremy Upsal
  \\
Department of Applied Mathematics,\\
University of Washington,\\
Seattle, WA 98195, USA}

\begin{document}

\maketitle
\begin{abstract}
      We examine the stability of the elliptic solutions of the focusing
      nonlinear Schr\"{o}dinger equation (NLS) with respect to subharmonic
      perturbations.  Using the integrability of NLS, we discuss the spectral
      stability of the elliptic solutions, establishing that solutions of
      smaller amplitude are stable with respect to larger classes of
      perturbations. We show that spectrally stable solutions are orbitally
      stable by constructing a Lyapunov functional using higher-order conserved
      quantities of NLS.
\end{abstract}

\section{Introduction}
The focusing, one-dimensional, cubic Nonlinear Schr\"{o}dinger equation (NLS),
\begin{align}
    i\Psi_t + \frac12 \Psi_{xx} + \Psi\abs{\Psi}^2 = 0,\label{eqn:NLS}
\end{align}
is a universal model for a variety of physical phenomena
\cite{chen2012introduction, gross1961structure, kivshar2003optical,
pitaevskii1961vortex, sulem2007nonlinear, zakharov1968stability}. In 1972,
Zakharov and Shabat \cite{zakharov_shabat} found its Lax Pair and the explicit
expression for the one-soliton solution. The \te{orbital} stability of the soliton was first
proved in 1982 by Cazenave and Lions \cite{cazenave-Lions} and later by
Weinstein \cite{MR820338} using Lyapunov techniques, as used here. Even with
such a rich history, a full stability analysis in the periodic setting has not
been completed. The simplest periodic solutions are the genus-one or elliptic
solutions \teee{(Section~\ref{section:ellipticSolns})}. Rowlands
\cite{rowlands1974stability} was the first to study their stability using
perturbation methods. Since then, Gallay and H\v{a}r\v{a}g\c{u}s have examined
the stability of small-amplitude elliptic solutions
\cite{gallay_haragus_smallAmplitude} and proven \te{orbital} stability with
respect to perturbations of the same period as the underlying solution
\cite{gallay_haragus_OrbitalStability} (\ie coperiodic perturbations).
Gustafson, Le Coz, and Tsai \cite{gustafson2017stability} establish instability
for the elliptic solutions with respect to sufficiently large perturbations. The
analysis of spectral instability with respect to perturbations of an integer
multiple of the period (\ie subharmonic perturbations) was completed in
\cite{deconinck2017stability}.

In this work we build upon the results in \cite{deconinck2017stability} to
examine the stability of elliptic solutions of arbitrary amplitude. Only
classical solutions of \eqref{eqn:NLS} and classical perturbations of those
solutions are considered in this paper. \teee{An outline of the steps followed and the conclusions obtained is given below.

\begin{enumerate}

\item Spectral stability is considered in Section \ref{section:linearStability}.
  This is motivated by considering the simpler case of the well-known Stokes
  waves in Section \ref{section:Stokes}. For these solutions, all operators
  involved have constant coefficients, and all calculations are explicit. We get
  to the spectrum of the operator obtained by linearizing about a solution
  through its connection with the \textit{Lax spectrum}. To this end, we
  introduce the Lax pair and its spectrum in Section
  \ref{section:stability_setup}. The results in Section
  \ref{section:LaxSpectrumDef} are from \cite{deconinck2017stability} while the
  results in Section \ref{section:LaxSpectrum} and all subsequent sections are
  new. Section \ref{section:EllipticSpectralStability} contains our main
  spectral stability result: solutions are spectrally stable with respect to
  subharmonic perturbations if the solution parameters meet a given sufficiency
  condition (Theorem \ref{thm:mainNTPStabilityResult}). This condition is shown
  to be necessary in most cases and is discussed in Appendix \ref{Appendix}.  In
  essence, Theorem \ref{thm:mainNTPStabilityResult} establishes that solutions
  of ``smaller amplitude'' are spectrally stable with respect to a larger class
  of subharmonic perturbations, \ie subharmonic perturbations of larger period.
  The notion of ``smaller amplitude'' is made more precise in
  Section~\ref{section:EllipticSpectralStability}.

\item In Section \ref{section:krein}, we examine how instabilities depend on the
  parameters of the solution. The orbital stability results of
  Section~\ref{section:orbitalStability} rely crucialy on understanding the
  spectrum for stable compared to unstable solutions. Thus we carefully examine
  the transition from stable to unstable dynamics as solution parameters are
  changed.

\item Finally, in Section \ref{section:orbitalStability} we use a Lyapunov
  method \cite{grillakis1987stability,kapitula2007stability, maddocksSachs} to
  prove (nonlinear) orbital stability in the cases where spectral stability
  holds. Our main result is found at the end of the section: we establish the
  orbital stability of almost all solutions that are spectrally stable. The only
  solutions for which such a result eludes us are those whose solution
  parameters are on the boundary of the parameter regions specifying with
  respect to which subharmonic perturbations the solutions are spectrally
  stable.

\end{enumerate}
}

This paper is part of an ongoing research program of analyzing the stability of
periodic solutions of integrable equations (\cite{bottman2009kdv,bottman2011elliptic,
  deconinck2010orbital, SGstability,
mkdvOrbitalStability, deconinck2017stability, kdvFGOrbitalStability}). The present work
is the first in the program to establish a nonlinear stability result for periodic
solutions for which the underlying Lax pair is not self adjoint.

\section{Elliptic solutions of focusing NLS}\label{section:ellipticSolns}
In this paper we study solutions of \eqref{eqn:NLS} whose only change in time is
through a constant phase-change. Such solutions are stationary solutions of
\begin{align}
  i \psi_t + \omega \psi + \frac12 \psi_{xx} + \psi\abs{\psi}^2 = 0,
  \label{eqn:stationaryPDE}
\end{align}
found by defining $\Psi(x,t) = e^{-i\omega t}\psi(x,t)$. Time-independent
solutions to \eqref{eqn:stationaryPDE} satisfy
\begin{align}
  \omega \phi + \frac12 \phi_{xx} + \phi\abs{\phi}^2 =
  0\label{eqn:stationaryODE},
\end{align}
and are expressed in terms of elliptic functions as
\begin{align}
    \Psi = e^{-i\omega t}\phi(x) =  R(x)e^{i\theta(x)}e^{-i\omega t}
    \label{eqn:ellipticSolns},
\end{align}
with
\begin{subequations}
\begin{align}
 R^2(x) &= b-k^2\sn^2(x,k) \label{eqn:Rsquared}, &
 \omega &= \frac12(1+k^2-3b),\\
 \theta(x) &= c\int_0^x \frac{1}{R^2(y)}~\d y, &
 c^2 &= b(1-b)(b-k^2) \label{eqn:c},
\end{align}
\end{subequations}
where $\sn(x,k)$ is the Jacobi elliptic sn function with elliptic modulus $k$
\cite[Chapter 22]{NIST:DLMF}. The parameters $b$ and $k$ are constrained by
\begin{align}
 0\leq k < 1, \qquad
 k^2 \leq b \leq 1 \label{eqn:parameterSpace},
\end{align}
see Figure \ref{fig:parameterSpace}. The solutions formally limit to the soliton
as $k\to 1$, which is omitted from our studies.
When $k=0$ and $b\neq 0$, \eqref{eqn:ellipticSolns} reduces to a so-called
Stokes wave (Section \ref{section:Stokes}). The boundary values, $b=k^2$ and
$b=1$, are special cases. In both cases $c= 0$ so $\theta=0$ and the solutions
are said to have trivial phase. When $c\neq 0$, the solutions have non-trivial
phase (NTP).  We call $\phi(x) = k\cn(x,k)$ and $\phi(x) = \dn(x,k)$ the cn and
dn solutions corresponding to $b=k^2$ and $b=1$, respectively. Here $\cn(x,k)$
and $\dn(x,k)$ are the Jacobi elliptic cn and dn functions with elliptic modulus
$k$ \cite[Chapter 22]{NIST:DLMF}. The trivial-phase solutions are periodic, with
periods $4K(k)$ and $2K(k)$ for the cn and dn solutions respectively, where
\begin{align}
    K(k) &:= \int_0^{\pi/2} \frac{\d y}{\sqrt{1-k^2\sin^2(y)}}\label{eqn:Kk},
\end{align}
the complete elliptic integral of the first kind \cite[Chapter 19]{NIST:DLMF}.

\begin{remark}
  The nontrivial-phase solutions are typically quasi-periodic but only the
  $x$-periodic amplitude $R^2(x)$ appears in our analysis. Therefore, unless
  otherwise stated, any mention of the periodicity of the solutions is in
  reference to the period of the amplitude which is $T(k) = 2K(k)$ for all
  solutions.
\end{remark}

\begin{figure}
  \centering
  \begin{tikzpicture}[scale=5]
    \node[below] (origin) at (-0.05,0) {0};
    \draw[->] (-0.1,0) -- (1.2,0) node[right] {$k$};
    \draw[->] (0,-0.1) -- (0,1.2) node[above] {$b$};
    \filldraw [color=blue, fill=gray!30, line width=0.7mm, domain=0:1, scale=1.0, smooth,
        variable=\x ]
          (0,0)
          -- plot ({\x},{(\x)^2})
          -- (0,1)
          -- cycle;
    \node[draw, circle, black, thick, fill=white, inner sep=0.7mm] (11) at (1,1) {};
    \draw (1, -0.03) -- (1, 0.03) node[below, yshift=-3mm] {$1$};
    \draw (-0.03, 1) -- (0.03, 1) node[left, xshift=-3mm] {$1$};
    \node[align=center, inner sep=0.8mm, draw] at (0.38, 0.7) {Nontrivial-phase
            \\ (NTP) solutions};
    \node[draw] (solitons) at (1.5,1.2) {soliton solution};
    \draw[decoration={markings,mark=at position 1 with
      {\arrow[scale=3,>=stealth]{>}}}, postaction={decorate}] (solitons) --
      (11);
    \node[draw] (dn) at (0.6, 1.2) {dn solutions};
    \node (dnanchor) at (0.5, 1.0) {};
    \draw[decoration={markings,mark=at position 1 with
      {\arrow[scale=3,>=stealth]{>}}}, postaction={decorate}] (dn) --
      (dnanchor);
    \node (cnanchor) at (0.6, 0.36) {};
    \node[draw] (cn) at (1.1, 0.36) {cn solutions};
    \draw[decoration={markings,mark=at position 1 with
      {\arrow[scale=3,>=stealth]{>}}}, postaction={decorate}] (cn) --
      (cnanchor);
    \node (stokesanchor) at (0, 0.66) {};
    \node[draw] (stokes) at (-0.5, 0.66) {Stokes wave solutions};
    \draw[decoration={markings,mark=at position 1 with
      {\arrow[scale=3,>=stealth]{>}}}, postaction={decorate}] (stokes) --
      (stokesanchor);
    \node[circle, fill, inner sep=1.2pt, color=green] (one-quarter) at (0,0.25)
                                                                       {};
    \node[circle, fill, inner sep=1.2pt, color=green] (one-ninth) at (0,0.1111)
                                                                       {};
    \node[circle, fill, inner sep=1.2pt, color=green] (one-sixteenth) at
                                                            (0,0.0625) {};
    \node[circle, fill, inner sep=1.2pt, color=green] (one-thirty-second) at
                                                            (0,0.04) {};
    \node[align=center, draw] (critical) at (-0.54, 0.15625) {Stokes wave
    critical\\ stability values};
    \draw[decoration={markings,mark=at position 1 with
      {\arrow[scale=2,>=stealth]{>}}}, postaction={decorate}] (critical) --
      (one-quarter);
    \draw[decoration={markings,mark=at position 1 with
      {\arrow[scale=2,>=stealth]{>}}}, postaction={decorate}] (critical) --
      (one-ninth);
    \draw[decoration={markings,mark=at position 1 with
      {\arrow[scale=2,>=stealth]{>}}}, postaction={decorate}] (critical) --
      (one-sixteenth);
    \node (moreCritical) at (-0.15,0.05) {\Large$\cdots$};
  \end{tikzpicture}
 \caption{\label{fig:parameterSpace}The parameter space for the elliptic
 solutions \eqref{eqn:ellipticSolns} with solution regions labeled. The first 4
 Stokes wave stability bounds are plotted in green dots on the line $k=0$, at
 which $b = 1/P^2$ for $P\in\{1,2,3,4\}$
 \eqref{eqn:TightStokesSubharmonicStabilityCriteria}.}

\end{figure}
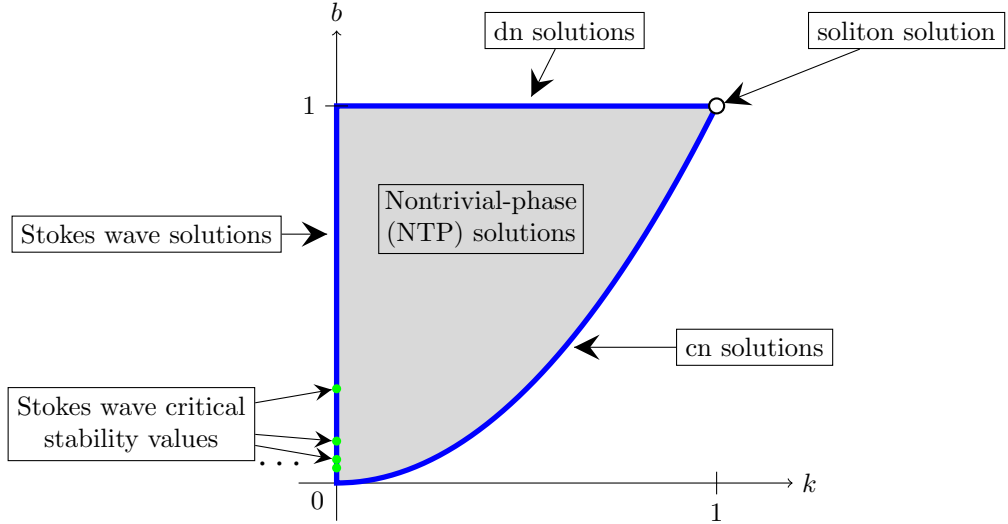

The elliptic solutions can be written in terms of Weierstrass elliptic functions
via
\begin{align}
 \wp( z + \omega_3; g_2, g_3) - e_3 = \lt(\frac{K(k) k}{\omega_1}\rt)^2\sn^2
 \lt(\frac{K(k) z}{\omega_1}\rt) \label{eqn:JacobiToWeierstrass},
\end{align}
where $\wp(z; g_2,g_3)$ is the Weierstrass elliptic $\wp$ function \cite[Chapter
23]{NIST:DLMF} with lattice invariants $g_2,~g_3$ and $\omega_1$ and $\omega_3$
are the half-periods of the Weierstrass lattice. Lastly, $e_1,$ $e_2$, and $e_3$
are the zeros of the polynomial $4t^3 - g_2 t - g_3$, and
\begin{subequations}
\begin{align}
     e_1 &= \frac13(2-k^2),
     \hspace{0.1\textwidth} e_2 = \frac13(2k^2-1),
     \hspace{0.1\textwidth} e_3 = -\frac13(1+k^2)\label{eqn:e0e1e2e3},\\
      g_2 & = \frac43(1-k^2+k^4),
    \hspace{0.2\textwidth} g_3 =
\frac4{27}(2-3k^2-3k^4+2k^6)\label{eqn:g2g3},\\
    \omega_1 &= \int_{e_1}^\infty \frac{\d z}{\sqrt{4z^3 - g_2 z - g_3}} = K(k)
    \label{eqn:omega1omega3},\qquad
    \omega_3 = \int_{-e_3}^\infty \frac{\d z}{\sqrt{4z^3 - g_2z-g_3}} =
    iK\lt(\sqrt{1-k^2}\rt).
\end{align}
\end{subequations}
The Weierstrass form of the elliptic solutions is explained in more detail in
\cite[Section 3.1.3]{painleveHandbook}.

\section{Spectral stability}\label{section:linearStability}

\noindent Spectral stability of elliptic solutions is examined by considering
\begin{align}
 \Psi(x,t) = e^{-i\omega t}e^{i\theta(x)}\lt(R(x) + \eps u(x,t) + \eps i
 v(x,t)\rt) + \0{\eps^2},\label{eqn:perturbation}
\end{align}
where $\eps$ is a small parameter and $u$ and $v$ are real-valued functions of
$x$ and $t$. Substituting this into \eqref{eqn:NLS} and keeping only first-order
in $\eps$ terms gives an autonomous ODE in $t$. Separating variables $(u(x,t),
v(x,t))=e^{\lam t}(U(x),V(x))$ results in the spectral problem
\begin{align}
 \lam \begin{pmatrix} U \\ V \end{pmatrix} =
    \begin{pmatrix} -S & \cL_- \\ -\cL_+ & -S
      \end{pmatrix}\begin{pmatrix} U \\ V \end{pmatrix}
        = J\begin{pmatrix} \cL_+ & S \\ -S &\cL_-
          \end{pmatrix}\begin{pmatrix} U \\ V \end{pmatrix}
          = J\mathscr{L} \begin{pmatrix}  U \\ V \end{pmatrix}
            = \cL \begin{pmatrix} U \\ V
     \end{pmatrix}\label{eqn:spectralProblem},
\end{align}
where
\begin{align}
  \cL = J \mathscr{L},
\end{align}
and
\begin{align}
  \begin{split}
 \cL_- &= -\frac12 \partial_x^2 - R^2(x) - \omega + \frac{c^2}{2R^4(x)}
 ,\\
 \cL_+ &= -\frac12 \partial_x^2 - 3R^2(x) - \omega + \frac{c^2}{2R^4(x)}
 ,\\
 S &= \frac{c}{R^2(x)} \partial_x - \frac{c R'(x)}{R^3(x)}
 .
\end{split}\label{eqn:Lcoeffs}
\end{align}
The stability spectrum is defined as
\begin{align}
  \sigma_{\cL} = \lt\{ \lam \in \C : U, V \in
  C_b^0(\R)\rt\}\label{eqn:stabilitySpectrumDef},
\end{align}
where $C_b^0(\R)$ is the space of real-valued continuous functions, bounded on
the closed real line. Due to the Hamiltonian symmetry of the spectrum
\cite{haragus_kapitula-SpectraPeriodicWaves}, an elliptic solution is spectrally
  stable to perturbations in $C_b^0(\R)$ if $\sigma_{\cL}\subset i\R$.

\subsection{Stability with respect to subharmonic perturbations}
\label{section:floquet}
The elliptic solutions are not stable with respect to general bounded
perturbations \cite{deconinck2017stability}. Therefore, we restrict to
\textit{subharmonic perturbations}. Subharmonic perturbations are those periodic
perturbations whose period is an integer multiple of the fundamental period of a
given elliptic solution. Since the operator $\cL$ has periodic coefficients
\eqref{eqn:Lcoeffs}, the eigenfunctions of the spectral problem
\eqref{eqn:spectralProblem} may be decomposed using a Floquet-Bloch
decomposition \cite{FFHM},
\begin{align}
 \begin{pmatrix} U(x) \\ V(x) \end{pmatrix} = e^{i\mu x}\begin{pmatrix}
 \hat U_\mu(x) \\ \hat V_\mu(x) \end{pmatrix}, \label{eqn:floquet-bloch}
\end{align}
where $\hat U_\mu, \hat V_\mu$ are $T(k)$ periodic and $\mu \in [0, 2\pi/T(k))$.

\begin{definition}
  A \textit{$P$-subharmonic perturbation} of a solution is a perturbation of
  integer multiple $P$ times the period of the solution. A $1$-subharmonic
  perturbation is called a \textit{coperiodic perturbation}.
\end{definition}

\noindent For $P$-subharmonic perturbations,
\begin{align}
  \mu = m\frac{2\pi}{PT(k)}, \qquad m = 0,\ldots, P-1.\label{eqn:muDefn}
\end{align}
Note that $\mu$ may be defined in any interval of length $2\pi/T(k)$ so the
$m=1$ and $m=P-1$ cases are connected via
\begin{align}
  \mu = -\frac{2\pi}{PT(k)} = (P-1)\frac{2\pi}{PT(k)} \quad \mod 2\pi/T(k).
\end{align}
Using the Floquet-Bloch decomposition, $\cL \mapsto \cL_\mu$ with $\partial_x
\mapsto \partial_x + i\mu$ in \eqref{eqn:spectralProblem}. We define the
subharmonic stability spectrum with parameter $\mu$,
\begin{align}
    \sigma_{\mu} = \lt\{\lam \in \C: \hat U_\mu, \hat V_\mu \in
      L^2_{\text{per}}\lt([-T(k)/2,
      T(k)/2]\rt)\rt\},\label{eqn:LmustabilitySpectrum}
\end{align}
where $L^2_{\text{per}}\lt([-L/2,L/2] \rt)$ is the space of square-integrable
functions with period $L$. The spectrum $\sigma_\mu$ consists of isolated
eigenvalues of finite multiplicity.

\subsection{Spectral stability of Stokes Waves}\label{section:Stokes}
\noindent We begin with the simplest case of \eqref{eqn:ellipticSolns}. When $k=0$, the
solution is a Stokes wave solution of \eqref{eqn:NLS}. The \te{spectral} stability of these
solutions is straightforward to analyze, but the analysis is informative for
understanding the general features of the stability of other solutions. We
choose to work with the Stokes waves in this form to link them with the general
elliptic solutions \eqref{eqn:ellipticSolns}. The Stokes waves are given by
\begin{align}
  \Psi(x,t) &=\sqrt{b} ~e^{\displaystyle ix\sqrt{1-b}}~
  e^{\displaystyle - i(1-3b) t/2},
  \label{eqn:Stokes}
\end{align}
with parameter $b \in (0,1]$.
%
The spectral problem \eqref{eqn:spectralProblem} becomes
\begin{align}
  \lam
  \begin{pmatrix} U \\
    V\end{pmatrix} = \begin{pmatrix}  -\sqrt{1-b} ~\partial_x & -\frac12
    \partial_x^2\\\frac12\partial_x^2 +2 b &
    -\sqrt{1-b}~\partial_x\end{pmatrix}\begin{pmatrix} U \\V\end{pmatrix} =
      \cL_S\begin{pmatrix} U \\ V \end{pmatrix}.\label{eqn:stokesStability}
\end{align}
We consider the constant coefficients of $\cL_S$ as $\pi$-periodic to match
results below for the more general solutions of Section
\ref{section:ellipticSolns}, but the results for the Stokes waves are
independent of this choice of period. Thus the eigenfunctions $(U,V)^T$ of
\eqref{eqn:stokesStability} may be decomposed via a Floquet-Bloch decomposition
\eqref{eqn:floquet-bloch}
\begin{align}
  \begin{pmatrix}
    U(x) \\ V(x)
  \end{pmatrix}  = e^{i\mu x}
  \begin{pmatrix} \hat
    U(x) \\ \hat V(x)
  \end{pmatrix}, \label{eqn:eigenfunctionFloquetDecomp}
\end{align}
where $\hat U,~ \hat V$ have period $\pi$ and $\mu \in [0,2)$. \te{Since
\eqref{eqn:stokesStability} has constant coefficients, it suffices to consider
each Fourier mode $(\hat U_n, \hat V_n)^T$ individually:}
\begin{align}
 \lam \begin{pmatrix} \hat U_n\\\hat V_n\end{pmatrix} =
 \begin{pmatrix} -i\sqrt{1-b}~(\mu + 2n) & \frac12 (\mu +2n)^2\\ 2b-\frac12(\mu +
   2n)^2 & -i\sqrt{1-b}~(\mu + 2n)\end{pmatrix}\begin{pmatrix} \hat U_n \\
 \hat V_n\end{pmatrix} = \hat \L_S^{(n,\mu)} \begin{pmatrix} \hat U_n \\ \hat
   V_n
 \end{pmatrix} \label{eqn:stokesStabilityFS},
\end{align}
where $n\in \Z$.  The eigenvalues of $\hat \L_S^{(n,\mu)}$ are
\begin{align}
  \begin{split}
    \lam_{\pm}^{(n,\mu)} &= \frac{\mu + 2n}{2} \lt(-2i \sqrt{1-b}  \pm \sqrt{4b -
    (\mu + 2n)^2}\rt)
\end{split} .
 \label{eqn:StokesEvals}
\end{align}
These eigenvalues are imaginary if
\begin{align}
 \mu + 2n=0 \quad \text{or} \quad  b\leq (\mu +2n)^2/4. \label{eqn:stokesStabilityCriterion1}
\end{align}
\te{The Stokes wave with amplitude $b$ is spectrally stable with respect to
  bounded perturbations if \eqref{eqn:stokesStabilityCriterion1} holds for all
$n \in \Z$ and $\mu\in[0,2)$.} \te{For a given $b$, there exist $\mu$ and $n$
  such that \eqref{eqn:stokesStabilityCriterion1} is not satisfied.
Consequently, the Stokes waves are not spectrally stable with respect to general
bounded perturbations.} To examine stability with respect to special classes of
perturbations, we consider special values of $\mu$.

Equating $\mu = 0$ corresponds to perturbations with the same period as the
  solution. The \te{spectral} stability criterion \eqref{eqn:stokesStabilityCriterion1} becomes
$n=0$ or $b\leq n^2$ which is satisfied for all $n$, independent of $b$,
consistent with \cite{deconinck2017stability, gallay_haragus_OrbitalStability}.
For $\mu \neq 0$, the tightest bound on $b$ from
\eqref{eqn:stokesStabilityCriterion1} is given by
\begin{align}
  b \leq
  \begin{cases} \mu^2/4, & \mu \in (0,1],\\ (\mu-2)^2/4, & \mu \in [1,2).
    \end{cases}\label{eqn:StokesStabilityCriteria}
\end{align}
With
\begin{align}
  \mu = \frac{2m}{P}, \quad  P\in
  \Z^{+}\label{eqn:muDefn}, \quad m \in \{0,\ldots,P-1\},
\end{align}
the perturbation \eqref{eqn:perturbation} has
$P$ times the period of the Stokes wave. The \te{spectral} stability criterion
\eqref{eqn:StokesStabilityCriteria} becomes
\begin{align}
  b \leq \begin{cases} m^2/P^2, & m \in \Z \cap (0, P/2],\\ (m/P-1)^2, & m \in
    \Z \cap [P/2, P).  \end{cases}\label{eqn:StokesSubharmonicStabilityCriteria}
\end{align}
When $P=1, ~\mu = 0$ for which the spectral stability criterion is always satisfied. When
$P>1$, the bounds on $b$ are tightest when $m = 1$ and when $m=P-1$
respectively. We call the eigenvalues with $\mu(m=1) = \mu_1$ and $\mu(m=P-1) =
\mu_{P-1}$ the critical eigenvalues. In either case we must have
\begin{align}
  b \leq 1/P^2 \label{eqn:TightStokesSubharmonicStabilityCriteria}
\end{align}
for spectral stability of Stokes waves with respect to \teee{$P$-subharmonic
perturbations} (see Figure \ref{fig:parameterSpace}). This result agrees with
\cite[Theorem 9.1]{deconinck2017stability} but is found in a more direct
manner.

Next we examine the process by which solutions transition from a spectrally
stable state to a spectrally unstable state with respect to a fixed $\mu$ as $b$
increases (see Figure \ref{fig:stokes_collision}). For a fixed $P=P_c$, consider
a value of $b$ such that \eqref{eqn:TightStokesSubharmonicStabilityCriteria} is
satisfied with $b<1/P_c^2$, \ie the solution is spectrally stable with respect
to \teee{$P_c$-subharmonic perturbations}. We know from the above that the
instability with respect to \teee{$P_c$-subharmonic perturbations} first arises
when $ b_c = 1/P_c^2$ from the critical eigenvalues with $ \mu_1 = 2/P_c$ and
with $\mu_{P_c-1}= 2(1-1/P_c)$. Defining
\begin{align}
  \lam_c(b) &:= i\frac{2}{P_c}\sqrt{1-b} = 2i \sqrt{
  b_c(1-b)}\label{eqn:Stokes-barlam},
\end{align}
we find
\begin{align}
\IM(\lam_+^{(0,\mu_1)}(b)) <  \IM(\lam_c(b)) < \IM(\lam_-^{(0, \mu_1)}(b)),
\quad \IM(\lam_+^{(-1,\mu_{P_c-1})}(b)) > \IM(\lam_c\cc(b)) >
\IM(\lam_-^{(-1,\mu_{P_c-1})}(b)) \label{eqn:StokesEvalOrdering}:
\end{align}
the critical eigenvalues for $n=0$ and for $n=-1$ are ordered on the imaginary
axis and straddle $\lam_c(b)$ or $\lam_c\cc(b)$.  Increasing $b$ leads to $b=b_c
= 1/P_c^2$ where
\begin{align}
  \lam_+^{(0,\mu_1)} = \lam_-^{(0,\mu_1)} = \lam_c = -\lam_+^{(-1,\mu_{
  P_c-1})} = -\lam_-^{(-1,\mu_{P_c-1})}  \in i \R,
\end{align}
and the critical eigenvalues collide at $\lam_c$ and $\lam_c\cc = -\lam_c$ in
the upper and lower half planes respectively. At the collision,
\begin{align}
  \lam_c(b_c) = 2i \sqrt{b_c(1-b_c)}.
\end{align}
This is the intersection of the top of the figure 8 spectrum and the imaginary
axis in the complex $\lam$ plane \cite[equation (92)]{deconinck2017stability}.
Instability occurs when two critical imaginary eigenvalues collide along the
imaginary axis in a Hamiltonian Hopf bifurcation and enter the right and left
half planes along the figure 8, see Figure~\ref{fig:stokes_collision}.
\begin{figure}
 \includegraphics[width=0.3\textwidth]{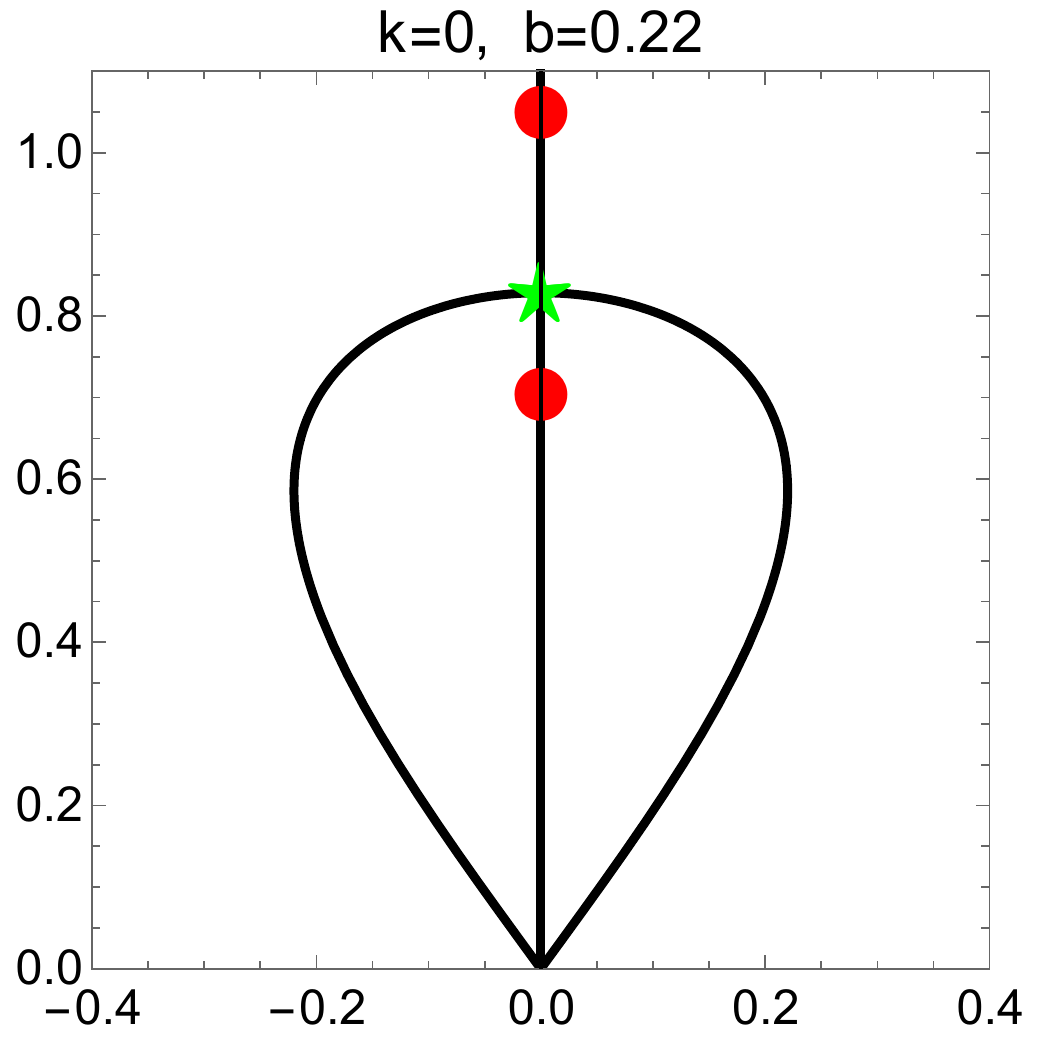}\qquad
 \includegraphics[width=0.3\textwidth]{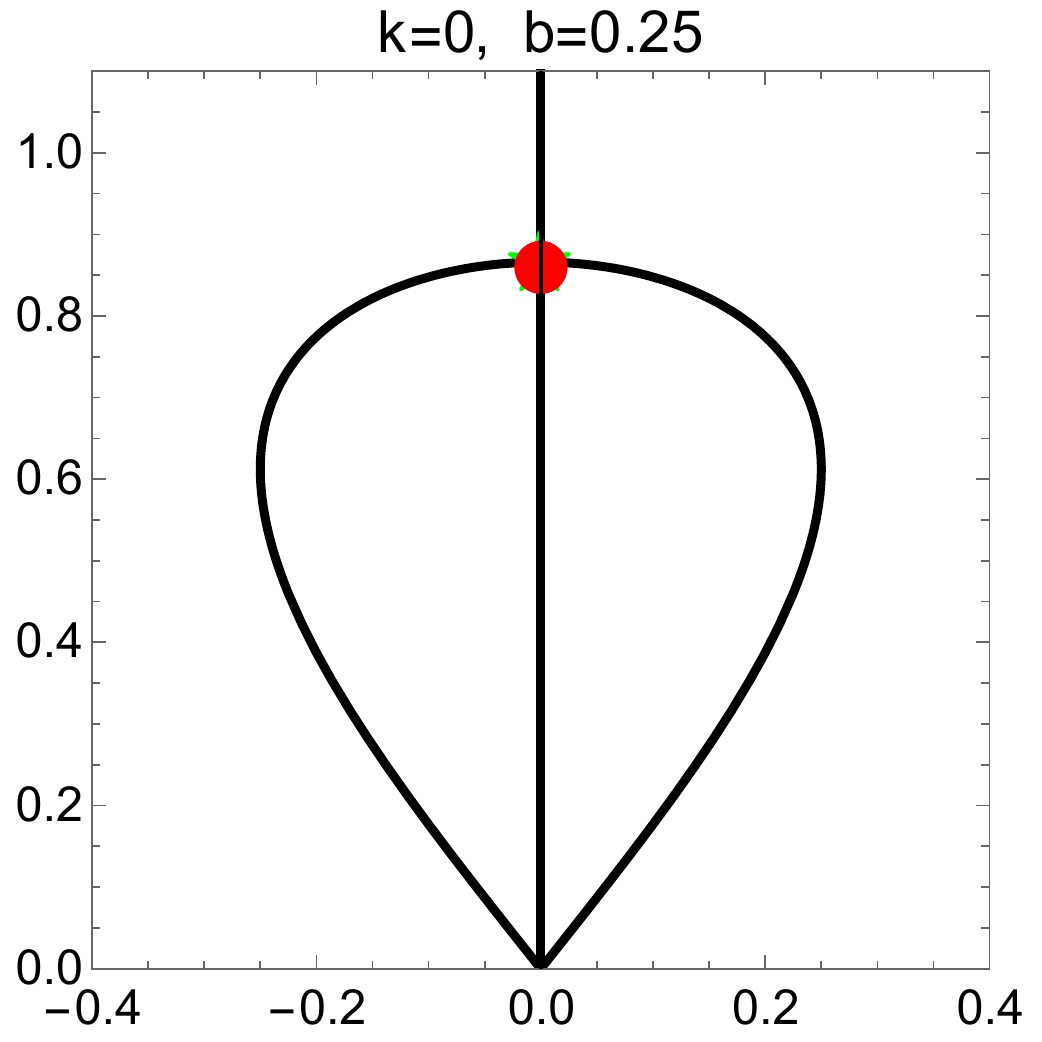}\qquad
 \includegraphics[width=0.3\textwidth]{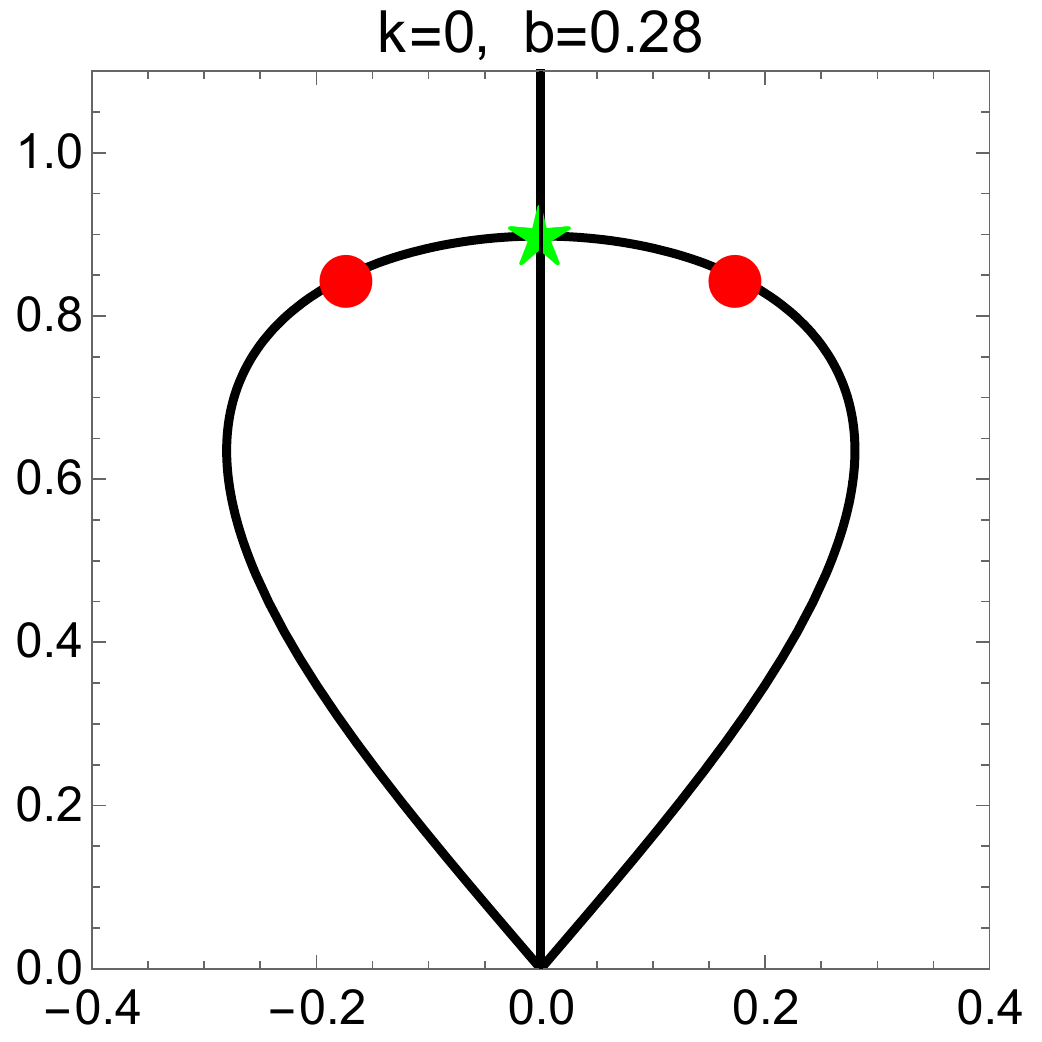}\qquad
  \caption{
  \label{fig:stokes_collision}
  The upper half complex $\lam$ plane, depicting part of the spectrum for Stokes
  waves using \eqref{eqn:StokesEvals} with $b=0.22,~ b=0.25, ~b=0.28$ from left
  to right. Red dots represent eigenvalues with $P=2$ and $n=0$ (using
    \eqref{eqn:muDefn}). The green star at the intersection of the curve and the
  imaginary axis represents $\lam_c$ \eqref{eqn:Stokes-barlam} where the
eigenvalues collide.}
\end{figure}
Other such collisions of eigenvalues occur at the top and bottom of the figure 8
leading to unstable modes as $b$ varies, but the classification of spectral
stability \vs instability is governed by the first unstable modes.

\teee{
  In the rest of Section \ref{section:linearStability}, we generalize these
  Stokes waves results to the elliptic solutions of \eqref{eqn:NLS}. Doing so is
  far more technical, but the main idea remains the same: solutions that are
  spectrally stable with respect to a given subharmonic perturbation become
  unstable with respect to that subharmonic perturbation when two
imaginary eigenvalues collide at the top of the figure 8 spectrum.}


\subsection{The Lax spectrum and the squared-eigenfunction
connection}\label{section:stability_setup}

\teee{The stability of the elliptic solutions is more difficult to analyze than
that of the Stokes waves since $\cL$ \eqref{eqn:spectralProblem} does not have
constant coefficients. To determine the spectrum $\sigma_{\cL}$, we use the
integrability of NLS (see Appendix \ref{section:integrability}). In particular,
we use that \eqref{eqn:stationaryPDE} is obtained by requiring that $\chi_{xt}
= \chi_{tx}$ hold, where
\begin{subequations}
\begin{align}
  &\chi_x  = \begin{pmatrix} -i\zeta & \psi \\ -\psi\cc & i\zeta
    \end{pmatrix}\chi, & &&
  &\chi_t =  \begin{pmatrix} A & B \\ C & -A\end{pmatrix}\chi \label{eqn:chi_x},&\\
 &A = -i\zeta^2 + \frac i2 \abs{\psi}^2 + \frac i2 \omega, \label{eqn:A}&
 &B = \zeta \psi + \frac i2 \psi_x, &
 &C = -\zeta \psi\cc + \frac i2 \psi_x\cc.&
\end{align}\label{eqn:LaxPairRepr}\end{subequations}
Equations \ref{eqn:LaxPairRepr} are known as the Lax pair of the focusing NLS equation.  }

\teee{\subsubsection{Finding the Lax spectrum and the squared eigenfunction
connection}\label{section:LaxSpectrumDef}}
We say that $\zeta \in \sigma_L$ (the Lax spectrum) if $\zeta$ gives rise to a
bounded (for $x\in\R$) eigenfunction of \eqref{eqn:LaxPairRepr}. To determine these
eigenfunctions, we restrict the Lax pair \eqref{eqn:LaxPairRepr} to the elliptic
solutions \eqref{eqn:ellipticSolns} by letting $\psi(x,t) = \phi(x)$.
\teee{Since now \eqref{eqn:LaxPairRepr} are autonomous in $t$, let $\chi(x,t) =
e^{\Omega t}\varphi(x)$. }In order for $\varphi$ to be nontrivial,
\begin{align}
  \Omega^2 = A^2 + BC = -\zeta^4 + \omega\zeta^2 + c\zeta
  -\frac1{16}\lt(4\omega b + 3b^2 + (1-k^2)^2\rt)
  \label{eqn:Omega}.
\end{align}
\te{For $\chi(x,t)$ to be a simultaneous solution of \eqref{eqn:LaxPairRepr}, we require
\begin{align}
  \begin{split}
    \chi(x,t) &= \begin{pmatrix} \chi_1 \\ \chi_2 \end{pmatrix}  =
    e^{\Omega t} \gamma(x) \begin{pmatrix} - B(x;\zeta) \\ A(x;\zeta) - \Omega\end{pmatrix},\\
    \gamma(x) &= \gamma_0 \exp \lt( -\int \mathcal{I} ~\d x \rt),
  \end{split}\label{eqn:gamma(x)}
\end{align}
whenever $\lt\langle \RE(\mathcal{I})\rt\rangle = 0$, \ie
$\RE(\mathcal{I})$ has zero average over one spatial period $T(k)$, and
$\gamma_0$ is a constant. The integrand $\mathcal{I}$ is defined by
\begin{align}
  \begin{split}
    \mathcal{I} &=  \frac{i \zeta B(x;\zeta) + (A(x;\zeta)- \Omega)\phi(x) +
    B_x(x;\zeta)}{B(x;\zeta)} \\
    &= \frac{A_x(x;\zeta) - \phi(x)\cc B(x;\zeta) -
    i\zeta(A(x;\zeta)-\Omega)}{A(x;\zeta)-\Omega}.
  \end{split}\label{eqn:IntegrandI}
\end{align}
}
Two seemingly different definitions for $\mathcal{I}$ are given in
\eqref{eqn:IntegrandI}. The two definitions arise from the fact that
\eqref{eqn:LaxPairRepr} defines two linearly dependent differential equations for
$\gamma(x)$. The two equivalent definitions for $\mathcal{I}$ follow from
$\chi_1$ and $\chi_2$ respectively. The average of $\mathcal{I}$ is computed in
\cite{deconinck2017stability} using the second representation:
\begin{align}
  I(\zeta) = -\int_0^{T(k)} \mathcal{I} ~\d x =
  -2i \zeta \omega_1 + \frac{4i (-c + 4\zeta^3 -2 \zeta \omega -
  4i\zeta \Omega(\zeta))}{\wp'(\alpha)} \lt(\zeta_w(\alpha)\omega_1 -
  \zeta_w(\omega_1)\alpha\rt)\label{eqn:Izeta},
\end{align}
where $\zeta_w$ is the Weierstrass-Zeta function \cite[Chapter 23]{NIST:DLMF},
and $\alpha$ is any solution of
\begin{align}
  \wp(\alpha) = 2i (\Omega(\zeta) + i \zeta^2 -
  i\omega/6)\label{eqn:alpha}.
\end{align}
Note that \eqref{eqn:Izeta} has the opposite sign of \cite[equation
(69)]{deconinck2017stability} \te{in which it is defined inconsistently.} Using
\begin{align}
  \lt(\wp'(\alpha)\rt)^2 = -4\lt( - c + 4\zeta^3 - 2\zeta\omega -
  4i\zeta\Omega(\zeta)\rt)^2,\label{eqn:wpalphasquared}
\end{align}
\eqref{eqn:Izeta} is given by the simpler form
\begin{align}
  I(\zeta) &= -2i \zeta \omega_1 + 2(\zeta_w(\alpha)\omega_1 -
  \zeta_w(\omega_1)\alpha)\Gamma\label{eqn:IzetaGamma},
\end{align}
where
\te{
\begin{align}
  \Gamma = \frac{ 2i\lt(-c + 4\zeta^3 - 2\zeta \omega - 4i \zeta
  \Omega(\zeta)\rt)}{\wp'(\alpha)}.
\end{align}
From \eqref{eqn:wpalphasquared}, $\abs{\Gamma}=1$.}  The condition for $\zeta
\in \sigma_L$ is
\begin{align}
  \zeta \in \sigma_L \Leftrightarrow \RE(I(\zeta)) = 0 \label{eqn:Icond}.
\end{align}
The derivative
\begin{align}
  \DD{I}{\zeta} &= \frac{2E(k) - (1+b-k^2 +
  4\zeta^2)K(k)}{2\Omega(\zeta)}\label{eqn:IDerivative},
\end{align}
\teee{
where
\begin{align}
  E(k) := \int_0^{\pi/2} \sqrt{1-k^2\sin^2(y)}~\d y,
\end{align}
the complete elliptic integral of the second kind \cite[Chapter 19]{NIST:DLMF},
is used for examining $\sigma_L$.} Tangent vectors to the curves constituting
$\sigma_L$ are given by the vector
\begin{align}
  \left(\IM\lt(\DD{I}{\zeta}\rt),
  \RE\lt(\DD{I}{\zeta}\rt)\right)^T,\label{eqn:LaxSpectrumTangents}
\end{align}
in the complex $\zeta$ plane.

When $\zeta \in \sigma_L$, the squared-eigenfunction connection \cite{AKNS74,
deconinck2017stability} gives the spectrum $\lam = 2\Omega(\zeta)$ and the
corresponding eigenfunctions of $\cL$ \eqref{eqn:spectralProblem},
\begin{align}
  \begin{pmatrix}
    U \\ V
  \end{pmatrix}
  &=  \begin{pmatrix} e^{-i\theta(x)}\varphi_1^2 - e^{i\theta(x)} \varphi_2^2\\
    -ie^{-i\theta(x)}\varphi_1^2 - i e^{i\theta(x)}\varphi_2^2
  \end{pmatrix},\label{eqn:squaredEF}
\end{align}
where $(\varphi_1,\varphi_2)^T = e^{-\Omega t}\chi$. \teee{The following theorem
  establishes that the squared-eigenfunction connection can be used to obtain
almost every eigenvalue of $\cL$.}

\begin{theorem}\label{thm:squaredEigCompleteness}
  All but six solutions of \eqref{eqn:spectralProblem} are obtained through the
  squared-eigenfunction connection \te{\eqref{eqn:squaredEF}}. Specifically, all solutions of
  \eqref{eqn:spectralProblem} bounded on the whole real line are obtained
  through the squared-eigenfunction connection, except at $\lam=0$.
\end{theorem}
\begin{proof}
  The proof is similar to the proof of \cite[Theorem~2]{bottman2011elliptic}.
  \te{For a complete proof, see Appendix \ref{Appendix:SQEFCompleteness}}
\end{proof}

Therefore the condition for spectral stability is that $\Omega(\sigma_L)\subset
i\R$.

\begin{remark}
  The explicit eigenfunction representation \eqref{eqn:squaredEF} can be used to
  construct an explicit representation for the Floquet discriminant which is a
  commonly used tool for computing $\sigma_L$ \cite{ablowitz1996computational,
calini2011squared, MR1123280, leeThesis}. The Floquet discriminant for NLS and
other integrable equations is constructed and analyzed in \cite{upsalInPrep}.
\end{remark}

To examine the stability with respect to subharmonic perturbations, we need
$\lam$ in terms of $\mu$. Except for the Stokes waves (Section
\ref{section:Stokes}), \te{we cannot express $\lam$ in terms of $\mu$
explicitly.} \te{Instead, we use an explicit expression for $\mu = \mu(\zeta)$
  and the connection between $\zeta$ and $\lam$  to say something about
  $\lam(\mu)$}.  Equation (112) in
  \cite{deconinck2017stability} gives
\begin{align}
\begin{split}
    e^{iT(k)\mu(\zeta)} &= \exp\lt( -2 \int_0^{T(k)} \frac{(A(x) -
    \Omega)\phi(x) +
  B_x(x) + i\zeta B(x)}{B(x)} \d x\rt) e^{i\theta(T(k))}\\
  &= e^{ 2I(\zeta) + i\theta(T(k))}.
\end{split}\label{eqn:muExponential}
\end{align}
It follows that
\te{
\begin{align}
  M(\zeta) := T(k)\mu(\zeta) = -2i I(\zeta) + \theta(T(k)) + 2\pi n, \quad n\in \Z.
    \label{eqn:muTk}
\end{align}
}
Here, $\theta(T(k))$ is defined to be continuous at $b=k^2$
by
\te{
\begin{equation}
  \theta(T(k)) :=
  \begin{cases}
    \int_0^{T(k)} \frac{c}{R^2(x)}~ \d x, & b>k^2, \\
    \pi, & b=k^2.
  \end{cases}\label{eqn:thetaTk}
\end{equation}
}
For nontrivial-phase solutions, the Weierstrass integral formula
\cite[equation 1037.06]{byrdFriedman} gives
\begin{align}
\begin{split}
  \theta(T(k)) &= \int_0^{2\omega_1} \frac{c}{e_0 - \wp(x; g_2, g_3)} ~\d x =
  \frac{4c}{\wp'(\alpha_0)} \lt( \alpha_0 \zeta_w(\omega_1)- \omega_1
  \zeta_w(\alpha_0)\rt)\\
&= -2i\lt(\alpha_0\zeta_w(\omega_1) - \omega_1 \zeta_w(\alpha_0)\rt),
\end{split}\label{eqn:thetaW}
\end{align}
where
\begin{equation}
  \wp(\alpha_0) = e_0\label{eqn:alpha0} = -\frac{2\omega}{3} = b+e_3,
\end{equation}
and $\wp'(\wp^{-1}(e_0)) = 2i c$ is obtained from
\cite[equation (3.51)]{painleveHandbook}.

\subsubsection{A description of the Lax spectrum}\label{section:LaxSpectrum}
\teee{Since the Lax spectrum is used to determine the stability spectrum, a
complete description of the Lax spectrum is required for our stability analysis.} In
what follows, we use the notation
\begin{subequations}
\label{eqn:zetaRoots}
\begin{align}
    \zeta_1 &= \frac12\lt(\sqrt{1-b} + i(\sqrt{b} - \sqrt{b-k^2})\rt),
    & \zeta_2 &= \frac12\lt(-\sqrt{1-b} + i(\sqrt{b} + \sqrt{b-k^2})\rt), \\
    \zeta_3 &= \frac12\lt(-\sqrt{1-b} - i(\sqrt{b} + \sqrt{b-k^2})\rt),
    & \zeta_4 &= \frac12\lt(\sqrt{1-b} - i(\sqrt{b} - \sqrt{b-k^2})\rt),
\end{align}
\end{subequations}
for the roots of $\Omega^2$ in the first, second, third, and fourth quadrants of
the complex $\zeta$ plane, respectively (for cn and NTP solutions). We refer to
the roots collectively as $\zeta_j$. We rely heavily on
\cite[Lemma~9.2]{deconinck2017stability} \te{which states that $M(\zeta)$
  \eqref{eqn:muTk}} must increase in absolute value along $\sigma_L$ until a
turning point is reached, where $\d I/ \d\zeta=0$. The only turning points occur
at $\zeta = \pm \zeta_c$ \te{where
\begin{align}
      \zeta_c^2 &:= \frac{2E(k) - (1+b-k^2)K(k)}{4K(k)} \label{eqn:zetac}.
\end{align}
}
Since $\zeta_c^2 \in \R, ~\zeta_c$ is real or imaginary depending on the solution
parameters $(k,b)$. We refer to $\zeta_c$ as the solution to $\eqref{eqn:zetac}$
with $\RE(\zeta_c)\geq 0$ and $\IM(\zeta_c)\geq 0$. We primarily use $-\zeta_c$
in the analysis to follow since the branch of spectrum in the left half plane
maps to the outer figure 8 (see Figure \ref{fig:stabilityToInstabilityFFHM})
which corresponds to the dominant instabilities. Further, $\zeta_c=0$ when
$b=B(k)$ where
\begin{align}
  B(k):=\frac{2 E(k) - (1-k^2)K(k)}{K(k)} \label{eqn:Bk}.
\end{align}
For $b>B(k),~ \zeta_c\in i\R\setminus\{0\}$ and for $b<B(k),~\zeta_c \in
\R\setminus\{0\}$. The following lemmas concern the shape of the Lax spectrum and
are important in our analysis of the stability of solutions.
\begin{lemma}\label{lem:muSymmetry}
  The Lax spectrum $\sigma_L$ is symmetric about $\IM\zeta=0$. Further,
  if $\mu(\zeta)$ increases (decreases) in the upper half plane, then
  $\mu(\zeta)$ decreases (increases) at the same rate in the lower half plane
  along $\sigma_L$.
\end{lemma}

\begin{proof}
  Though the proof for the symmetry of $\sigma_L$ comes more directly from the
  spectral problem, we prove it by other means here to setup the proof for the
  second part of the lemma.

  The tangent line to the curve $\RE(I)=0$ is given by
  \eqref{eqn:LaxSpectrumTangents}, where
  \begin{subequations}
  \begin{align}
    \RE\lt(\DD{I}{\zeta}\rt) &= \frac{2 E(k)\Omega_r - K(k)\lt( 8 \zeta_i
    \zeta_r \Omega_i + (1+b-k^2 + 4(\zeta_r^2 -
    \zeta_i^2))\Omega_r\rt)}{2(\Omega_i^2 + \Omega_r^2)}\label{eqn:ReDI},\\
    \IM\lt(\DD{I}{\zeta}\rt) &= \frac{-2E(k) \Omega_i + K(k)\lt(-8
    \zeta_i\zeta_r \Omega_r + (1+b-k^2 + 4(\zeta_r^2 -
    \zeta_i^2))\Omega_i\rt)}{2(\Omega_i^2 + \Omega_r^2)},\label{eqn:ImDI}
  \end{align}
    \label{eqn:ReImDI}
\end{subequations}
and $\Omega_r$ $(\Omega_i)$ and $\zeta_r$ $(\zeta_i)$ are the real
(imaginary) parts of $\Omega$ and $\zeta$ respectively. Since
  \begin{align}
  \begin{split}
    \RE \lt(\Omega^2\rt) &= -\frac{1}{16}\lt( 1 + 3b^2 -2k^2 + k^4 - 16 c \zeta_r + 4b
    \omega + 16(\zeta_i^4 + \zeta_r^4 - \zeta_r^2 \omega + \zeta_i^2(\omega -
    6\zeta_r^2))\rt),\\
    \IM\lt(\Omega^2\rt) &= \zeta_i\lt(-4\zeta_r^3 + 2\omega \zeta_r + c + 4\zeta_i^2
    \zeta_r\rt),
  \end{split}\label{eqn:ReImOmega}
  \end{align}
  only $\IM\lt(\Omega^2\rt)$ changes sign as $\zeta_i \to -\zeta_i$. It follows
  that $\Omega_i \to -\Omega_i$ and $\Omega_r\to\Omega_r$ as $\zeta_i \to
  -\zeta_i$. From \eqref{eqn:LaxSpectrumTangents} and
  \eqref{eqn:ReImDI},
  \begin{align}
    \lt( \IM\lt(\DD{I}{\zeta}\rt), \RE\lt(\DD{I}{\zeta}\rt)\rt) \to
    \lt( -\IM\lt(\DD{I}{\zeta}\rt), \RE\lt(\DD{I}{\zeta}\rt)\rt),
    \qquad \text{as}\qquad \zeta_i \to -\zeta_i.
    \label{eqn:vfsymmetry}
  \end{align}
  Therefore, $\sigma_L$ looks qualitatively the same from $\zeta_j$ to
  $-\zeta_c$ as it does from $\zeta_j\cc$ to $-\zeta_c$.

  We calculate the directional derivative of $\mu(\zeta)$ along $\sigma_L$:
  \begin{align}
    \begin{split}
    \lt(\DD{\mu(\zeta)}{\zeta_r}, \DD{\mu(\zeta)}{\zeta_i}\rt) \cdot
    \lt( \IM \DD{I}{\zeta}, \RE\DD{I}{\zeta}\rt)
    &=
    2\lt( \DD{\IM(I)}{\zeta_r}, \DD{\IM(I)}{\zeta_i}\rt) \cdot
    \lt( \IM \DD{I}{\zeta}, \RE\DD{I}{\zeta}\rt) \\
    &= 2 \lt( \lt(\IM\DD{I}{\zeta} \rt)^2 + \lt(\RE\DD{I}{\zeta}\rt)^2\rt),
    \end{split}
    \label{eqn:muDirectionalDeriv}
  \end{align}
  which is symmetric about $\IM \zeta = 0$.

\end{proof}
\begin{lemma}\label{lem:LaxSpectrumIntersection}
  When $b\leq B(k)$, given in \eqref{eqn:Bk}, the branch of the Lax
  spectrum in the left half plane (right half plane) intersects the real axis at
  $\zeta = -\zeta_c$ ($\zeta = \zeta_c$).
\end{lemma}
\begin{proof}
  Let $\zeta_r\in\R$ and $\eps>0$. Since the vector field
  \eqref{eqn:LaxSpectrumTangents} is continuous across the real $\zeta$ axis,
  and since $\sigma_L$ is vertical at the intersection with the real $\zeta$
  axis by virtue of \eqref{eqn:vfsymmetry}, we must have
 \begin{align}
    \IM\lt(\DD{I}{\zeta}\bigg|_{\zeta =\zeta_r+i\eps}\rt) =
    \IM\lt(\DD{I}{\zeta}\bigg|_{\zeta
    =\zeta_r-i\eps}\rt)\label{eqn:vectorFieldContinuity}, \qquad \text{as }
    \eps\to 0.
  \end{align}
  We notice that
  \begin{align}
    \Omega^2(\zeta_r \pm i\eps) &= \Omega^2(\zeta_r) \pm i\eps\lt(c -4
    \zeta_r^3 + 2\zeta_r\omega\rt) + \0{\eps^2},
  \end{align}
  so that
  \begin{align}
    \Omega(\zeta_r \pm i\eps) &=\Omega(\zeta_r) \pm
    \frac{i\eps}{2\Omega(\zeta_r)}(c-4\zeta_r^3 + 2\zeta_r\omega) + \0{\eps^2} = i\Omega_i +
    \Omega_r + \0{\eps^2},
  \end{align}
  where $\Omega_r = \0{\eps}$ since $\Omega(\zeta_r)\in i\R$. By
  \eqref{eqn:ImDI}, equation \eqref{eqn:vectorFieldContinuity} is only
  satisfied as $\eps\to 0$ if
  \begin{align}
    \zeta_r &= \pm \frac{\sqrt{2 E(k) -(1+b-k^2)K(k)}}{2\sqrt{K(k)}} = \pm
    \zeta_c.
  \end{align}
\end{proof}

\noindent The next lemma details the topology of the Lax spectrum. To our
surprise, there exist few rigorous results describing the Lax spectrum in the
literature even though it has been used in various contexts (see \eg
\cite{ablowitz1996computational, MR1123280, leeThesis, maAndAblowitz}). Some
representative plots of the Lax spectrum are shown in Figure
\ref{fig:LaxSpectrumPics}.

\begin{lemma}\label{lem:laxSpectrumTopology}
The Lax spectrum for the elliptic solutions consists only of the real line and
two bands, each connecting two of the roots of $\Omega$.
\end{lemma}

\begin{proof}
The fact that the entire real line is part of the Lax spectrum is proven in
\cite{deconinck2017stability} but we present a different, simpler proof that
does not rely on integrating $\mathcal{I}$ \eqref{eqn:Izeta}.
If $\zeta \in \R$, the only possibility
for a real contribution to the integral of $\mathcal{I}$ over a period $T(k)$
is through
\begin{align}
    \mathcal{E} := \frac{\phi \cc B}{A - \Omega},
\end{align}
since $A(x)$ is $T(k)$-periodic. Using the definitions for $A,~ B$, and $\phi$,
\begin{align}
    \RE \mathcal{E} &= \frac12 \D{x} \log(R^2  -2\zeta^2 + \omega + 2i\Omega),
\end{align}
which has zero average since $R^2$ is $T(k)$-periodic. It follows that $\R
\subset \sigma_L$. That the roots of $\Omega(\zeta)$ are in the Lax spectrum
follows from the fact that $M(\zeta_j) \in\R$ (Lemma~\ref{lem:0atroots}).
Because the coefficients of $\mathcal{L}$ are periodic, there can exist no
isolated eigenvalues of $\sigma_{\mathcal{L}}$.  It follows that the Lax
spectrum can be continued away from the roots of $\Omega$. In what follows, we
explain the shape of the spectrum emanating from the roots of $\Omega$ and show
that these branches and $\R$ constitute the Lax spectrum.

The operator \eqref{eqn:LaxPairRepr} is a second-order differential operator, so it
has two linearly independent solutions. The solutions obey
\begin{align}
  &\chi_1(x;\zeta) \sim \begin{pmatrix} e^{-i\zeta x} \\ 0\end{pmatrix}, &
  &\chi_2(x;\zeta) \sim \begin{pmatrix} 0 \\ e^{i\zeta x} \end{pmatrix},&
  &\text{as} \qquad \abs{\zeta}\to\infty.&
\end{align}
As $\abs{x}\to\infty$, the above two solutions are bounded if and only if $\zeta
\in \R$. Therefore, $\R$ is the only unbounded component of $\sigma_L$.  We
examine all possibilities for the finite components of $\sigma_L$ in the next
two paragraphs.

Finite components of the spectrum can only terminate when $\d I/\d \zeta \to
\infty$ by the implicit function theorem. This only occurs at the roots of
$\Omega$. A component of the spectrum can only cross another component when $\d
I/\d\zeta = 0$. This only occurs at $\zeta_c$ which is real if the conditions of
Lemma~\ref{lem:LaxSpectrumIntersection} are satisfied and imaginary otherwise.
It follows that the spectral bands emanating from the roots of $\Omega$ must
intersect either the real or imaginary axis. For the dn solutions, this
band lies entirely on the imaginary axis (see Section \ref{section:dn}). Since
there are no other points at which $\d I/\d \zeta = 0$, there can be no other
non-closed curves in the spectrum. However, we must still rule out closed curves
along which it is not necessary that $\d I/\d\zeta = 0$ anywhere.

Since $I$ is an analytic function away from the roots of $\Omega$ and $\zeta =
\infty$, $\RE I$ is a harmonic function of $\zeta$ away from the roots of
$\Omega$, which we will deal with next. Therefore, if the spectrum contained a
closed curve, we would have $\RE I = 0$ on the interior of that closed curve by
the maximum principle for harmonic functions. If this were true, then it must
also be that the directional derivative of $I(\zeta)$ vanishes on the interior
of the region bounded by the closed curve. However $\d I/\d\zeta = 0$ only at
two points which are either on the real or imaginary axis \eqref{eqn:zetac}. It
follows that there are no closed curves in $\sigma_L$ disjoint from the roots of
$\Omega$. If there were a closed curve which was tangent to the roots of
$\Omega$, the above argument would not hold since $\RE I$ is not analytic at the
root. However, such a curve would imply that the origin of $\sigma_\cL$ has
multiplicity greater than 4 (the origin has multiplicity $4$ since the $4$ roots
of $\Omega$ map to the origin). This is not possible since $\cL$ is a
fourth-order differential operator, and such a tangent curve can not exist.
\end{proof}

\begin{remark}
  The above result may also be proven by examining the large-period limit of
  \eqref{eqn:ellipticSolns} which is the soliton solution of \eqref{eqn:NLS}.
  The spectrum of the soliton is well known \cite{kaup1990perturbation}. Using
  the results of \cite{gardner-LongWavelength, sandstede-scheel-longPeriod,
  yang2018convergence}, the spectrum of the periodic solutions with large period
  can be understood. Once the spectrum for solutions with large period is
  understood, the analysis presented in this paper applies and can be extended
  to solutions with smaller period by continuity.

\end{remark}

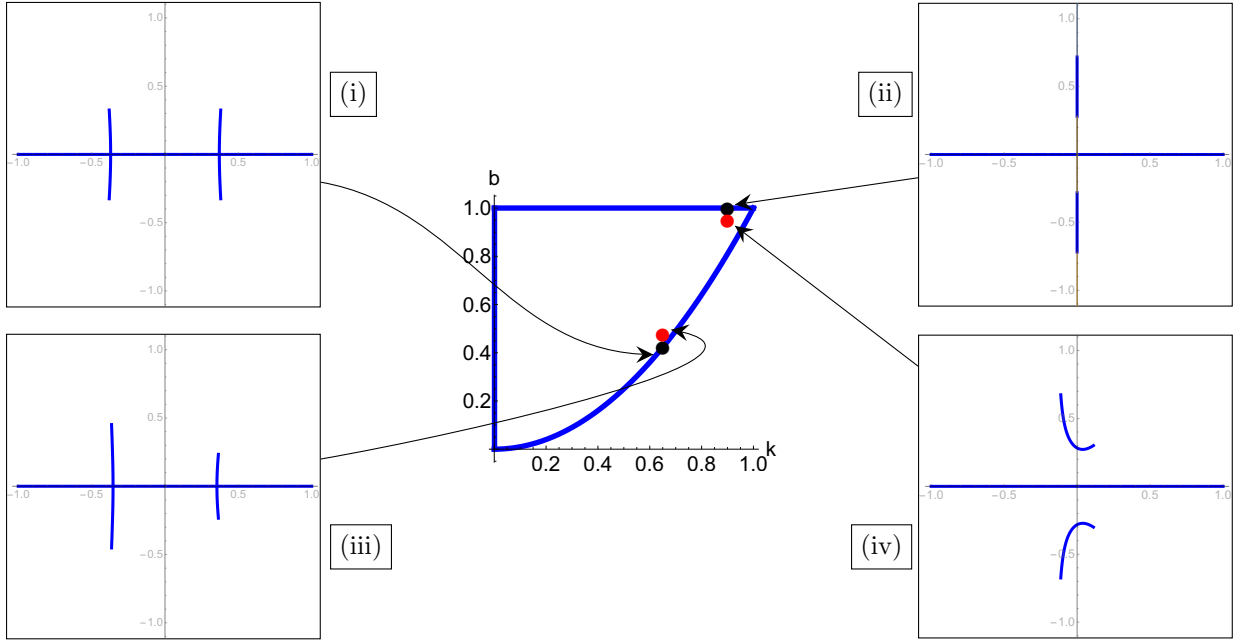
\begin{figure}
    \begin{tikzpicture}
      \node[inner sep=0pt] (paramSpace) at (0,0)
      {\includegraphics[width=0.25\textwidth]{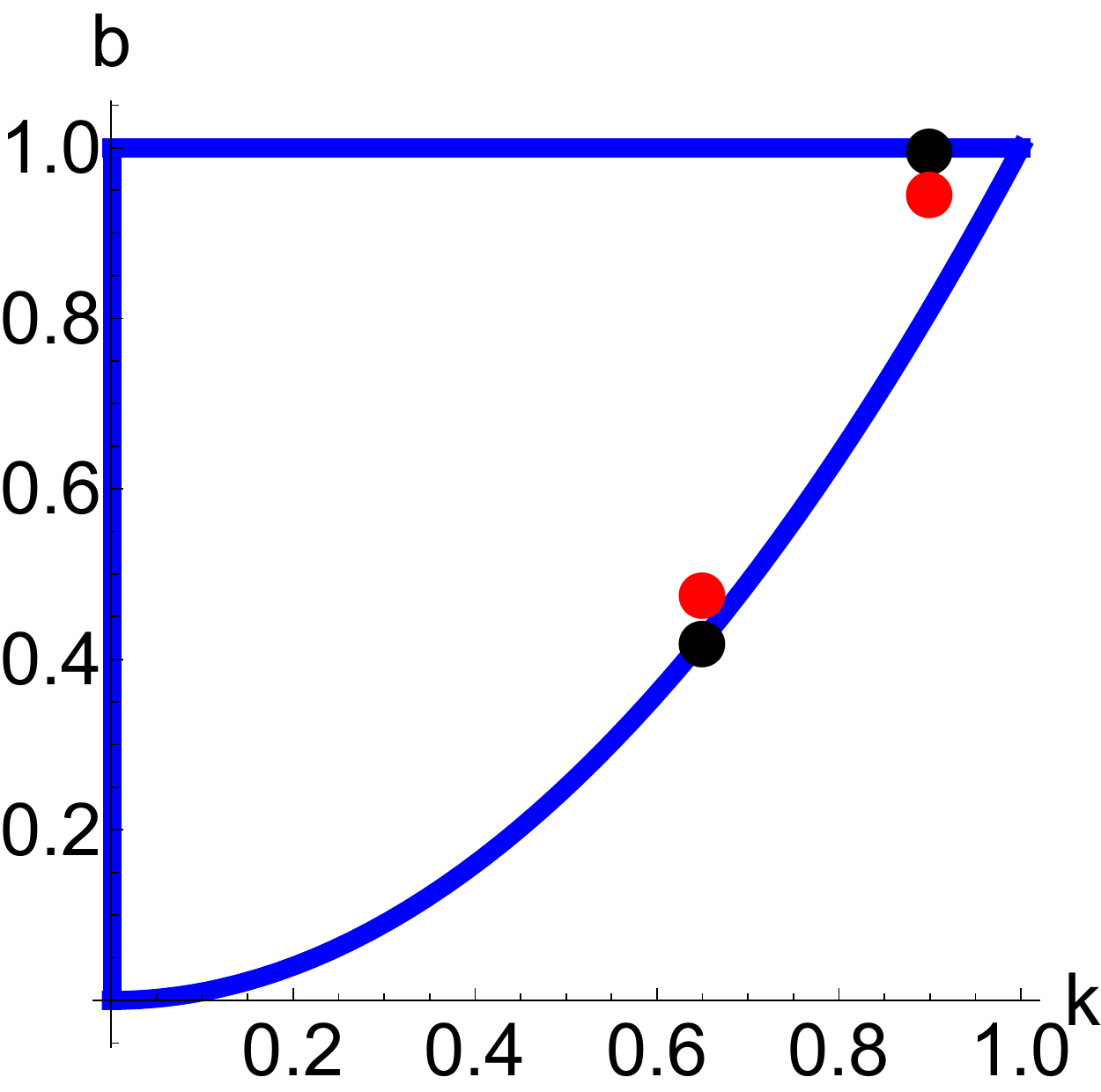}};
      \node[draw, inner sep=0pt] (dnSoln) at (6,2.2)
      {\includegraphics[width=0.25\textwidth]{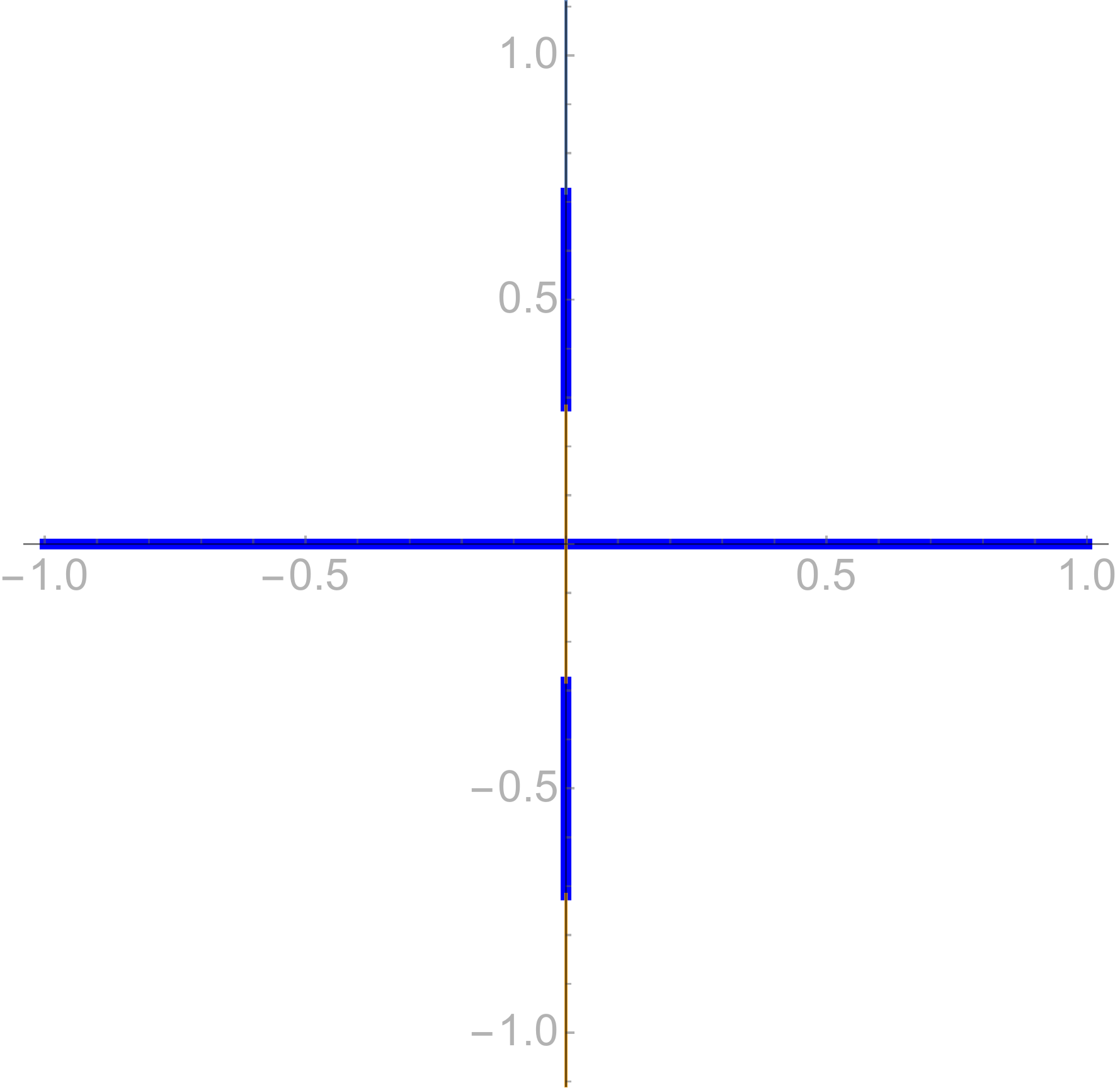}};
      \node[draw, inner sep=0pt] (dnSoln-perturb) at (6,-2.2)
      {\includegraphics[width=0.25\textwidth]{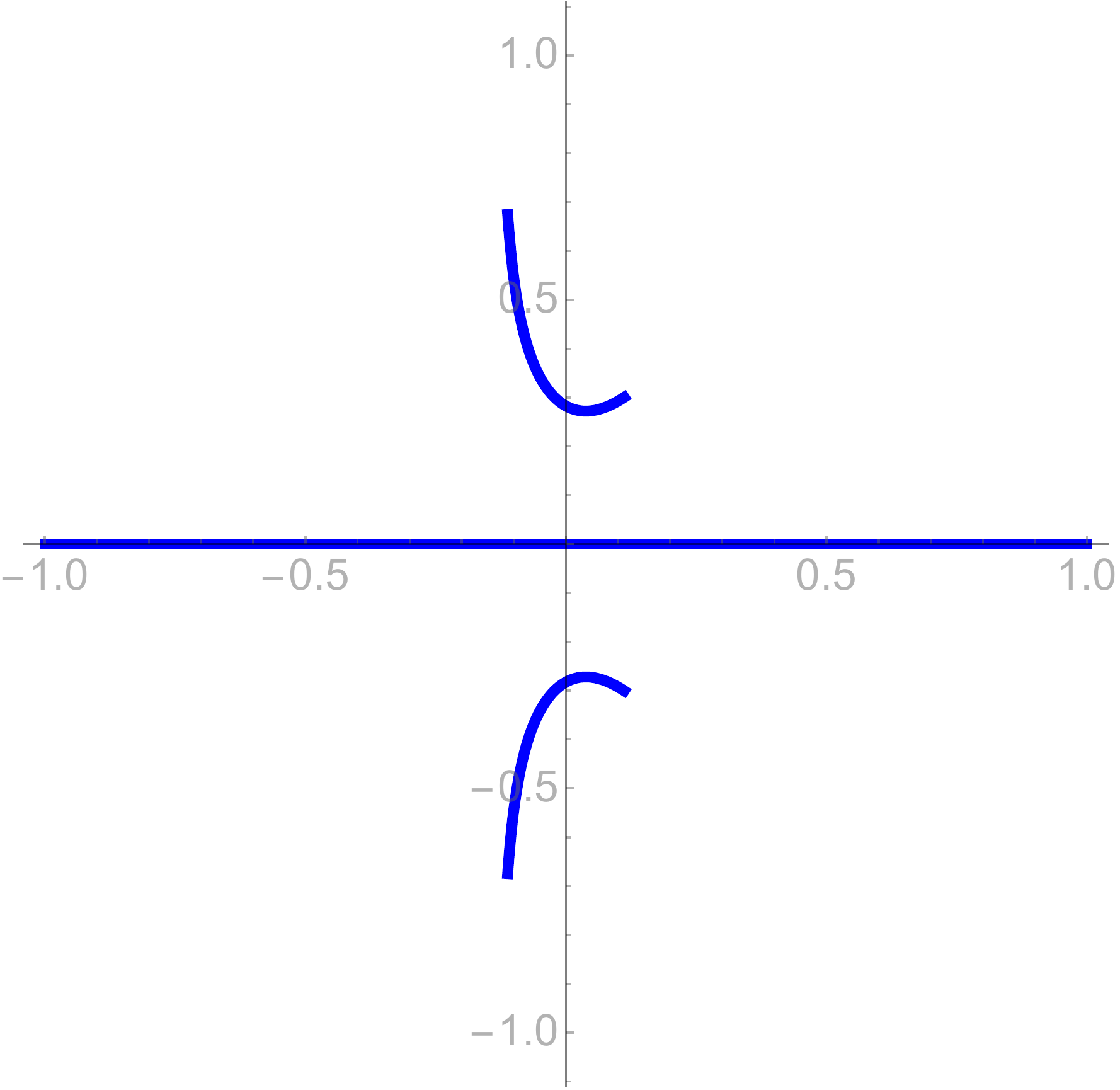}};
      \node[draw, inner sep=0pt] (cnSoln) at (-6,2.2)
      {\includegraphics[width=0.25\textwidth]{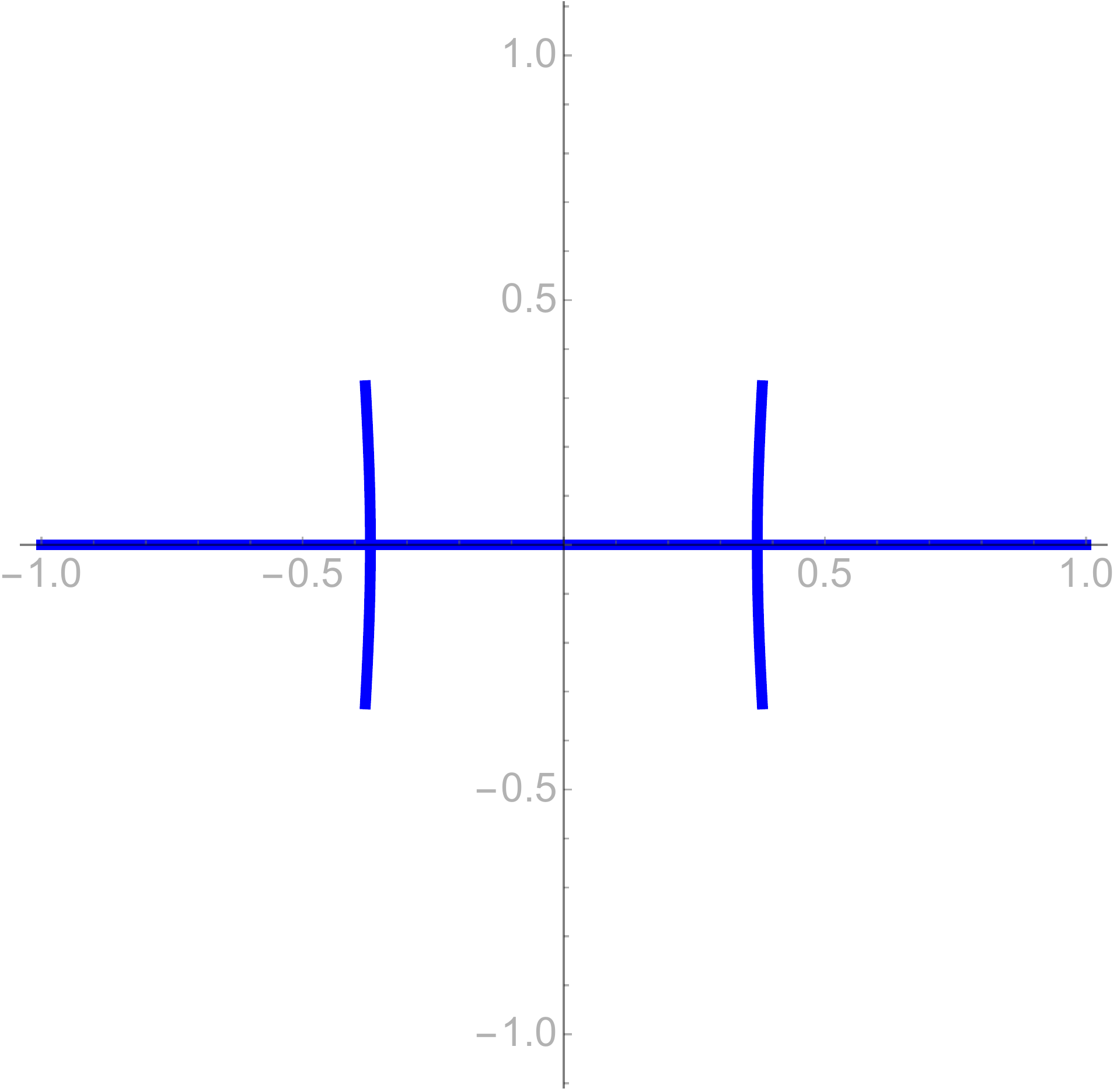}};
      \node[draw, inner sep=0pt] (cnSoln-perturb) at (-6,-2.2)
      {\includegraphics[width=0.25\textwidth]{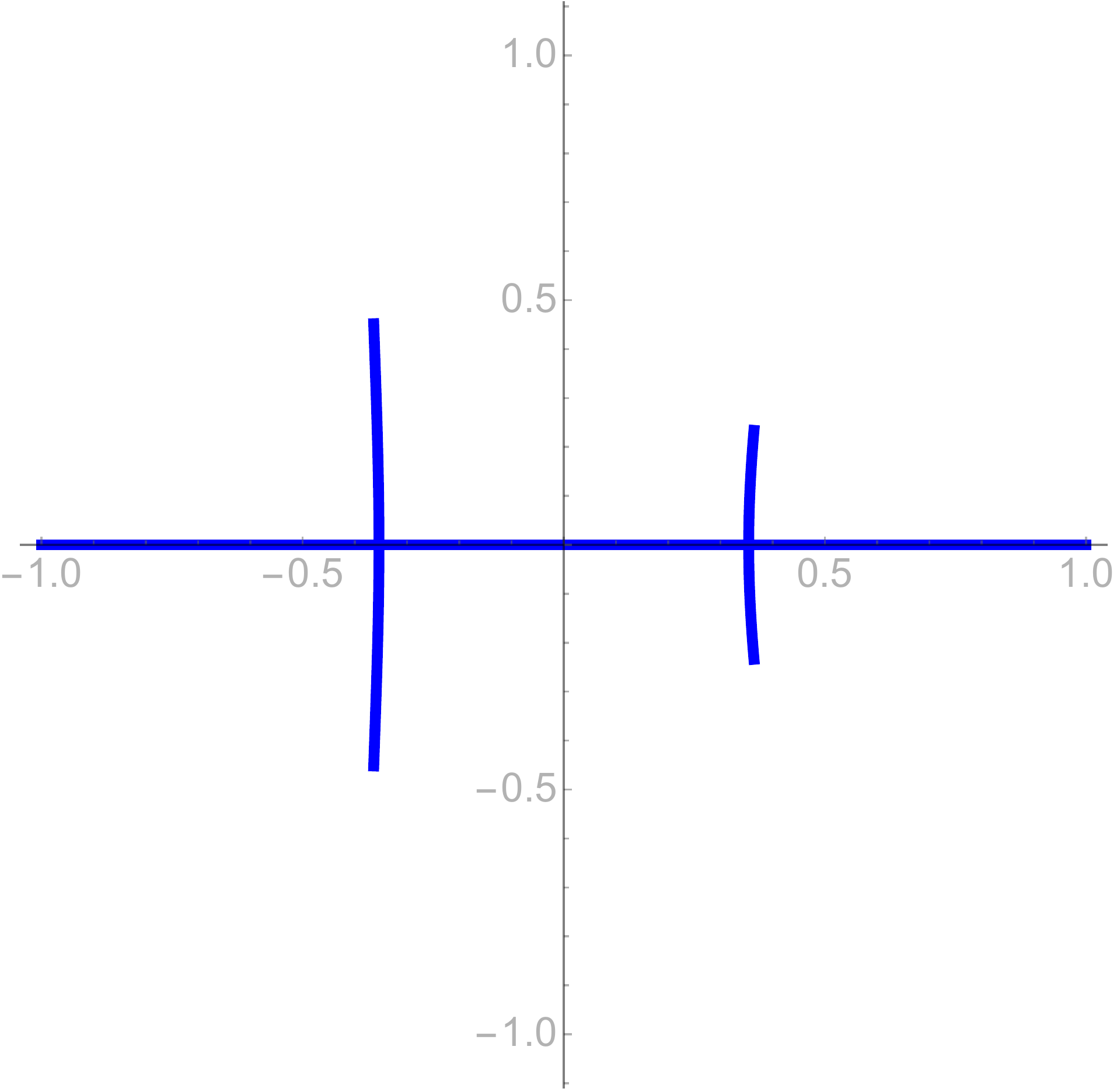}};
      \node (dnAnchor) at (1.4,1.52) {};
      \draw[decoration={markings,mark=at position 1 with {\arrow[scale=2,
      >=stealth]{>}}}, postaction={decorate}] (dnSoln) -- (dnAnchor) ;
      \node (dn-perturb-Anchor) at (1.4,1.35) {};
      \draw[decoration={markings,mark=at position 1 with {\arrow[scale=2,
      >=stealth]{>}}}, postaction={decorate}] (dnSoln-perturb) --
      (dn-perturb-Anchor) ;
      \node (cnAnchor) at (0.57,-0.45) {};
      \draw[decoration={markings,mark=at position 1 with {\arrow[scale=2,
      >=stealth]{>}}}, postaction={decorate}] (cnSoln) to [out=-10, in=180]
      (cnAnchor) ;
      \node (cn-perturb-Anchor) at (0.56,-0.1) {};
      \draw[decoration={markings,mark=at position 1 with {\arrow[scale=2,
      >=stealth]{>}}}, postaction={decorate}] (cnSoln-perturb) to
      [out=10,in=-10] (cn-perturb-Anchor) ;
      \node[draw] at (-3.5,3) {(i)};
      \node[draw] at (3.5,3) {(ii)};
      \node[draw] at (-3.4,-3) {(iii)};
      \node[draw] at (3.45,-3) {(iv)};
    \end{tikzpicture}
  \caption{\label{fig:LaxSpectrumPics} Plots of the Lax spectrum. $\RE\zeta$ \vs
  $\IM\zeta$ for $\zeta \in \sigma_L$. Plots (i) and (ii) are for the cn and dn
  solutions respectively. Plots (iii) and (iv) are for nontrivial-phase
  solutions, where the symmetry in all quadrants is broken. Red dots indicate
  nontrivial-phase solutions, which are plotted in the lower panels. Parameters
  are chosen close together to contrast nearby solutions of trivial and
  nontrivial phase.}
\end{figure}

\begin{lemma}\label{lem:NoDoublePointsSpine}
  $0\leq M(\zeta) <2\pi$ for $\zeta \in \sigma_L\setminus\R$ \te{with equality
  only at the end of the bands, when $\Omega(\zeta)=0$.}
\end{lemma}
\begin{proof}
  See Appendix \ref{Appendix:NoDoublePointsProof}.
\end{proof}
\begin{remark}
  We note that Lemma~\ref{lem:NoDoublePointsSpine} can be rephrased in the
  language of the Floquet discriminant approach
\cite{ablowitz1996computational, calini2011squared, MR1123280, leeThesis} as
the nonexistence of \textit{periodic eigenvalues} (those with
$M(\zeta)=0\mod{2\pi})$ on the
  interior of the complex bands of spectra for the elliptic solutions. Before
  this result, three things were known about the existence of periodic
  eigenvalues on the complex bands: (i) The number of periodic eigenvalues on
  the complex bands was known to have an explicit bound \cite{leeThesis}; (ii)
  for the symmetric solutions (our cn and dn solutions), the number of periodic
  eigenvalues is zero \cite{calini2011squared}; and (iii) the nonexistence of
  periodic eigenvalues on the complex band had been verified numerically
  \cite{calini2011squared, lafortunePersonal}.
  Lemma~\ref{lem:NoDoublePointsSpine} settles this question: there are no
  periodic eigenvalues on the complex bands of the Lax spectrum for the elliptic
  NLS solutions.
\end{remark}

\vspace{1em}
\teee{\subsection{Spectral stability of the elliptic
solutions}\label{section:EllipticSpectralStability}}

Results about spectral stability with respect to subharmonic perturbations are
found in \cite[Section 9]{deconinck2017stability}. There, sufficient conditions
for stability with respect to subharmonic perturbations are found in Theorems
9.1, 9.3, 9.4, 9.5 and 9.6 for spectra with different topology. In this section
we present these known sufficient conditions for spectral stability while
providing more detailed proofs. For some choices of parameters we show that the
sufficient condition is necessary and comment on progress made towards showing
that this condition is necessary for the entire parameter space in Appendix
\ref{Appendix}.

We begin by showing that $\Omega: \R\cap\sigma_L \mapsto \sigma_\cL \cap i\R$,
and therefore the real line of the Lax spectrum always maps to stable modes.
Showing that these (and the roots of $\Omega$) are the only parts of the Lax
spectrum mapping to stable modes is an important challenge (see Appendix
\ref{Appendix}).

\begin{lemma}\label{lem:zeta_real}
    If $\zeta \in \R$, then $\Omega(\zeta) \in i\R$.
\end{lemma}
\begin{proof}
  If $\zeta \in \R$, then the \teee{matrix defining the $t$-evolution in
    \eqref{eqn:LaxPairRepr}} is skew-adjoint, and separation of variables yields
    imaginary $\Omega$.
\end{proof}
\begin{remark}
  \teee{The above result} is proven in \cite{deconinck2017stability}. We
present the proof above because it is significantly simpler and is extendable
to other stationary solutions of the AKNS hierarchy. Work is currently in
progress to extend this and other arguments in this paper to other equations,
both in the AKNS hierarchy and not \cite{upsalInPrep}.
\end{remark}

\subsubsection{Trivial-phase solutions, $b=1$ (dn solutions) or $b=k^2$ (cn solutions)}
The trivial-phase solutions have $c=0$ so that
\begin{align}
  \Omega^2(\zeta) &= -\zeta^4 + \omega \zeta^2 - \frac1{16}\lt(4\omega b + 3b^2 +
  (1-k^2)^2\rt),
\end{align}
and $\Omega^2(\zeta) = \Omega^2(-\zeta)$. Since $\Omega^2(i\R)\subset \R$,
$\lam(\zeta)$ is real or imaginary for $\zeta \in i\R$. Along with Lemmas~
\ref{lem:muSymmetry} and \ref{lem:laxSpectrumTopology}, this implies that
trivial-phase solutions have symmetric Lax spectrum across both the real and
imaginary axes (see Figure \ref{fig:LaxSpectrumPics}).

\paragraph{Solutions of dn-type, $b=1$}\label{section:dn}

When $b=1$, $\zeta_j \in i\R$ \eqref{eqn:zetaRoots} and
\begin{align}
  \IM(\zeta_2) > \IM(\zeta_1) > 0 > \IM(\zeta_4) > \IM(\zeta_3),
\end{align}
with $\zeta_2 = -\zeta_3$ and $\zeta_1 = -\zeta_4$. The following lemmas are
needed. The proofs are found in Appendix \ref{Appendix:proofs}.

\begin{lemma}\label{lem:M0dn}
  $M(\zeta_j) = T(k)\mu(\zeta_j) = 0 \mod 2\pi$ for the dn solutions.
\end{lemma}

\begin{lemma}\label{lem:dnRoots}
  Let $\zeta \in i\R$ with either $\abs{\IM(\zeta)}\geq\IM(\zeta_2)$ or
  $\abs{\IM(\zeta)}\leq
 \IM(\zeta_1)$. Then $\Omega(\zeta)\in i\R$.
\end{lemma}

\noindent The above lemmas allow us to find necessary and sufficient conditions
on the spectral stability of dn solutions.

\begin{theorem}
  The dn solutions ($b=1$) are \te{spectrally} stable \te{with respect to
  perturbations of the same period as the underlying solution and no other
subharmonic perturbations.}
\end{theorem}

\begin{proof}
  Lemmas~\ref{lem:laxSpectrumTopology}, \ref{lem:dnRoots}, and the tangent
  vectors \eqref{eqn:LaxSpectrumTangents} show that the complex bands of the Lax
  spectrum are confined to the imaginary axis between the roots of $\Omega$.
  Using Lemmas~\ref{lem:NoDoublePointsSpine} and \ref{lem:M0dn} and the fact
  that $M(\zeta)$ must increase in absolute value between the roots of
  $\Omega(\zeta)$, $M(\zeta) \in [0,2\pi]$ on the bands of the Lax spectrum.
  Equality is attained only at the roots $\zeta_j$. By Lemma~\ref{lem:dnRoots},
  $\Omega(\zeta)\in \R$ on the interior of the bands, so the eigenvalues are
  unstable.  Since $\zeta \in \R$ only maps to stable modes
  (Lemma~\ref{lem:zeta_real}), \te{spectral} stability only exists for $T(k)\mu = 0$, which is
  what we wished to show.
\end{proof}

\paragraph{Solutions of cn-type, $b=k^2$}
When $b=k^2$, the inequality
\begin{align}
  \Omega^2(i\xi) = -\xi^4 + \frac12(2k^2-1)\xi^2-1/16<0,
\end{align}
is satisfied for all $\xi \in \R$. We need the following lemmas whose proofs can
be found in Appendix \ref{Appendix:proofs}.

\begin{lemma}\label{lem:cnMuOnImaginaryAxis}
  For cn solutions, when $\zeta \in i\R$, $M(\zeta) = \pi \mod{2\pi}$.
\end{lemma}

\begin{lemma}\label{lem:M0cn}
  $M(\zeta_j) = T(k) \mu(\zeta_j) = 0 \mod 2\pi$ for the cn solutions.
\end{lemma}

\begin{lemma}\label{lem:noCrossingscn}
  For $b=k^2$ and $\zeta\in\sigma_L
  \setminus(\{\zeta_1,\zeta_2,\zeta_3,\zeta_4\}\cup \R \cup i\R$), we have that
  $\Omega(\zeta)\notin i\R$.
\end{lemma}

\noindent The above lemmas allow us to find necessary and sufficient conditions
on the spectral stability of cn solutions.

\begin{theorem}\label{thm:cnStabilityResult}
  If $k> k\cc \approx 0.9089$ where $k\cc$ is the unique root of $2E(k) - K(k)$
  for $k\in [0,1)$, then solutions of cn-type ($b=k^2)$, are \te{spectrally
  stable with respect to coperiodic and 2-subharmonic perturbations, but no
other subharmonic perturbations}. If instead $k\leq k\cc$, then solutions are
\te{spectrally} stable with respect to perturbations of period $QT(k)$ for all
$Q\in \N$ with $Q\leq P \in \N$ if and only if
  \begin{align}
    M(-\zeta_c)\leq \frac{2\pi}{P},
  \end{align}
  defined in the $2\pi$-interval in which $M(\zeta_j)=0$.
\end{theorem}

\begin{proof}
  First choose a solution by fixing $k$. Then choose a $P$-subharmonic
  perturbation. If $k>k\cc$, then $2E(k) - K(k)<0$ so that $b>B(k)$ and $\zeta_c
  \in i\R$ (\eqref{eqn:Bk} when $b=k^2$). If $k\leq k\cc$, $\zeta_c\in\R$.
  Consider the band of the spectrum with endpoint $\zeta_2$ at which
  $M(\zeta_2)=0$ (Lemma \ref{lem:M0cn}). If $\zeta_c\in i\R$, this band
  intersects the imaginary axis at $\hat \zeta \in i\R$, otherwise it intersects
  the real axis at $-\zeta_c\in \R$.

  Let $S$ represent the band connecting $\zeta_2$ to $\hat \zeta$ when
  $\zeta_c\in i\R$. When $\zeta_c \in i\R$, $\abs{\RE(\lam)}>0$ on $S$ (Lemma
  \ref{lem:noCrossingscn}) so every $T(k)\mu$ value on $S$ corresponds to an
  unstable eigenvalue. Since $\mu \neq 0 \mod{2\pi}$ on $S$
  (Lemma~\ref{lem:NoDoublePointsSpine}), $M(\zeta)$ is increasing from $\zeta_2$
  to $\hat \zeta$ \cite[Lemma~9.2]{deconinck2017stability}, and $\partial S =
  \{0,\pi\}$ (Lemmas~ \ref{lem:cnMuOnImaginaryAxis} and \ref{lem:M0cn}),
  $M(\zeta) \in (0,\pi)$ on the interior of $S$. Therefore every
  $T(k)\mu\in(0,\pi)$ corresponds to an unstable eigenvalue. By the symmetry of
  the Lax spectrum in each quadrant, the analysis beginning at any of the roots
  $\zeta_j$ gives the same result, except perhaps with $(0,\pi)$ replaced with
  $(\pi,2\pi)$, which yields the same stability results. Since
  $\RE(\lam(\zeta))=0$ only at $T(k)\mu = 0,$ or $T(k)\mu = \pi$, if
  $2E(k)-K(k)<0$, the cn solutions are \te{spectrally stable with respect to
  coperiodic and 2-subharmonic perturbations, but no other subharmonic
perturbations.}

  If $2E(k)-K(k)\geq 0$, the band emanating from $\zeta_2$ intersects the real
  axis at $-\zeta_c$ (Lemma~\ref{lem:LaxSpectrumIntersection}). Then
  $M(\zeta)\in (0, T(k)\mu(-\zeta_c))$ along the interior of this band and
  $M(\zeta) = 0$ and $M(\zeta) = T(k)\mu(-\zeta_c)$ at the respective endpoints
  (Lemma \ref{lem:M0cn}).  Since $\abs{\RE(\lam)}>0$ on the interior of this
  band (Lemma \ref{lem:noCrossingscn}), every $T(k)\mu$ value along this band
  corresponds to an unstable eigenvalue. By Lemma~\ref{lem:NoDoublePointsSpine},
  $M(-\zeta_c)<2\pi$.  Therefore, in order to have spectral stability with
  respect to $P$-subharmonic perturbations, it must be that $M(-\zeta_c)$ is at
  least as small as the smallest nonzero $\mu$ value obtained in
  $\eqref{eqn:muDefn}$ for our $P$. The smallest nonzero $\mu$ value corresponds
  to $m=1$ or $m=P-1$, so if
  \begin{align}
    M(-\zeta_c)\leq \frac{2\pi}{P}, \label{eqn:cnBound}
  \end{align}
  then solutions are \te{spectrally} stable with respect to perturbations of
  period $PT(k)$.  Since the Lax spectrum is symmetric about the real and
  imaginary axes for the cn solutions (see Figure \ref{fig:LaxSpectrumPics}(i)),
  the same bound is found by starting the analysis at each $\zeta_j$. Since the
  preimage of all eigenvalues with $\RE(\Omega(\zeta))>0$ is the interior of the
  bands (Lemma \ref{lem:noCrossingscn}), \eqref{eqn:cnBound} is also a necessary
  condition for \te{spectral} stability. Since the bound holds for each $Q\leq P,~Q\in \N$,
  \te{spectral} stability with respect to $P$-subharmonic perturbations also
  implies \te{spectral} stability
  with respect to $Q$-subharmonic perturbations.
\end{proof}

\begin{remark}
  The calculations throughout this paper use the period of the modulus of the
  solution, $T(k) = 2K(k)$. However, the cn solution itself (not its modulus) is
  periodic with period $4K(k)$. When taking this into account, $I(\zeta)$ gets
  replaced by $2I(\zeta)$, and
  \begin{align}
    T(k)\mu(\zeta) &= 4i I(\zeta) + 2\pi n.\label{eqn:cnTrueM}
  \end{align}
  Using \eqref{eqn:cnTrueM} for $M(\zeta)$, Theorem \ref{thm:cnStabilityResult}
  can be updated to cover subharmonic perturbations with respect to the period
  $4K(k)$ of the cn solutions. We find that when $2E(k)-K(k)<0$, the solutions
  are \te{spectrally} stable with respect to perturbations of period $4K(k)$. The bound
  \eqref{eqn:cnBound} may also be updated using \eqref{eqn:cnTrueM} and upon
  letting $T(k) = 4K(k)$. In particular, we recover the cn solution stability
  results found in \cite{gustafson2017stability, ivey2008spectral}.
\end{remark}

\subsubsection{Nontrivial-phase solutions}

\te{For the nontrivial-phase solutions, $c\neq0$ and $\Omega$ is defined by
  \eqref{eqn:Omega}.  }The statement for the stability of nontrivial-phase
solutions is very similar to that for the stability of cn solutions. We begin
with a lemma whose proof can be found in Appendix \ref{Appendix:proofs}.

\begin{lemma}\label{lem:0atroots}
  $M_j := M(\zeta_j) = T(k)\mu(\zeta_j)= 0 \mod 2\pi$ for each root
  $\{\zeta_j\}_{j=1}^4$ of $\Omega(\zeta)$.
\end{lemma}

\noindent With this lemma, the following sufficient condition for spectral
stability of nontrivial-phase solutions holds.

\begin{theorem}\label{thm:mainNTPStabilityResult}
  Consider a solution with parameters $k$ and $b \leq B(k)$
  \eqref{eqn:Bk}.  The solution is \te{spectrally} stable with respect to
  perturbations of period $QT(k)$ for all $Q\in\N$, $Q\leq P\in \N$ if
  \begin{align}
    M(-\zeta_c)\leq\frac{2\pi}{P},\label{eqn:mainStabilityBound}
  \end{align}
  defined in the $2\pi$ interval in which $M(\zeta_j)=0$.
\end{theorem}

\begin{proof}

    The proof here, much like the statement of the theorem, is similar to the
    proof of Theorem~\ref{thm:cnStabilityResult}.

    Choose a solution by fixing $k$ and $b\leq B(k)$ so that $\zeta_c$ is real.
    Choose a $P$-subharmonic perturbation.  Consider the band of the spectrum
    with endpoint $\zeta_2$ (see Figure \ref{fig:LaxSpectrumPics} (iii, iv)) ,
    at which $M(\zeta_2)=0$ (Lemma~\ref{lem:0atroots}), and which intersects the
    real line at $-\zeta_c$ (Lemmas~\ref{lem:laxSpectrumTopology} and
    \ref{lem:zeta_real}). Since $M(\zeta)$ is increasing along the band (Lemma
    \ref{lem:muSymmetry}), $ 0 <M(\zeta)<T(k)\mu(-\zeta_c)$ along the interior
    of the band with $M(\zeta) = 0$ and $M(\zeta) = T(k)\mu(-\zeta_c)<2\pi$
    (Lemma~\ref{lem:NoDoublePointsSpine}) at the respective endpoints.

    Since the tangent lines of $\sigma_L$ are nonvertical at the origin for
    $b<B(k)$ and $\abs{\RE(\lam(-\zeta_c \pm i\eps))}>0$
    \cite{deconinck2017stability}, there exist $\zeta$ on the bands in a
    neighborhood of $-\zeta_c$ and a neighborhood of $\zeta_2$ which correspond
    to eigenvalues $\lam$ with $\lam_r>0$, \ie unstable eigenvalues.  Since
    there exist unstable eigenvalues on this band, in order to have spectral
    stability with respect to $P$-subharmonic perturbations, it must be that
    $M(-\zeta_c)$ is at least as small as the smallest nonzero $\mu$ obtained in
    $\eqref{eqn:muDefn}$ for our $P$. The smallest nonzero $\mu$ value
    corresponds to $m=1$ or $m=P-1$, so if
    \begin{align}
      M(-\zeta_c)\leq \frac{2\pi}{P},
    \end{align}
    then solutions are \te{spectrally} stable with respect to perturbations of
    period $PT(k)$.

    By Lemma~\ref{lem:muSymmetry}, the same bound is found for the starting
    point $\zeta_3$. Starting at $\zeta_1$ or $\zeta_4$ gives the bound
    \begin{align}
      M(\zeta_c)\leq \frac{2\pi}{P}.
    \end{align}
    However, since
    \begin{align}
      M(-\zeta_c)> M(\zeta_c),
    \end{align}
    as shown in \cite{deconinck2017stability}, the tighter bound is found with
    $M(-\zeta_c)$. This is the sufficient condition for \te{spectral} stability.
    As for the cn case, if the bound is satisfied for $P$, then it is also
    satisfied for all $Q\leq P$.
\end{proof}

\begin{remark}
  Determining whether or not the bound \eqref{eqn:mainStabilityBound} is also a
  necessary condition for \te{spectral} stability is a significant challenge.
  Work in this direction is presented in Appendix \ref{Appendix:real}.
\end{remark}

\begin{remark}\label{remark:muSymmetry}
  We note that Lemma~\ref{lem:muSymmetry} implies that near $-\zeta_c\in\R$, two
  eigenvalues with the same $\abs{T(k)\mu}$ value are found equidistant from
  $-\zeta_c$ along the band above and below the real axis. Since two eigenvalues
  with the same $\abs{T(k)\mu} \mod{2\pi}$ value represent the same perturbation
  of period $PT(k)$, the eigenvalues associated with a perturbation of period
  $PT(k)$ straddle $-\zeta_c$ on either of the arcs and come together or
  separate as the solution parameters vary, see Figure \ref{fig:muSymmetry}.
\end{remark}

\begin{theorem}\label{thm:imaginaryNTPStabilityResult}
  If $b> B(k)$ \eqref{eqn:Bk}, solutions are \te{spectrally} stable with respect
  to coperiodic perturbations. Additionally, they can be \te{spectrally} stable
  with respect to perturbations of twice the period, but they are not stable
  with respect to any other \te{subharmonic} perturbations.
\end{theorem}
\begin{proof}
  See Appendix \ref{Appendix:imaginary}.
\end{proof}

\begin{remark}\label{remark:StabilityImaginaryZetac}
  Numerical evidence suggests that when $b>B(k)$, NTP solutions are spectrally
  stable \te{with respect to coperiodic perturbations and no other subharmonic
  perturbations}. However, there are some parameter values for which the
  stability spectrum intersects the imaginary axis at a point. We cannot rule
  out the possibility of this point corresponding to 2-subharmonic
  perturbations. For the cn solutions, this intersection point corresponds to
  $M(\zeta) = \pi$, which gives rise to \te{spectral} stability with respect to 2-subharmonic
  perturbations. Because of this, a cn solution and a NTP solution with $b>B(k)$
  can be arbitrarily close to each other but have different stability
  properties. One way to rule out this spurious stability for NTP solutions with
  $b>B(k)$ is to show that the point $M(\zeta) = \pi$, which we know occurs
  exactly once on the band of Lax spectrum in the upper half plane, remains in
  the left half plane (see Lemma \ref{lem:RHSCrossings}) for all parameter
  values.
\end{remark}

\begin{figure}
  \includegraphics[width=0.4\textwidth]{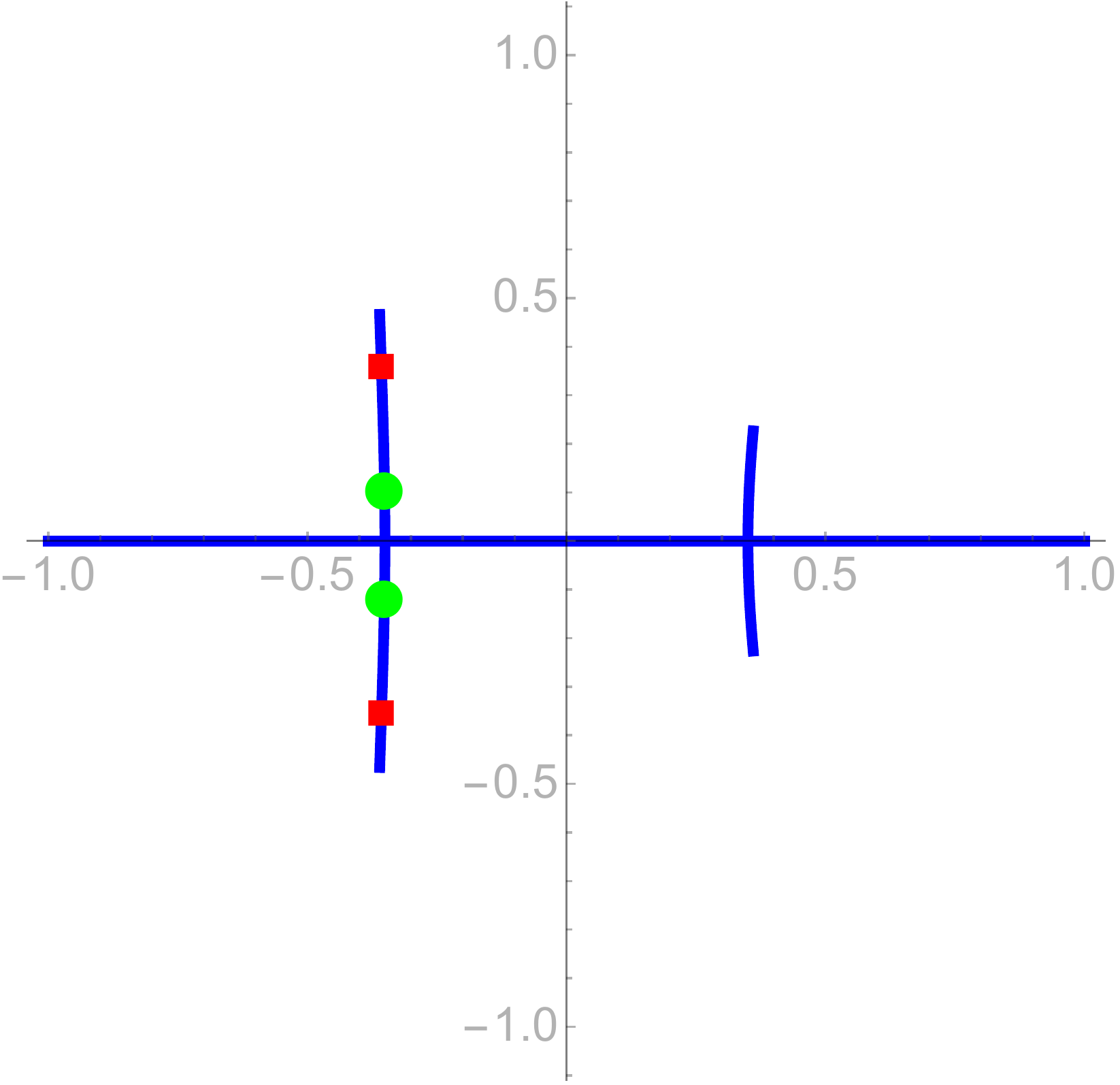}
  \caption{The Lax spectrum for $(k,b) =( 0.65, 0.48)$. Green circles map to
    eigenvalues of $\cL_\pi$ (elements of $\sigma_\pi$
    \eqref{eqn:LmustabilitySpectrum}) through $\Omega(\zeta)$ \eqref{eqn:Omega}.
    In other words, $P=2$ and $T(k)\mu = \pi$. Red squares map to eigenvalues of
    $\cL_{2\pi/3}$: $P=3$ and $T(k)\mu = 2\pi/3$. See
    Remark~\ref{remark:muSymmetry}. \label{fig:muSymmetry}}
\end{figure}

Having put the subharmonic stability results from \cite{deconinck2017stability}
on a rigorous footing, we summarize the findings in Figure
\ref{fig:stabilityCurves}.
\begin{figure}
    \centering
     \begin{tikzpicture}
        \node[inner sep=0pt] (paramSpace) at (0,0)
        {\includegraphics[width=0.4\textwidth]{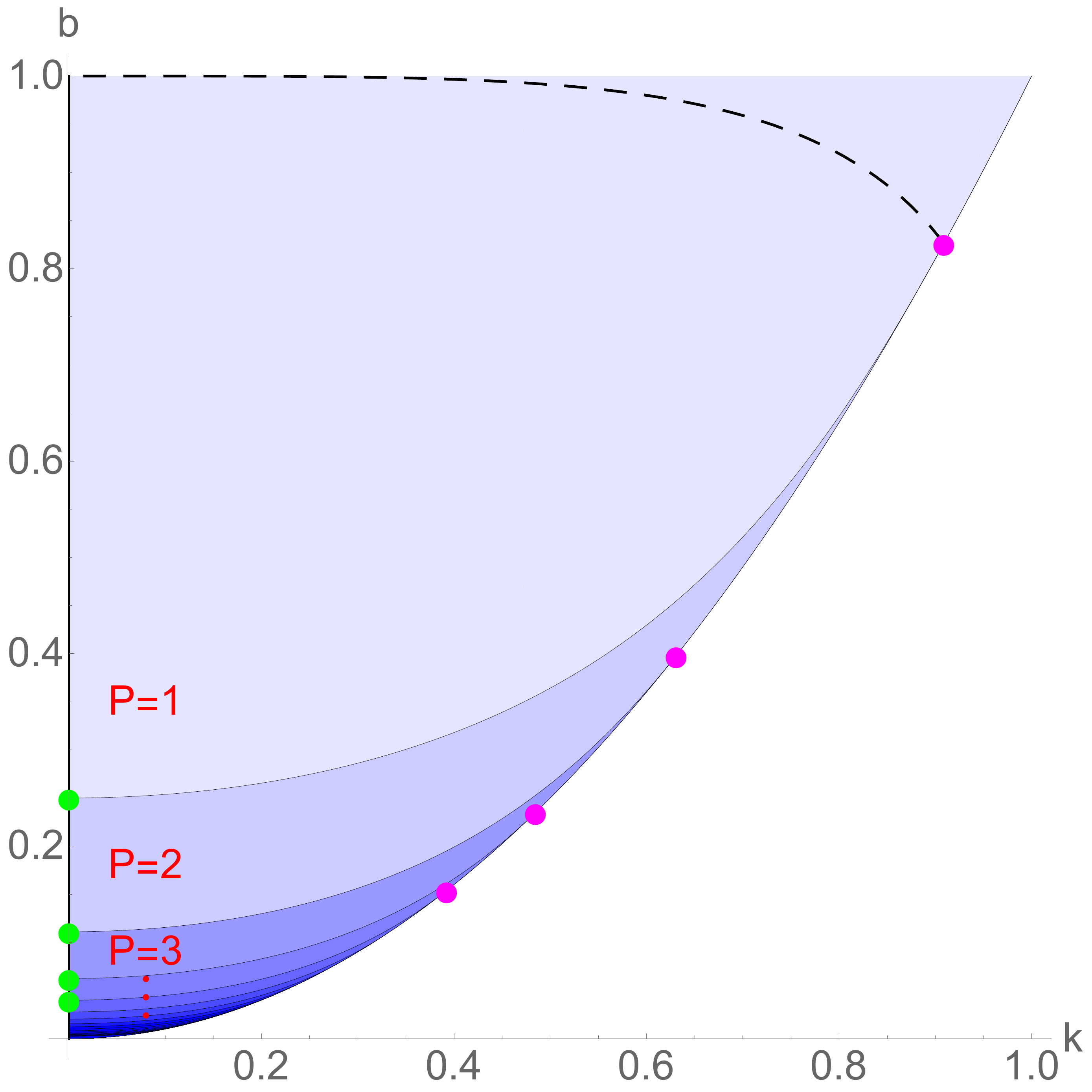}};
        \node (imagCurveAnchor) at (1.7,2.4) {};
        \node[draw] (imagCurve) at (4.9,2.6) {$b=B(k)$, \eqref{eqn:Bk}};
        \draw[decoration={markings,mark=at position 1 with {\arrow[scale=2,
        >=stealth]{>}}}, postaction={decorate}] (imagCurve) -- (imagCurveAnchor) ;
        \node (P2anchor) at (0.55,-0.4) {};
        \node[draw] (P2curve) at (3.9,-0.5) {$M(-\zeta_c) = \pi$,
          \eqref{eqn:mainStabilityBound}};
        \draw[decoration={markings,mark=at position 1 with {\arrow[scale=2,
        >=stealth]{>}}}, postaction={decorate}] (P2curve) -- (P2anchor) ;
        \node (P3anchor) at (-0.51,-1.71) {};
        \node[draw] (P3curve) at (1.9,-2.5) {$M(-\zeta_c) = 2\pi/3$,
          \eqref{eqn:mainStabilityBound}};
        \draw[decoration={markings,mark=at position 1 with {\arrow[scale=2,
        >=stealth]{>}}}, postaction={decorate}] (P3curve) -- (P3anchor) ;
      \end{tikzpicture}
    \caption{The parameter space split up into different regions of subharmonic
      \te{spectral} stability. \label{fig:stabilityCurves} Each solid curve separating regions
      of different color corresponds to equality in
      \eqref{eqn:mainStabilityBound} for different values of $P$. Curves end at
      $b=1/P^2$ (green), the stability bound
      \eqref{eqn:TightStokesSubharmonicStabilityCriteria} for Stokes waves. \te{The
      magenta dots, along the curve $b=k^2$, show where the stability curves,
      which are the boundary of different stability regions, intersect the cn
    solution regime.} The dashed line corresponds to \eqref{eqn:Bk}.  Below it,
      $\zeta_c \in \R\setminus\{0\}$ and above it $\zeta_c \in
      i\R\setminus\{0\}$.
    }
\end{figure}
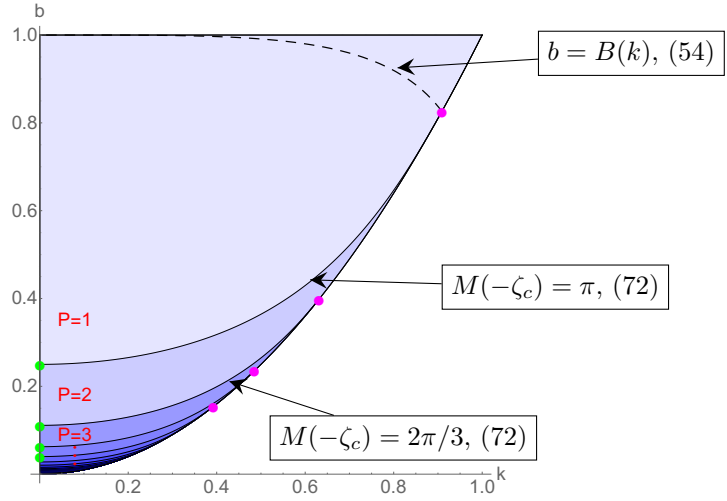
\te{Equality in condition \eqref{eqn:mainStabilityBound} defines a family of
``stability curves'', for $P \in \N$, in the parameter space which split up the
parameter space into regions bounded by these different curves.} The dashed curve
shows where $\zeta_c=0$. Below (above) the dashed curve, $\zeta_c$ is real
(imaginary). The lightest shading represents \te{spectral} stability with respect to
coperiodic perturbations: all solutions are \te{spectrally} stable with respect to such
perturbations \cite{gallay_haragus_OrbitalStability}. Darker shaded regions
represent where solutions additionally are \te{spectrally} stable with respect to perturbations
of higher multiples of the fundamental period. The $P$ labels inside of the
parameter space indicate which solutions are \te{spectrally} stable with respect to
P-subharmonic perturbations in the given region.

\section{The advent of instability}\label{section:krein}
\teee{
Many results on the spectral stability of the elliptic solutions with respect to
subharmonic perturbations were shown in \cite{deconinck2017stability}.  However,
no explanation is given there as to how a solution which is spectrally stable
with respect to subharmonic perturbations loses stability as its parameters are
varied.} We show \teee{here} that as the amplitude increases, the instabilities
of elliptic solutions arise in the same manner as was demonstrated for the
Stokes waves (Section \ref{section:Stokes}).  We begin by using the
Floquet-Fourier-Hill-Method \cite{FFHM} to compute the point spectrum for a
single subharmonic perturbation \eqref{eqn:LmustabilitySpectrum} (Figure
\ref{fig:stabilityToInstabilityFFHM}).  We show that two eigenvalues collide on
the imaginary axis and leave it at the intersection of the figure 8 spectrum and
the imaginary axis.

\sloppy
\te{Consider a point $(k_Q,b_Q)$ in the parameter space lying below a stability
curve labeled $P=Q$ \mbox{($Q = 1,2,3,\ldots$)}, \ie
$M(-\zeta_c(k_Q,b_Q))<2\pi/Q$ (see Figure \ref{fig:stabilityCurves}).} This
solution is spectrally stable with respect to perturbations of period $QT(k)$
and the Lax eigenvalues corresponding to $Q$-periodic perturbations lie on the
real axis. Two Lax eigenvalues, $\hat \zeta_R$ and $\tilde \zeta_R =
\hat\zeta_R\cc$, corresponding to $R>Q$ perturbations lie equidistant from
$-\zeta_c(k_Q, b_Q)\in\R$ on the bands connecting to $-\zeta_c(k_Q,b_Q)$ (see
Remark \ref{remark:muSymmetry} and Figure \ref{fig:stabilityCurves}). The value
$-\zeta_c(k_Q,b_Q)$ lies at the intersection of $\overline{\sigma_L\setminus
\R}$ and $\sigma_L\cap\R$ which maps to the intersection of the figure 8 and the
imaginary axis in the $\sigma_{\cL}$ plane \cite{deconinck2017stability}. The
stability spectrum eigenvalues, $\hat \lam_R = 2\Omega(\hat\zeta_R)$ and $\tilde
\lam_R = 2\Omega(\tilde\zeta_R)$, corresponding to $R$-subharmonic
perturbations, are on the figure 8 to the left and right of the intersection
with the imaginary axis.  As the solution parameters are monotonically varied
approaching the stability curve which is the boundary of the stability region
for $R$-subharmonic perturbations, where $M(-\zeta_c(k_R, b_R)) = 2\pi/R$,
$\tilde\zeta_R$ and $\hat\zeta_R$ move to $-\zeta_c(k_R, b_R)$, and $\hat\lam_R$
and $\tilde\lam_R$ converge to the top of the figure 8.  When this happens, the
solution gains \te{spectral} stability with respect to perturbations of period $RT(k)$.
\te{Spectral} Stability is gained through a Hamiltonian Hopf bifurcation in which two complex
conjugate pairs of eigenvalues come together onto the imaginary axis in the
upper and lower half planes.

\begin{figure}

 \begin{subfigure}[b]{0.31\textwidth}
 \includegraphics[width=\textwidth]{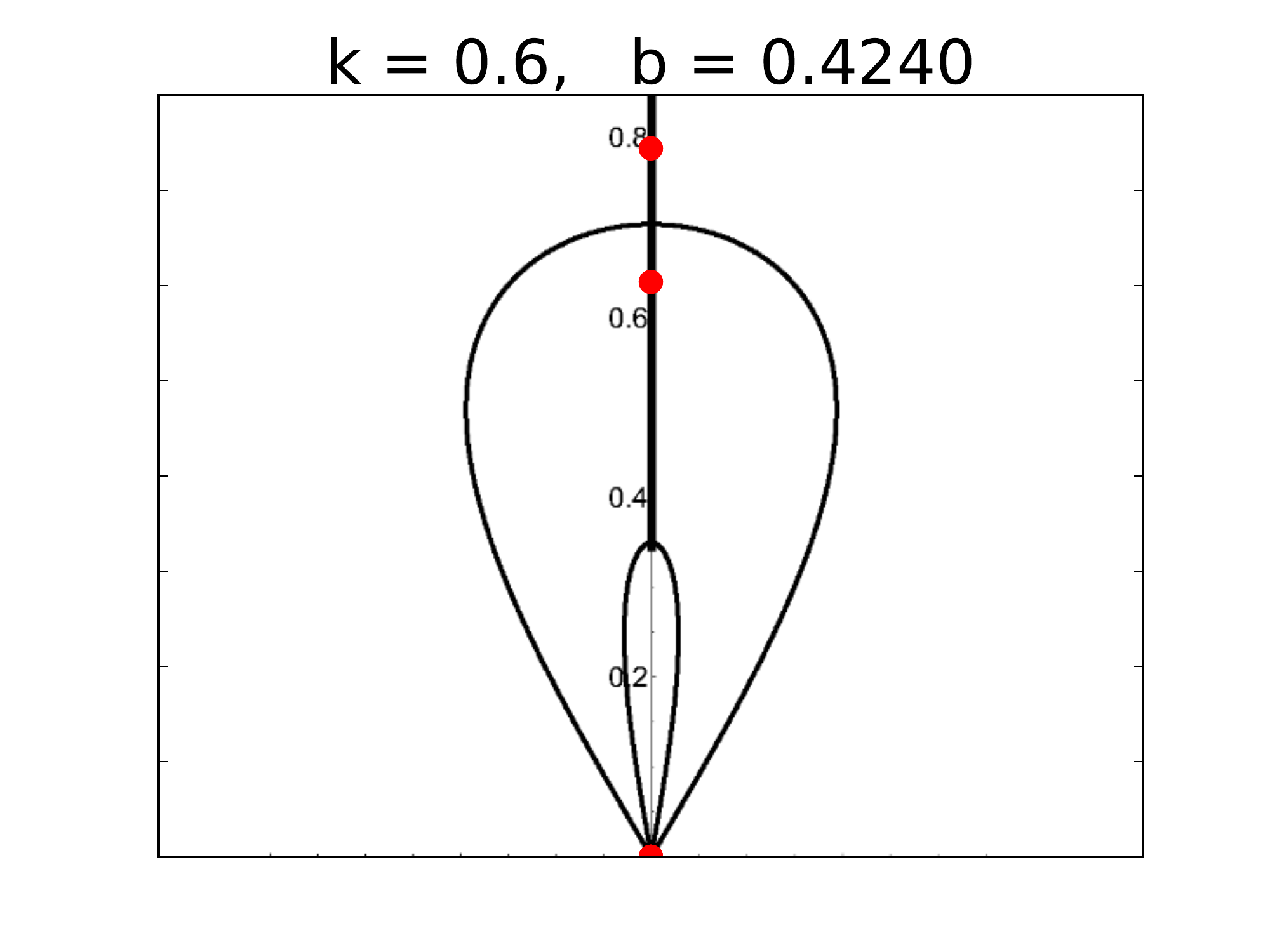}
 \end{subfigure}
 \quad
 \begin{subfigure}[b]{0.31\textwidth}
   \includegraphics[width=\textwidth]{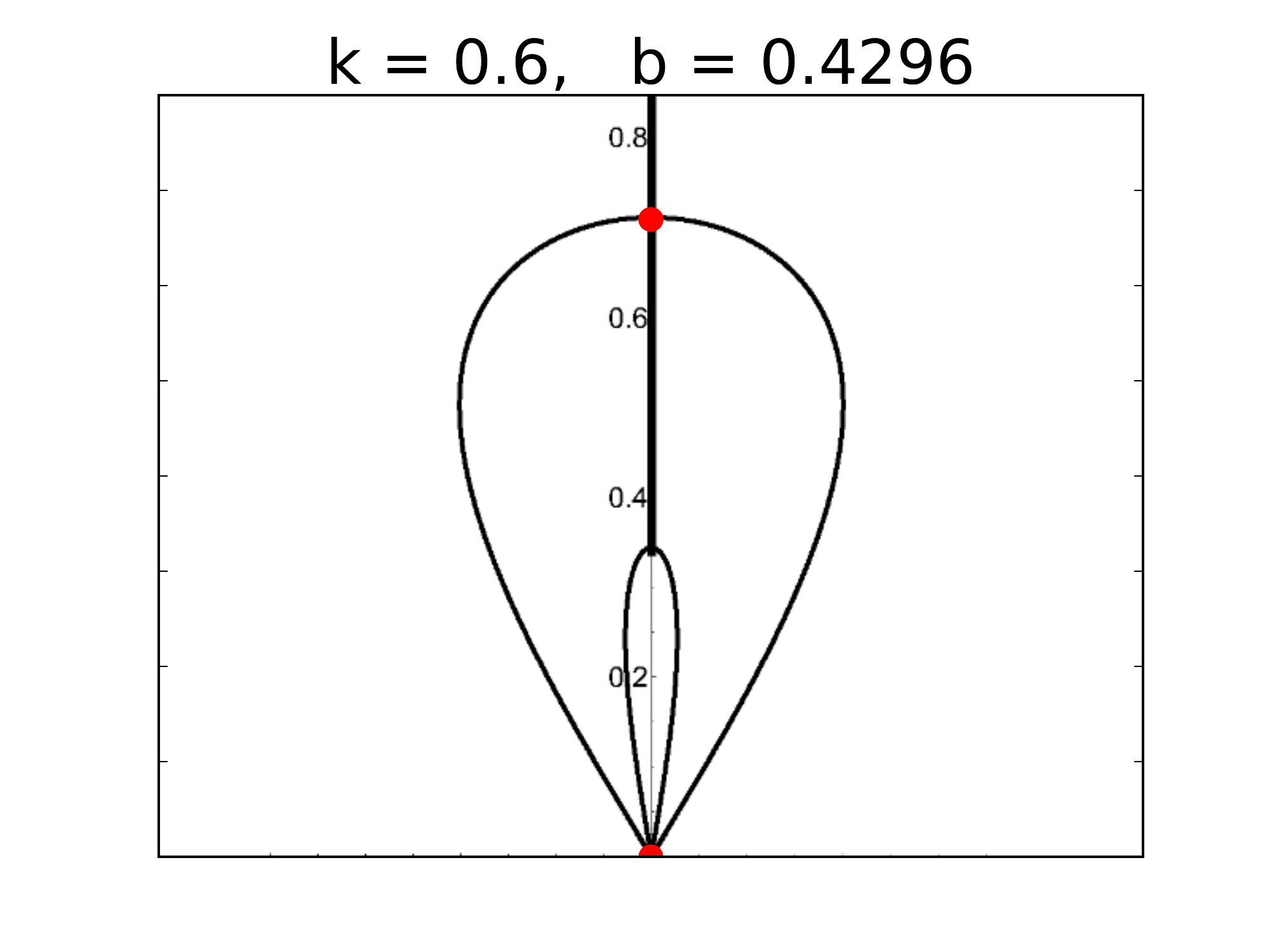}
 \end{subfigure}
 \quad
 \begin{subfigure}[b]{0.31\textwidth}
   \includegraphics[width=\textwidth]{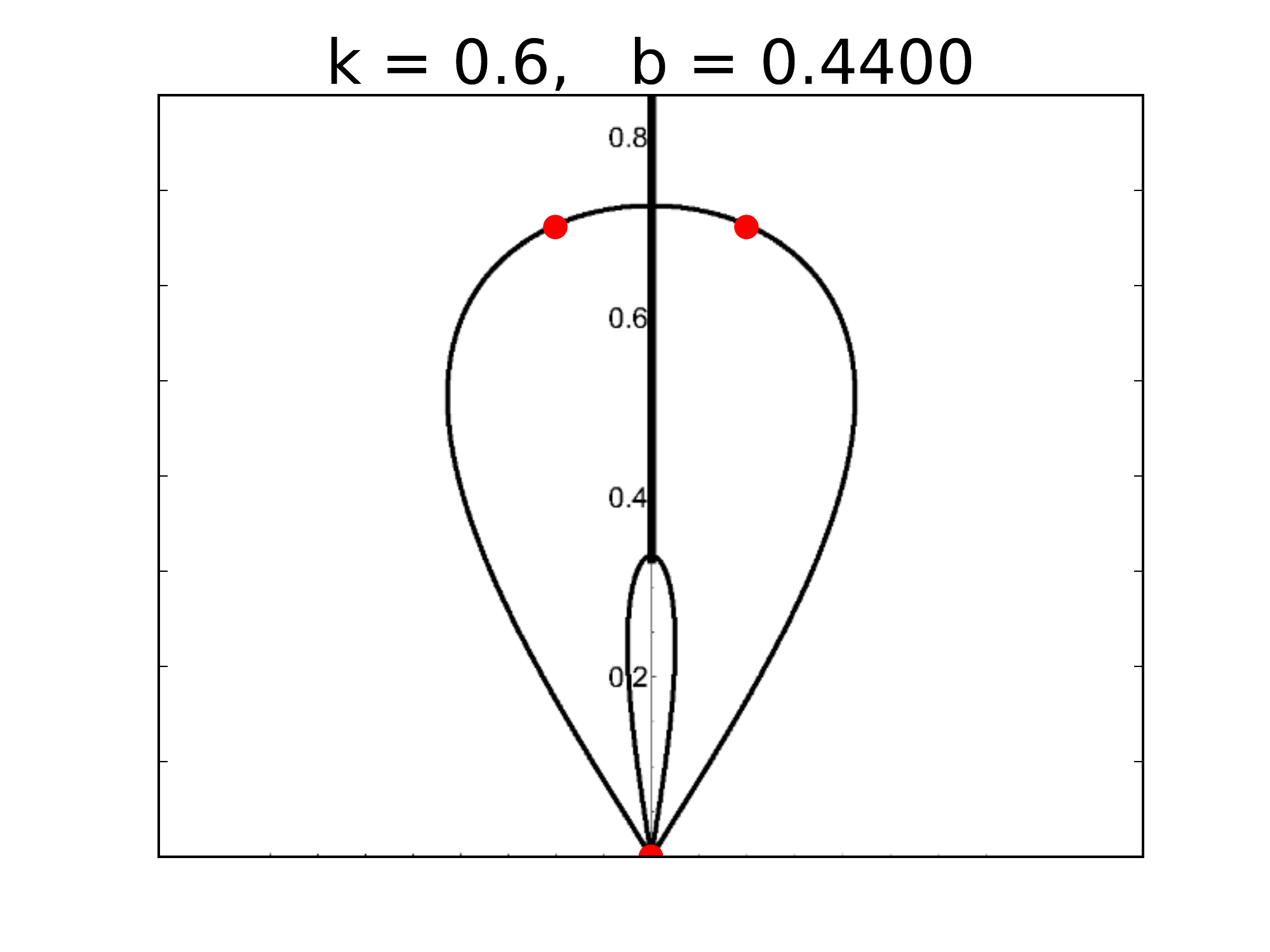}
 \end{subfigure}
 \quad
 \begin{subfigure}[b]{0.29\textwidth}
   \begin{tikzpicture}[spy using outlines={circle, magenta, magnification=5,
        size=2cm, connect spies, transform shape}]
        \node[scale=1.0]{
         \includegraphics[width=\textwidth]{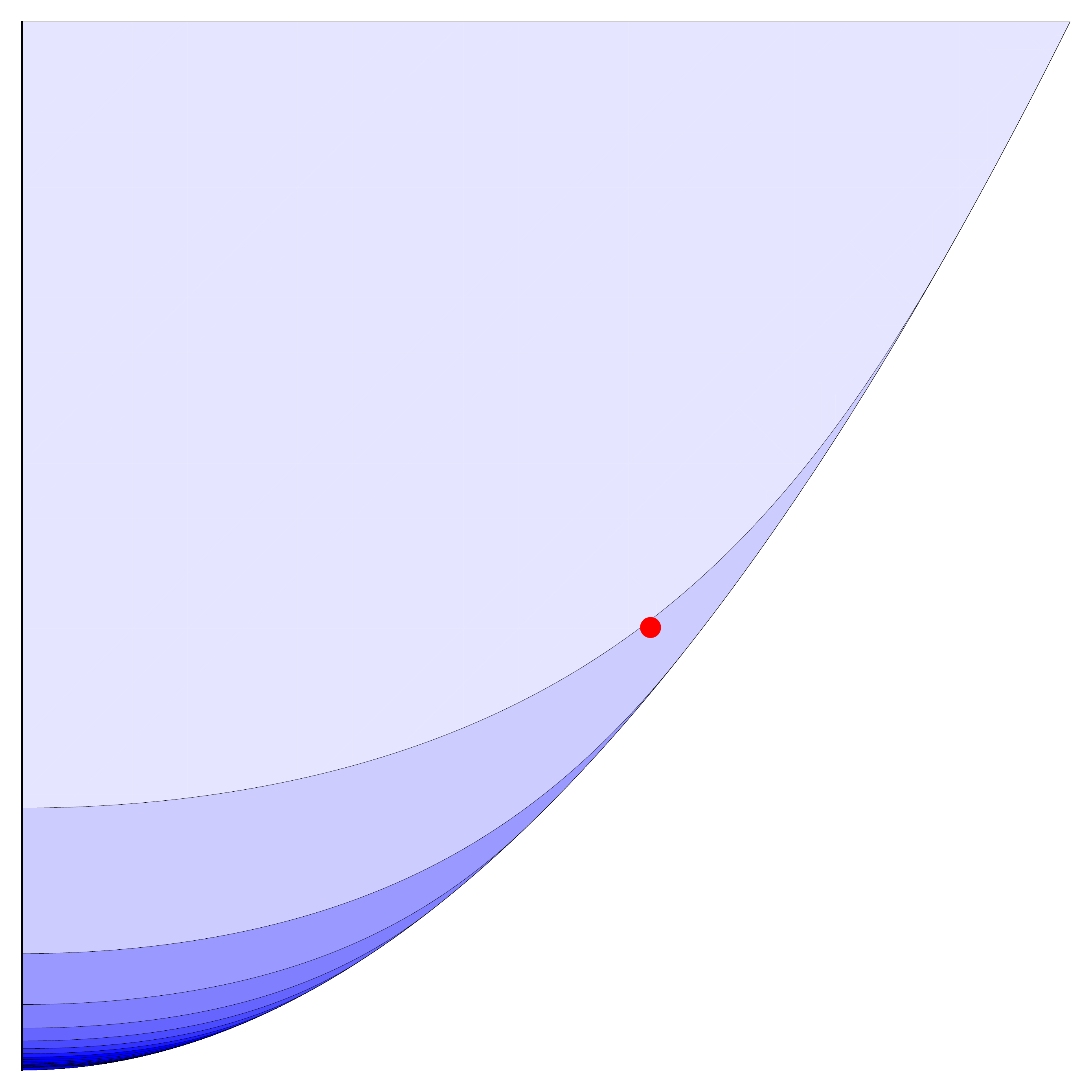} };
        \spy[style={thick}, size=2cm] on (0.47, -0.35) in node [left] at
        (2.2,-1.5);
    \end{tikzpicture}
  \end{subfigure}
  \quad
  \begin{subfigure}[b]{0.29\textwidth}
    \begin{tikzpicture}[spy using outlines={circle, magenta, magnification=5,
      size=2cm, connect spies, transform shape}]
    \node[scale=1.0]{
  \includegraphics[width=\textwidth]{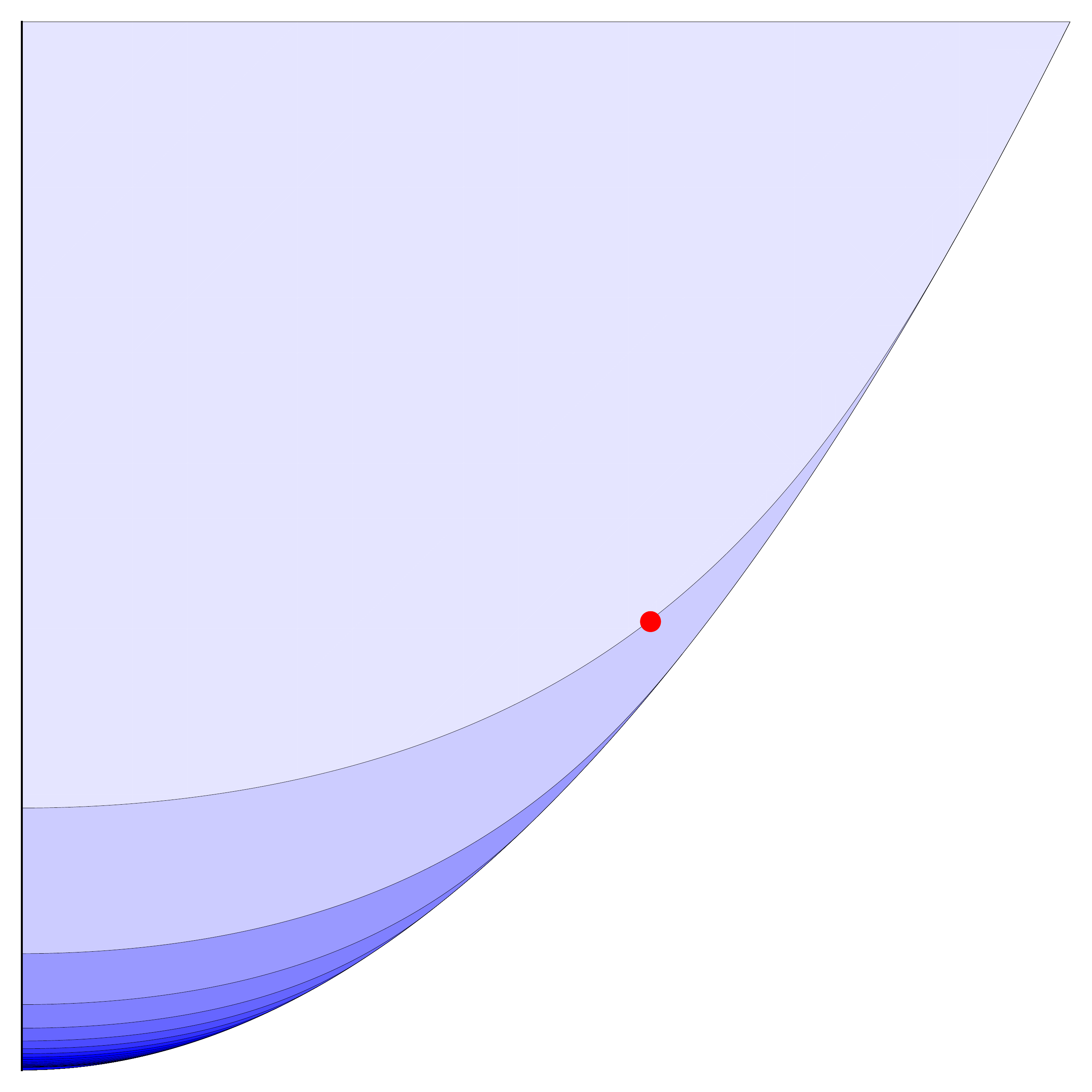} };
    \spy[style={thick}, size=2cm] on (0.47, -0.35) in node [left] at
      (2.2,-1.5);
    \end{tikzpicture}
  \end{subfigure}
  \quad
  \begin{subfigure}[b]{0.29\textwidth}
    \begin{tikzpicture}[spy using outlines={circle, magenta, magnification=5,
      size=2cm, connect spies, transform shape}]
    \node[scale=1.0]{
    \includegraphics[width=\textwidth]{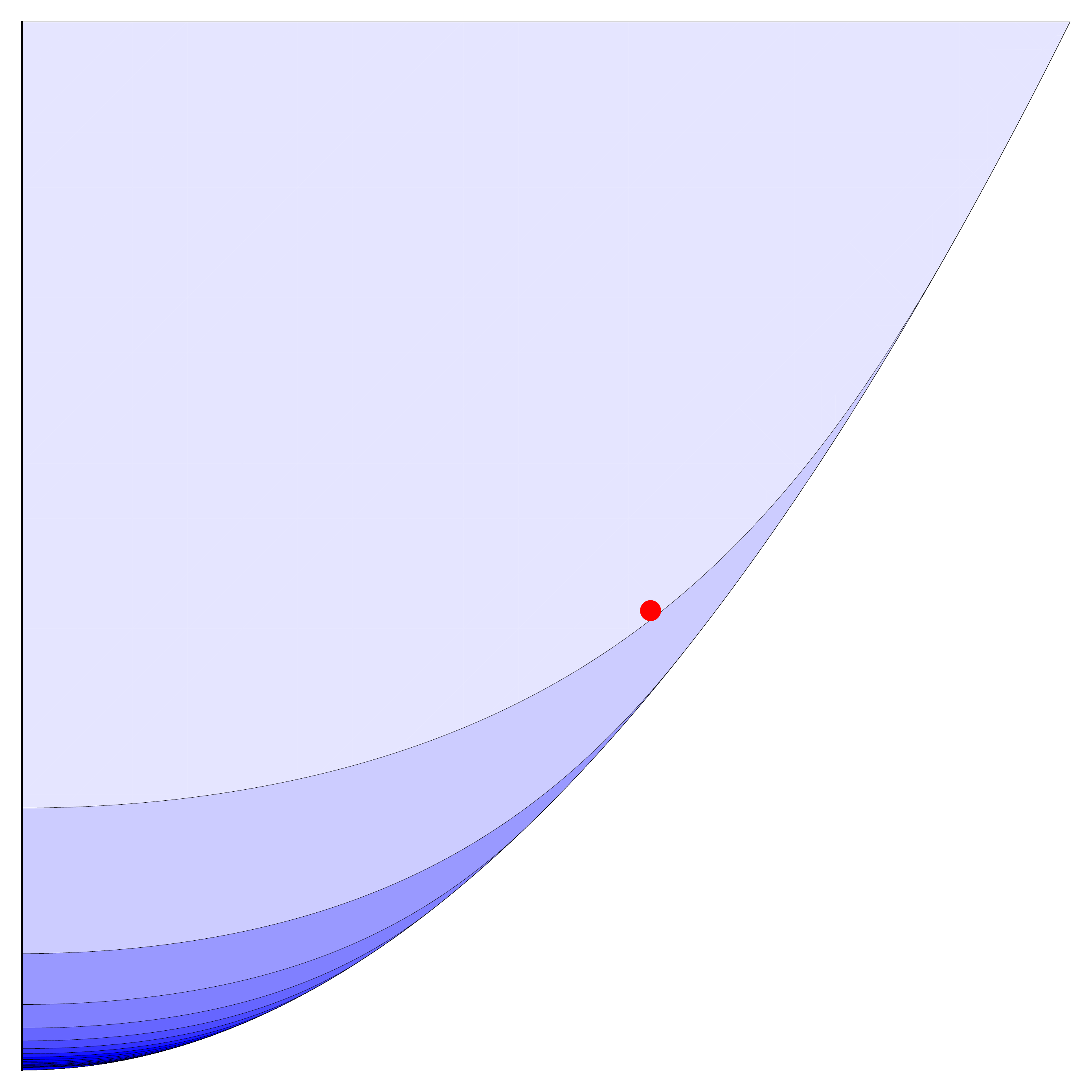} };
      \spy[style={thick}, size=2cm] on (0.47, -0.35) in node [left] at
      (2.2,-1.5);
    \end{tikzpicture}
  \end{subfigure}

\caption{ \label{fig:stabilityToInstabilityFFHM}
For $k=0.6$ and $P=2$ (perturbations of twice the period) we vary $b$ to go from
\te{spectrally} stable to unstable solutions. Top: the top half of the continuous spectrum of
$\cL$ (black, $\RE\lam$ \vs $\IM\lam$, plotted using the analytic expression
\eqref{eqn:Icond} and \eqref{eqn:Omega}) and two eigenvalues with $P=2$
highlighted with red dots computed using the FFHM. Bottom: Location in parameter
space ($k$ \vs $b$)}
\end{figure}

We are interested in the transition from \te{spectrally} stable to unstable solutions. For fixed
$\mu$, consider two eigenvalues $\hat\lam=2\Omega(\hat\zeta)\in i\R$ and
$\tilde\lam=2\Omega(\tilde\zeta) \in i\R$ (\te{spectrally} stable). \te{Stability is lost as the
solution parameters are varied to cross a stability curve, $\hat\zeta \to
-\zeta_c$ and $\tilde\zeta \to -\zeta_c$, entering a new stability region.} The
Krein signature \cite{kollarMiller2014} gives a necessary condition for two
colliding eigenvalues to leave the imaginary axis, leading to instability. For a
given eigenvalue $\lam$ of the operator $\cL_\mu$ associated with a perturbation
of period $PT(k)$ and eigenfunction $W = (W_1, W_2)$, the Krein signature is the
sign of
\begin{align}
  K_2\te{(\zeta)} &:= \ip{W, \mathscr{L}_2 W} = \ip{W, \hat H''(\tilde r, \tilde \ell)W} =
  \int_{-PT(k)/2}^{PT(k)/2} W\cc \hat H''(\tilde r,\tilde \ell) W ~\d
  x,\label{eqn:K2Defn}
\end{align}
\teee{where $\mathscr{L}_2 = \hat H''(\tilde r, \tilde \ell)$ is the Hessian of $\hat
H(r,\ell)$, defined in Appendix \ref{section:integrability}, evaluated at the
elliptic solution.

To relate the eigenfunctions of $J\mathscr{L}_2$ to those of $\cL$, we use
\eqref{eqn:NLSHamiltonianDynamics} in Appendix \ref{section:integrability}.
Linearizing \eqref{eqn:NLSHamiltonianDynamics} about the elliptic solution,
\begin{align}
  \begin{pmatrix}
    r(x,t) \\ \ell(x,t)
  \end{pmatrix} &= \begin{pmatrix} \tilde r(x) \\ \tilde \ell(x) \end{pmatrix} +
  \eps\begin{pmatrix} w_1(x,t) \\ w_2(x,t) \end{pmatrix} + \0{\eps^2},
\end{align}
we obtain
\begin{align}
  \pd{}{t} \begin{pmatrix}
    w_1\\ w_2
\end{pmatrix} &= \begin{pmatrix}
  - \tilde r \tilde \ell & -\frac12 \partial_x^2 - \frac12(\tilde r^2 + 3\tilde
  \ell^2) - \omega\\
  \frac12 \partial_x^2 + \frac12(3\tilde r^2 + \tilde \ell^2) + \omega & \tilde
  r \tilde \ell
\end{pmatrix} \begin{pmatrix} w_1 \\ w_2\end{pmatrix}
= J \hat H''(\tilde r,\tilde \ell) \begin{pmatrix} w_1 \\ w_2\end{pmatrix}
= J \mathscr{L}_2 \begin{pmatrix} w_1 \\ w_2\end{pmatrix}.
\end{align}
Separation of variables,
$(w_1,w_2)^T=e^{\lam t}(W_1, W_2)^T$, and the squared-eigenfunction give $\lam =
2\Omega(\zeta)$ and
\begin{align}
  W = \begin{pmatrix}
    W_1 \\ W_2
  \end{pmatrix} &=
  \begin{pmatrix}
  \varphi_1^2 + \varphi_2^2 \\
  -i\varphi_1^2 + i\varphi_2^2
\end{pmatrix}.\label{eqn:squaredEFReIm}
\end{align}
From the expressions for the eigenfunctions \eqref{eqn:squaredEF} and
\eqref{eqn:squaredEFReIm} it is clear that if an eigenfunction $(U,V)^T$ of
$\mathcal{L}$ corresponds to a spectral element $\lambda$, then there is a
corresponding eigenfunction $(W_1,W_2)^T$ of $J \mathscr{L}_2$ with the same
spectral element $\lambda = 2\Omega(\zeta)$.
}

Since $2\Omega W = J\mathscr{L}_2W$ and since $J$ is invertible, we find from
\eqref{eqn:squaredEFReIm} that
\begin{align}
  W\cc \mathscr{L}_2 W =2\Omega W\cc J^{-1}W =  2\Omega(W_1 W_2\cc -
  W_2W_1\cc) = 4i\Omega(\abs{\varphi_1}^4 - \abs{\varphi_2}^4) \label{eqn:WLW},
\end{align}
with $\varphi_1 = -\gamma(x)B(x),$ and $\varphi_2 = \gamma(x) (A(x) - \Omega)$.
For a fixed $\mu$ and a corresponding spectrally stable solution, $\Omega(\zeta)
\in i\R$ for $\zeta \in \R$. When $\zeta \in \R$, the preimage of
$\lam(\zeta)\in i\R$ is only one point \eqref{eqn:ReImOmega} so that by
Theorem~\ref{thm:squaredEigCompleteness}, $\lam(\zeta)$ is a simple eigenvalue.
Therefore, we compute the Krein signature only for $\zeta \in \R$. From
\eqref{eqn:gamma(x)},
\begin{align}
  \gamma(x) &= \frac{\gamma_0}{B} \exp(\text{imag}) \exp\lt(-\int
    \frac{(A-\Omega)\phi}{B}~\d x\rt),
\end{align}
where ``imag'' represents imaginary terms which are not important for the
magnitude of $\gamma(x)$. The magnitude of $\gamma(x)$ depends
critically on
\begin{align}
\begin{split}
 -\frac{(A-\Omega)\phi}{B} &=  \frac{C \phi}{(A+\Omega)}
    = \frac{-i\zeta|\phi|^2 - (\tilde r \tilde r_x + \tilde\ell
    \tilde\ell_x + i \tilde \ell \tilde r_x - i\tilde r\tilde \ell_x)/2}{\zeta^2
    - |\phi|^2/2 - \omega/2 + i\Omega}\\
 &= \text{imag} + \frac12\D{x}\ln\abs{i(A(x)+\Omega)},
\end{split}
\end{align}
where $\phi = \tilde r + i \tilde \ell$ is the elliptic solution whose stability
is being investigated. Since $A(x) -\Omega \in i\R$, it follows that
\begin{align}
 \gamma(x)
 &= \frac{\gamma_0}{B}\exp\lt(\text{imag}\rt)\abs{i(A(x)+\Omega)}^{1/2}.
\end{align}
Equating $|\gamma_0|=1$,
\begin{align}
 \abs{\gamma(x)}^2 &= \frac{\abs{A+\Omega}}{|B|^2} = \frac{1}{|A-\Omega|},
\end{align}
so that
\begin{align}
 \abs{\varphi_1}^4 &= |\gamma|^4|B|^4 = |A+\Omega|^2, \qquad |\varphi_2|^4 =
|\gamma|^4|A-\Omega|^4 = |A-\Omega|^2.
\end{align}
Further,
\begin{align}
  W\cc \mathscr{L}_2 W &= 4i\Omega\lt(|A+\Omega|^2 - |A-\Omega|^2\rt) = -16\Omega^2 i A,
\end{align}
implies
\begin{align}
 K_2(\zeta) = -16\Omega^2(\zeta)\int_{-PT(k)/2}^{PT(k)/2}\lt(\zeta^2 - \frac 12|\phi|^2 -
\frac{\omega}{2}\rt)~\d x \label{K2-orig},
\end{align}
which is the same $K_2$ found in \cite{bottman2011elliptic} with appropriate
modifications for the focusing case. This integral can be computed directly
using elliptic functions \cite[equation (310.01)]{byrdFriedman}:
\begin{align}
\begin{split}
  K_2\te{(\zeta)} &= -32\Omega^2(\zeta)PT(k)\lt(\zeta^2 + \frac{b}{4} + \frac14\lt(
    1-k^2-2\frac{E(k)}{K(k)}\rt)\rt) = -32\Omega(\zeta)^2 PT(k)(\zeta^2-\zeta_c^2).\label{eqn:K2def}
\end{split}
\end{align}
Note that since $\Omega^2<0$ for stable eigenvalues, $K_2(\zeta) <0$ for $\zeta
\in (-\zeta_c,\zeta_c)$ changing sign at $\zeta = \pm \zeta_c$. Therefore the
two eigenvalues which collide at $-\zeta_c$ have opposite Krein signatures, a
necessary condition for instability.

For the trivial-phase solution, $\Omega(\zeta) = \Omega(-\zeta)$, so the Krein
signature calculation here might not be sufficient, since the colliding
eigenvalues, $\hat\lam$ and $\tilde\lam$, might not be simple. Our remaining 
stability results do not rely on this fact. \teee{Computing the
  Krein signature for Stokes waves (Section \ref{section:Stokes}) is simpler
than the calculation here, but it is omitted for brevity.}

\section{Orbital stability}\label{section:orbitalStability}
The results on spectral stability may be strengthened to orbital stability by
constructing a Lyapunov functional in conjunction with the results of
\cite{grillakis1987stability, maddocksSachs}. \teee{In Theorem
  \ref{thm:mainNTPStabilityResult}, we have established spectral stability for
  solutions below the curve \eqref{eqn:mainStabilityBound} (see Figure
  \ref{fig:stabilityCurves})}. In this section we show that those solutions are
  also orbitally stable. \teee{To this end, we use the higher-order conserved
    quantities of NLS (see Appendix \ref{section:integrability}).} \te{
  \begin{definition}
    A stationary solution $\tilde\Psi$ of \eqref{eqn:NLS} is orbitally stable
    with respect to the norm $\norm{\cdot}$ if for any given $\eps>0$ there
    exists a $\delta>0$ such that
    \begin{align}
      ||\Psi(x,0) - \tilde\Psi(x,0)||<\delta,
    \end{align}
    implies that for all $t>0$,
    \begin{align}
      \inf_{g \in G} ||\Psi(x,t) - T(g)\tilde\Psi(x,t)||<\eps,
    \end{align}
    where $T(g)$ is the action of an element $g$ of the group of symmetries $G$.
  \end{definition}
}

To prove nonlinear stability, we construct a Lyapunov function, \ie a constant
of the motion $\cK(r,\ell)$ for which the solution $(\tilde r, \tilde \ell)$
is an unconstrained minimizer:
\begin{align}
  \cK'(\tilde r, \tilde
  \ell)=0, \quad
  \D{t}\cK(\tilde r,\tilde\ell)=0,
  \quad \ip{v, \cK''(\tilde r, \tilde \ell)v}>0, \quad \forall v \in
  \mathbb{V}, v\neq 0. \label{eqn:LyapunovDefn}
\end{align}
In Section \ref{section:krein} it is shown that the energy $\hat H$ satisfies
the first two conditions in \eqref{eqn:LyapunovDefn} but not the third since
$K_2$ is not of definite sign. When evaluated at stationary solutions, each
equation defined in \eqref{eqn:nthNLS} satisfies the first and second
conditions.  Following the work of \cite{bottman2011elliptic,
mkdvOrbitalStability, maddocksSachs,kdvFGOrbitalStability} we choose
one member of \eqref{eqn:nthNLS} to satisfy the third condition by choosing the
constants $c_{n,j}$ in a particular manner. \te{A different approach to finding
  a Lyapunov function is used in \cite{GallayPelinovskyOrbital} for
defocusing NLS.}

Linearizing the $n$-th NLS equation about the elliptic solution results in
\[
  w_{t_n} = J \mathscr{L}_n w, \qquad \mathscr{L}_n = \hat H_n''(\tilde r,
  \tilde \ell).
\]
The squared-eigenfunction connection and separation of variables gives
\begin{align}
  2 \Omega_n W(x) = J \mathscr{L}_n W(x),\label{eqn:OmeganLn}
\end{align}
where $\Omega_n$ is defined by
\begin{align}
w(x,t_n) = e^{\Omega_n t_n}\begin{pmatrix}W_1(x) \\ W_2(x)\end{pmatrix} =
    e^{\Omega_n t_n} W(x),
\end{align}
and where $W(x)$ \te{is any eigenfunction of $\cL_2$}.
The relation
\begin{align}
 \Omega_n^2(\zeta) = p_n^2(\zeta)\Omega^2(\zeta),
 \qquad n\geq2\label{eqn:OmeganSquared},
\end{align}
where $p_n$ is a polynomial of degree $n-2$, is found in
\cite{bottman2011elliptic} and applies in the focusing case as well. When $n=2$,
$p_2 = 1$ so that $\Omega_2 = \Omega$ and \eqref{eqn:OmeganLn} implies
\begin{align}
  2J^{-1}W = \frac{1}{\Omega}\mathscr{L}_2 W=\frac{1}{\Omega}\mathscr{L} W,
\end{align}
\te{for any eigenfunction $W$ of $\cL_2$.} The definition of $K_2$
\eqref{eqn:K2Defn} and \eqref{eqn:OmeganSquared} imply
\begin{align}
  K_n\te{(\zeta)} &:=  \ip{W, \mathscr{L}_n W} = \ip{W, \hat H_n''(\tilde r, \tilde \ell)
  W} = \frac{\Omega_n}{\Omega}\int_{-PT(k)/2}^{PT(k)/2} W\cc \mathscr{L}_2 W~\d
  x = p_n(\zeta) K_2(\zeta).\label{eqn:KnDefn}
\end{align}
$K_2(\zeta)$ takes the sign $+,-,+$ for $\zeta \in (-\infty, -\zeta_c),~\zeta \in
(-\zeta_c,\zeta_c),$ and $\zeta \in (\zeta_c, \infty)$ respectively.  Since
$p_4(\zeta)$ is quadratic, we use $K_4(\zeta) = p_4(\zeta) K_2(\zeta)$, where
$p_4(\zeta)$ is defined by \eqref{eqn:OmeganSquared}. Adjusting the constants of
$p_4$ so that it has the same sign as $K_2$ with zeros at $\zeta = \pm \zeta_c$
makes $K_4$ nonnegative.  In order to calculate $\Omega_4(\zeta)$, we need
\begin{align}
 \hat T_4 = T_4 + c_{4,3}T_3 + c_{4,2} T_2 + c_{4,1} T_1 + c_{4,0} T_0,
\end{align}
since $\Omega_4$ is defined by $\hat T_4\chi = \Omega_4\chi$ by separation of
variables in \eqref{eqn:tau4dynamics}. The $c_{4,k}$ are not entirely arbitrary.
They are determined by requiring that the stationary elliptic solutions are
stationary with respect to $t_4$, or
\begin{align}
 \pd{}{t_4} \begin{pmatrix} r \\ \ell\end{pmatrix} = J \hat H_4' =  J \lt( H_4' +
 c_{4,3}H_3' + c_{4,2} H_2' + c_{4,1} H_1' + c_{4,0}H_0'\rt) = 0.
\end{align}
Since $J$ is invertible,
\begin{align}
 \hat H_4' = H_4' + c_{4,3} H_3' + c_{4,2}H_2' + c_{4,1} H_1' + c_{4,0} H_0' =
 0,
\end{align}
when evaluated at the stationary solution. Equating
\begin{align}
 0 = \Psi_{\tau_4} + c_{4,3}\Psi_{\tau_3} + c_{4,2}\Psi_{\tau_2} +
c_{4,1}\Psi_{\tau_1} + c_{4,0}\Psi_{\tau_0},
\end{align}
and using \eqref{eqn:taundynamics} with $\Psi$ defined in
\eqref{eqn:ellipticSolns}, we find
\begin{subequations}
\begin{align}
 c_{4,0} &= \omega c_{4,2} - c c_{4,3} + \frac 18\lt( 1 + 15b^2 + 4k^2 + k^4 +
10b + 10bk^2\rt),\label{eqn:c40}\\
 c_{4,1} &= \frac12 c - \frac12 \omega c_{4,3},\label{eqn:c41}
\end{align}
\end{subequations}
with $c_{4,2}$ and $c_{4,3}$ arbitrary. Then
\begin{align}
 \Omega_4^2 = \frac{1}{16}\lt(2 \omega + 4\zeta^2 + 4 c_{4,2} + 4\zeta
c_{4,3}\rt)^2 \Omega_2^2,\label{eqn:Omega4defn}
\end{align}
so that
\begin{align}
 p_4(\zeta) = \zeta^2 + \zeta c_{4,3} + \frac12 \omega + c_{4,2}.
\label{eqn:p4defn}
\end{align}
The constants $c_{4,2}$ and $c_{4,3}$ are chosen so that $K_4(\zeta) =
p_4(\zeta) K_2(\zeta)\geq 0$. Setting
\begin{subequations}
\begin{align}
 c_{4,3} &= 0 \label{eqn:c43},\\
    c_{4,2} &= -\frac{\omega}{2} + \frac b4 + \frac14\lt(1-k^2- 2\frac{E(k)}{K(k)}\rt)
\label{eqn:c42},
\end{align}
\end{subequations}
we have
\begin{align}
 K_4(\zeta) &= -32\Omega^2(\zeta) PT(k) \lt(\zeta^2 - \zeta_c^2\rt)^2 \geq
    0\label{eqn:K4positive},
\end{align}
for $\zeta \in \R$ and equality only at $\zeta = \pm\zeta_c$ and the roots of
$\Omega$.
The result \eqref{eqn:K4positive} has only been proven for eigenfunctions of
$\cL_2$. However since the eigenfunctions of $\cL_2$ are complete in
$L^2_{\text{per}}([-T(k)/2,T(k)/2])$ \cite{haragus_kapitula-SpectraPeriodicWaves}
the results apply to all functions in $L^2_{\text{per}}([-T(k)/2,T(k)/2])$.
\te{This result implies that $H_4$, with the constants chosen above, acts as a
  Lyapunov functional for the spectrally stable elliptic solutions with respect
  to the $t_4$ dynamics. However since all flows of the NLS hierarchy commute,
  $H_4$ is a conserved quantity with respect to the $t$ dynamics as well.}
  Therefore whenever solutions are spectrally stable with respect to a given
  subharmonic perturbation, they are also formally stable \cite{maddocksSachs}.

To go from formal to orbital stability, the conditions of
\cite{grillakis1987stability} must be satisfied. The kernel of the functional
$\hat H_4''(\tilde r, \tilde \ell)$ must consist only of the infinitesimal
generators of the symmetries of the solution $(\tilde r, \tilde \ell)$. The
infinitesimal generators of the Lie point symmetries correspond to the values of
$\zeta$ for which $\Omega(\zeta) = 0$, so the kernel of $\hat H_4''(\tilde r,
\tilde \ell)$ contains the infinitesimal generators of the Lie point symmetries.
In order for the kernel to consist only of this set, we need strict inequality
in \eqref{eqn:mainStabilityBound}. This comes from the following lemma.

\begin{lemma}\label{lem:badset}
    Let $b,k$ and $P$ be such that \eqref{eqn:mainStabilityBound} holds with a
    strict inequality. Then the set
    \begin{align}
      S := \{\zeta \in \sigma_L : M(\zeta) = m2\pi / P, \quad m = 0, \ldots,
      P-1\}\label{eqn:badset}
    \end{align}
    does not contain $\pm \zeta_c$.
\end{lemma}
\begin{proof}
  Since $M(-\zeta_c)<2\pi /P$, the only possibility for $-\zeta_c$ to be in $S$
  is that $M(-\zeta_c) = 0 \mod{2\pi}$. But since $-\zeta_c$ represents the
  intersection of the branch of spectra and the real line,
  Lemma~\ref{lem:NoDoublePointsSpine} applies and $M(-\zeta_c)\neq 0
  \mod{2\pi}$. Since $M(\zeta_c)< M(-\zeta_c)$, it is also the case that
  $\zeta_c$ is not in $S$
\end{proof}

The above lemma implies that if $M(-\zeta_c)<2\pi /P$, the kernel of $\hat
H_4''(\tilde r, \tilde \ell)$ consists only of the roots of $\Omega(\zeta)$.
\te{\textbf{It follows that, for a fixed perturbation with period $PT(k)$, all
  solutions which are spectrally stable with respect to that perturbation and
  whose parameters do not lie \textit{on} stability curves (the boundary of
subharmonic stability regions, at which $M(-\zeta_c) = 2\pi/P$) are also
orbitally stable. }}

\section*{Conclusion}
We have proven the orbital stability with respect to subharmonic perturbations
for the elliptic solutions of the focusing nonlinear Schr\"{o}dinger equation.
The necessary condition for stability \eqref{eqn:mainStabilityBound} is shown to
also be a sufficient condition with the help of a numerical check. We see three
main remaining tasks to be completed for this problem: (i) remove the numerical
check for sufficiency of Theorem~\ref{thm:mainNTPStabilityResult}; (ii) determine
whether or not solutions lying on stability curves, $M(-\zeta_c) = 2\pi/P$, are
orbitally stable; and (iii) prove that the solutions satisfying $b>B(k)$ in
Theorem \ref{thm:imaginaryNTPStabilityResult} are not stable with respect to
2-subharmonic perturbations.

The main difficulty in establishing the results presented in this paper is that
the Lax pair does not define a self-adjoint spectral problem. Work towards
establishing similar nonlinear stability results for the sine-Gordon equation
\cite{SGstability}, for which the Lax spectral problem is both not self-adjoint
and is a quadratic eigenvalue problem, is currently underway. This is
generalized in \cite{upsalInPrep} by computing the Floquet discriminant for all
equations in the AKNS hierarchy and other integrable equations.

\section*{Acknowledgments}
The authors acknowledge Greg Forest, Stephane Lafortune, and Vishal Vasan for
helpful conversations and ideas. In addition, the referees are thanked for many
useful suggestions. This work was generously supported by the National Science
Foundation under award number NSF-DMS-1522677 (BD). Any opinions, findings, and
conclusions or recommendations expressed in this material are those of the
authors and do not necessarily reflect the views of the funding sources.

\begin{appendices}
\tee{  \section{Integrability background}\label{section:integrability}
The results presented in this section are found in more detail in classic
sources such as \cite{AKNS74, SolitonsAndIST}. NLS \eqref{eqn:NLS} is a
Hamiltonian system with canonical variables $\Psi$ and $i\Psi\cc$, \ie it can be
written as an evolution equation
\begin{align}
  \pd{}{t} \begin{pmatrix} \Psi \\ i\Psi\cc \end{pmatrix} = J H'(\Psi,i \Psi\cc)
  = J \begin{pmatrix} \delta H/\delta \Psi \\ \delta H /
    \delta(i\Psi\cc)\end{pmatrix},
\end{align}
\te{for a functional $H$ and where
\begin{align}
   J &= \begin{pmatrix} 0 & 1 \\ -1 & 0 \end{pmatrix}.\label{eqn:J}
\end{align}
\te{We define the variational gradient \cite{SolitonsAndIST} of a function $F(u,v)$
by}
\begin{align}
     F'(u,v) = \lt(\displaystyle
         \frac{\delta F}{\delta u}~, ~ \displaystyle \frac{\delta F}{\delta v}
        \rt)^T
               &= \lt(
                 {\displaystyle \sum_{j=0}^N (-1)^j \partial_x^j
                 \frac{\partial F}{\partial u_{jx}}}~,~
                 {\displaystyle \sum_{j=0}^N (-1)^j \partial_x^j
                 \frac{\partial F}{\partial v_{jx}}}
               \rt)^T,
\end{align}
where $u_{jx} = \partial_x^j u$, and $N$ is the highest-order $x$-derivative of
$u$ or $v$ in $F$. }The quantity $H(\Psi,i\Psi\cc)$ is conserved under
\eqref{eqn:NLS} and is the Hamiltonian of \eqref{eqn:NLS}. The Hamiltonian is
one of an infinite number of conserved quantities of NLS. We label these
quantities $\{H_j\}_{j=0}^\infty$. We need the first five conserved
quantities:
\begin{subequations}
\begin{align}
  H_0 &= 2\int \abs{\Psi}^2~ \d x,\\
  H_1 &= i \int \Psi_x\Psi\cc~ \d x,\\
  H_2 &= \frac12 \int\lt(\abs{\Psi_x}^2 - \abs{\Psi}^4\rt)~ \d x,\\
  H_3 &= \frac{i}{4}\int\lt(\Psi_x\cc \Psi_{xx} - 3 \abs{\Psi}^2 \Psi\cc
  \Psi_x\rt)~\d x,\\
  H_4 &= \frac18 \int \lt( \abs{\Psi_{xx}}^2 - \Psi^2 \Psi_x^{*2}  - 6
  \abs{\Psi}^2 \abs{\Psi_x}^2 + \abs{\Psi}^2\Psi\cc \Psi_{xx} + 2
  \abs{\Psi}^6\rt)~\d x.
\end{align}
\end{subequations}
The above equations can be written in terms of $\Psi$ and $i\Psi\cc$ by using
$\abs{\Psi_{jx}}^2 = \Psi_{jx}\Psi_{jx}\cc$. The above integrals are evaluated
over one period $T(k)$, for periodic or quasi-periodic solutions. Each $H_n$
defines an evolution equation with respect to a time variable $\tau_n$ by
\begin{align}
  \pd{}{\tau_n} \begin{pmatrix} \Psi \\ i\Psi\cc\end{pmatrix} &= J
    H_n'(\Psi,i \Psi\cc) = J \begin{pmatrix} \delta H_n/\delta \Psi \\
\delta H_n/\delta (i\Psi\cc) \end{pmatrix}.\label{eqn:NLShierarchy}
\end{align}
When $n=2$ and $\tau_2 = t$, $H_2 = H$ is the NLS Hamiltonian:
\eqref{eqn:NLShierarchy} is equivalent to \eqref{eqn:NLS}. Letting $\Psi =
(r+i\ell)/\sqrt{2}$ and $i\Psi\cc = i(r-i\ell)/\sqrt{2}$, where $r$ and $\ell$
are the real and imaginary parts of $\Psi$ respectively,
\eqref{eqn:NLShierarchy} becomes
\begin{align}
  \pd{}{\tau_n} \begin{pmatrix} r \\ \ell\end{pmatrix} &= J
    H_n'(r,\ell) = J \begin{pmatrix} \delta H_n/\delta r \\
\delta H_n/\delta \ell \end{pmatrix}.\label{eqn:NLShierarchycomplex}
\end{align}
We use \eqref{eqn:NLShierarchy} and \eqref{eqn:NLShierarchycomplex}
interchangeably and refer to $H_n(r,\ell)$ and $H_n(\Psi, i \Psi\cc)$ as $H_n$
when the context is clear. The collection of equations \eqref{eqn:NLShierarchy}
is the NLS hierarchy \cite[Section 1.2]{SolitonsAndIST}. The first five members
of the hierarchy are
\begin{subequations}
\begin{align}
  \Psi_{\tau_0} &= -2 i \Psi \label{eqn:tau0dynamics},\\
  \Psi_{\tau_1} &= \Psi_x \label{eqn:tau1dynamics},\\
  \Psi_{\tau_2} &= i \abs{\Psi}^2 \Psi + \frac i2 \Psi_{xx}
  \label{eqn:tau2dynamics},\\
  \Psi_{\tau_3} &= -\frac32 \abs{\Psi}^2\Psi_x - \frac14
  \Psi_{xxx}\label{eqn:tau3dynamics},\\
  \Psi_{\tau_4} &= -\frac34i |\Psi|^4 \Psi - \frac34 i \Psi\cc
  \Psi_x^2 - \frac i2 \Psi |\Psi_x|^2 - i |\Psi|^2 \Psi_{xx} - \frac i4 \Psi^2
  \Psi_{xx}\cc - \frac i8 \Psi_{xxxx}\label{eqn:tau4dynamics}.
\end{align}
\label{eqn:taundynamics}
\end{subequations}
Each equation obtained in this manner is integrable and shares the conserved
quantities $\{H_j\}_{j=0}^\infty$.

Through the AKNS method, the $n$-th member of the NLS hierarchy is obtained by
enforcing the compatibility of a pair of ordinary differential equations, the
$n$-th Lax Pair. \te{The first equation of the pair is $\chi_{\tau_1} = T_1\chi$ and
the second is $\chi_{\tau_n} = T_n \chi$, for the $n$-th member of the
hierarchy. Here, $T_1$ and $T_n$ are $2\times 2$ matrices, the first five of
which are defined in \eqref{eqn:linearHierarchy}. } The $n$-th member of the NLS
hierarchy is recovered by requiring $\partial_{\tau_n}\chi_{\tau_1} =
\partial_{\tau_1}\chi_{\tau_n}$.  For example, \eqnref{eqn:NLS} is recovered
from the compatibility condition of $\chi_{\tau_1}$ and $\chi_{\tau_2}$ with $t
= \tau_2$. We call the collection of the Lax equations for the hierarchy the
linear NLS hierarchy. The first five members of the linear NLS hierarchy are
\begin{subequations}
\begin{align}
  \chi_{\tau_0} &= T_0 \chi = \begin{pmatrix}  -i & 0\\0 &
i\end{pmatrix}\chi \label{eqn:linearHierarchy0},\\
\chi_{\tau_1} &= T_1 \chi =\begin{pmatrix}
                            -i\zeta & \Psi\\ -\Psi\cc & i\zeta
                        \end{pmatrix}\chi \label{eqn:linearHierarchy1},\\
\chi_{\tau_2} &= T_2\chi = \begin{pmatrix}
  -i\zeta^2 + i\abs{\Psi}^2/2 & \zeta\Psi + i\Psi_x/2\\
  -\zeta\Psi\cc + i\Psi_x\cc/2 & i\zeta^2 - i\abs{\Psi}^2/2
              \end{pmatrix}\chi\label{eqn:linearHierarchy2},\\
\chi_{\tau_3} &= T_3 \chi = \begin{pmatrix}
 -i \zeta^3 + i \zeta \abs{\Psi}^2/2 + i\IM\lt(\Psi \Psi_x\cc\rt)/2 &
 \zeta^2 \Psi + i \zeta \Psi_x /2- \abs{\Psi}^2 \Psi /2 - \Psi_{xx}/4\\
 -\zeta^2 \Psi\cc + i\zeta \Psi_x\cc/2 + \abs{\Psi}^2\Psi\cc/2 + \Psi_{xx}\cc/4 &
 i\zeta^3 - i\zeta \abs{\Psi}^2/2 - i \IM\lt(\Psi \Psi_x\cc\rt)/2
 \end{pmatrix}\chi\label{eqn:linearHierarchy3},\\
  \begin{split}
   \chi_{\tau_4} &= T_4 \chi = \begin{pmatrix} N_1 & N_2 \\ N_3 & -N_1
   \end{pmatrix}\chi,
 \end{split}
\end{align}
  \label{eqn:linearHierarchy}
\end{subequations}
where
\begin{subequations}
\begin{align}
    &  N_1 = -i\zeta^4 + i\zeta^2 \abs{\Psi}^2/2 + i\zeta\IM(\Psi\Psi_x\cc)/2 -
    3i\abs{\Psi}^4/8 + i\abs{\Psi_x}^2/8 -i \RE(\Psi\cc \Psi_{xx})/4,\\
    & N_2 = \zeta^3 \Psi + i\zeta^2 \Psi_x/2 -\zeta \lt(\abs{\Psi}^2 \Psi/2 +
   \Psi_{xx}/4\rt) -3 i \abs{\Psi}\Psi_x/4 - i\Psi_{xxx}/8,\\
    & N_3 = -\zeta^3 \Psi\cc + i\zeta^2 \Psi_x\cc/2 + \zeta \lt(\abs{\Psi}^2
   \Psi\cc/2 + \Psi_{xx}\cc/4\rt) - 3 i \abs{\Psi}\Psi_x\cc/4 -
   i\Psi_{xxx}\cc/8,
 \end{align}
\end{subequations}
and $\zeta$ is referred to as the Lax parameter.

Each of the $H_n$ are mutually in involution under the canonical Poisson
bracket \eqref{eqn:J} \cite{SolitonsAndIST}. As a result, the flows of all members of the NLS
hierarchy commute and any linear combination of the conserved quantities gives
rise to a dynamical equation whose flow commutes with all equations of the
hierarchy. We define a family of evolution equations in $t_n$ by
\begin{align}
  \frac{\partial}{\partial_{t_n}} \begin{pmatrix} r \\ \ell \end{pmatrix} = J \hat H_n' (r,\ell)
= J \lt(H_n' + \sum_{j=0}^{n-1} c_{n,j} H_j'\rt), \quad n\geq
0\label{eqn:nthNLS},
\end{align}
where the $c_{n,j}$ are constants. We loosely call \eqref{eqn:nthNLS} the
``$n$-th NLS equation.'' Similarly we define the $n$-th linear NLS equation to
be
\begin{align}
  \chi_{t_n} = \hat T_n \chi = \lt(T_n + \sum_{j=0}^{n-1} c_{n,j} T_j\rt)\chi.
  \label{eqn:nthLinearNLS}
\end{align}
The $n$-th NLS equation is obtained by enforcing the compatibility of
$\chi_{\tau_1}$ with $\chi_{t_n}$.

The second NLS equation \eqref{eqn:stationaryPDE} is obtained from
\eqref{eqn:tau0dynamics} and \eqref{eqn:tau2dynamics} and has Hamiltonian \mbox{$\hat
H~=~\hat H_2~=~H_2 - \omega H_0 / 2$.} With $\psi(x,t) = (r(x,t) + i
\ell(x,t))/\sqrt{2}$, \eqref{eqn:stationaryPDE} is
\begin{align}
  \partial_t \begin{pmatrix} r \\ \ell \end{pmatrix} &=
    \begin{pmatrix} - \omega
    \ell - \ell(r^2 + \ell^2)/2 -  \ell_{xx}/2 \\ \omega r + r(\ell^2 + r^2)/2
    +  r_{xx}/2\end{pmatrix} =
  J \hat H'(r,\ell).
\label{eqn:NLSHamiltonianDynamics}
\end{align}
The associated linear NLS equation is $\hat T_2 = T_2 - \omega T_0/2$. Defining
$\tau_1 = x$ and $t_2=t$, \eqref{eqn:NLSHamiltonianDynamics} (or
equivalently~\eqref{eqn:stationaryPDE}) is obtained via the compatibility
condition of the two matrix equations
\begin{subequations}
\begin{align}
\chi_x  &= \chi_{\tau_1} = T_1 \chi\label{eqn:chi_x_general},\\
  \chi_t &= \chi_{\tau_2} - \frac\omega2 \chi_{\tau_0} = \lt(T_2 -
  \frac{\omega}{2}T_0\rt) \chi \label{eqn:chi_t_general}.
\end{align}\label{eqn:LaxPairRepr_general}
\end{subequations}
}

\section{Proofs of some lemmas}\label{Appendix:proofs}
In this appendix we present proofs for lemmas used in Section
  \ref{section:EllipticSpectralStability}.

\begin{proof}[\textbf{Proof of Lemma~\ref{lem:M0dn}}]
  Formulae for Weierstrass Elliptic functions used here and in what follows
  are in \cite[Chapter 23]{NIST:DLMF} \cite{byrdFriedman, whittakerWatson} . We use the
  notation $\eta_k=\zeta_w(\omega_k)$, $k=1,2,3.$

  For the dn solutions, $b=1$ and the four roots of $\Omega(\zeta)$ are
  \begin{align}
    \zeta_1 &= \frac i2 (1-\sqrt{1-k^2}), \quad \zeta_2 = \frac i2 (1 +
    \sqrt{1-k^2}), \quad \zeta_3 = -\zeta_2, \quad \zeta_4 = -\zeta_1.
  \end{align}
  Since $c=\theta = 0$,
  \begin{align}
    M(\zeta_j) &= -2i I(\zeta_j)\mod 2\pi.
  \end{align}
  The quantities $\alpha(\zeta_j), ~\wp'(\alpha(\zeta_j)),$ and $
  \zeta_w(\alpha(\zeta_j))$ are needed for the computation of
  $I(\zeta_j)$.
  Using \eqref{eqn:e0e1e2e3} and \eqref{eqn:alpha},
  \begin{subequations}
  \begin{align}
    \alpha(\zeta_2) = \alpha(\zeta_3) &= \wp^{-1}\lt( e_1 +
    \sqrt{(e_1-e_3)(e_1-e_2)}\rt) = \sigma_1 \frac{\omega_1}{2} + 2n\omega_1 +
    2m\omega_3,\\
    \alpha(\zeta_1) = \alpha(\zeta_4) &= \wp^{-1}\lt(e_3 + \frac{(e_3 -
    e_1)(e_3-e_2)}{\wp(\omega_1/2)-e_3}\rt) =\sigma_1\lt( \frac{\omega_1}{2} -
    \omega_3\rt) + 2n\omega_1 + 2m\omega_3,
  \end{align}
\end{subequations}
  where $n, m \in \Z$ and $\sigma_1$ is either $\pm 1$. From \cite[equation
  1033.04]{byrdFriedman} and the addition formula for $\wp'(z)$,
  \begin{subequations}
  \begin{align}
    \wp'(\omega_1/2) &= -2\lt((e_1-e_3)\sqrt{e_1-e_2} +
    (e_1-e_2)\sqrt{e_1-e_3}\rt) =
    -2(1-k^2+\sqrt{1-k^2})\label{eqn:wp'(omega1/2)},\\
    \wp'(\omega_1/2 - \omega_3) &= 2\lt((e_1-e_3) \sqrt{e_1-e_2} -
    (e_1-e_2)\sqrt{e_1-e_3}\rt) = -2(1-k^2 - \sqrt{1-k^2}).
  \end{align}
  \end{subequations}
  Using the addition formula for $\zeta_w(z)$,
  \begin{align}
    \zeta_w(\omega_1/2) = \zeta_w(-\omega_1/2 + \omega_1) =
    -\zeta_w(\omega_1/2) + \eta_1 - \frac12
    \frac{\wp'(\omega_1/2)}{\wp(\omega_1/2)-e_1},
  \end{align}
  so that
  \begin{subequations}
  \begin{align}
    \zeta_w(\omega_1/2) &= \frac12 \lt(\eta_1 -\frac12
    \frac{\wp'(\omega_1/2)}{\wp(\omega_1/2) - e_1}\rt)  =
    \frac12\lt(\eta_1 + 1 + \sqrt{1-k^2}\rt),\\
    \zeta_w(\omega_1/2 - \omega_3) &= \zeta_w(\omega_1/2) - \eta_3
    + \frac12\frac{\wp'(\omega_1/2)}{\wp(\omega_1/2) - e_3} =
    \frac12\lt(\eta_1 + 1  - \sqrt{1-k^2}\rt).
  \end{align}
  \end{subequations}
  Using the parity and periodicity of $\wp'(z)$, and the quasi-periodicity of
  $\zeta_w(z)$ we arrive at
  \begin{subequations}
  \begin{align}
    \wp'(\alpha(\zeta_2)) &= \sigma_1\wp'(\omega_1/2), \\
    \wp'(\alpha(\zeta_1)) &= \sigma_1\wp'(\omega_1/2 - \omega_3),\\
    \zeta_w(\alpha(\zeta_2)) &= \sigma_1\zeta_w(\omega_1/2) +
    2n\eta_1 + 2m\eta_3,\\
    \zeta_w(\alpha(\zeta_1)) &= \sigma_1\zeta_w(\omega_1/2-\omega_3) +
    2n\eta_1 + 2m\eta_3.
  \end{align}
  \end{subequations}
  Substituting the above quantities into \eqref{eqn:Izeta} and using
  $\omega_3 \eta_1 -\omega_1 \eta_3 =i\pi/2$ results in
  $I(\zeta_j) = 0 \mod 2\pi$ for $j=1,2,3,4.$
\end{proof}

\begin{proof}[\textbf{Proof of Lemma~\ref{lem:dnRoots}}]
    Let $\zeta = i\xi$ with $\xi \in \R$. Then
    \begin{align}
      \Omega^2(\zeta) &= -\xi^4 - \frac{1}{2}(k^2-2) \xi^2 - \frac{k^4}{16} \in \R,
    \end{align}
    so $\Omega(\zeta)$ is either real or imaginary. Then $\Omega(\zeta)\in i\R$
    if and only if
    \begin{align}
        \xi^2 \geq \frac{1}{4}(2-k^2) + \frac12\sqrt{1-k^2} \quad \text{ or
        }\quad \xi^2 \leq \frac{1}{4}(2-k^2) - \frac12\sqrt{1-k^2},
    \end{align}
    which is equivalent to
    \begin{align}
      \abs{\xi}\leq \IM(\zeta_1) \quad \text{ or }\quad \abs{\xi}\geq
      \IM(\zeta_2).\label{eqn:dnSpecialZetaValues}
    \end{align}
\end{proof}

\begin{proof}[\textbf{Proof of Lemma~\ref{lem:cnMuOnImaginaryAxis}}]
    First, this holds for $\zeta = 0$, since
    \begin{align}
        \alpha(0) &= \wp^{-1} (e_3) = \omega_3 + 2n\omega_1 + 2m\omega_2,
    \end{align}
    where $m,n\in \Z$, so that
    \begin{align}
        I(0) &= 2\Gamma(\omega_1(\eta_3 + 2n\eta_1 + 2m\eta_3) - \eta_1(\omega_3
        + 2n\eta_1 +2m\eta_3) = -\Gamma p\pi i,
    \end{align}
    for $p\in \Z$ \cite[Chapter 23]{NIST:DLMF}. Then
    \begin{align}
      M(\zeta) &= -2i(-\Gamma p\pi i) + \pi = \pi\mod{2\pi}.
    \end{align}
    Since the curves for $\RE(I)=\text{constant}$, given by
    \eqref{eqn:LaxSpectrumTangents}, and for $\IM(I)=\text{constant}$ are
    orthogonal, the vector field for $\IM(I)=\text{constant}$ is vertical on the
    imaginary axis as $\Omega(\zeta)\in i\R$ there
    \eqref{eqn:LaxSpectrumTangents}. Since $M(\zeta) =\pi \mod 2\pi$ at $\zeta =
    0$ and is constant on the imaginary axis, it follows that $M(\zeta) = \pi
    \mod{2\pi}$ on the imaginary axis.
\end{proof}

\begin{proof}[\textbf{Proof of Lemma~\ref{lem:M0cn}}]
  For the cn solutions, $b=k^2$ and the four roots of $\Omega(\zeta)$ are
\begin{align}
 \zeta_1 = \frac12 \lt(\sqrt{1-k^2} + i k\rt), \quad \zeta_2 = \frac12
  \lt(-\sqrt{1-k^2} + i k\rt), \quad \zeta_3 = -\zeta_1, \quad \zeta_4 = -\zeta_2.
\end{align}
  Here, $c=0$ and $\theta(T(k)) = \pi$ give
  \begin{align}
    M(\zeta_j) = -2i I(\zeta_j)  + \pi \mod 2\pi.
  \end{align}
  The quantities $\alpha(\zeta_j), ~\wp'(\alpha(\zeta_j)),$ and
  $\zeta_w(\alpha(\zeta_j))$ are needed. Using
  \eqref{eqn:e0e1e2e3} and \eqref{eqn:alpha},
  \begin{subequations}
  \begin{align}
    \alpha(\zeta_1) &= \alpha(\zeta_3) = \wp^{-1}\lt( e_2 -
    i\sqrt{(e_1-e_2)(e_2-e_3)}\rt) = \sigma_1\frac{\omega_2}{2} + 2n\omega_1 +
    2m\omega_3,\\
    \alpha(\zeta_2) &= \alpha(\zeta_3) = \wp^{-1}\lt( e_3 +
    \frac{(e_3-e_1)(e_3-e_2)}{e_2-e_3-i\sqrt{(e_1-e_2)(e_2-e_3)}}\rt)
    =\sigma_1
    \lt(\frac{\omega_2}{2}-\omega_3\rt) + 2n\omega_1 + 2m\omega_3.
  \end{align}
  \end{subequations}
  From \cite[equation 1033.04]{byrdFriedman} and the addition formula for
  $\wp'(z)$,
  \begin{subequations}
  \begin{align}
    \begin{split}
      \wp'(\omega_2/2) &= -\wp'(\omega_1/2 + \omega_3/2) =
    -2\lt((e_1-e_2)\sqrt{e_2-e_3} + i(e_2-e_3)\sqrt{e_1-e_2}\rt) \\
      &=-2k(1-k^2
    + ik \sqrt{1-k^2}),
    \end{split}\\
    \wp'(\omega_2/2 - \omega_3) &= -2k(1-k^2 - ik\sqrt{1-k^2}).
  \end{align}
  \end{subequations}
  $\zeta_w(\omega_2/2)$ is found in a similar manner to
  $\zeta_w(\omega_1/2)$ (Lemma~\ref{lem:M0dn}) to be
  \begin{align}
    \zeta_w(\omega_2/2) &= \frac12 \lt(\zeta_w(\omega_2) - k + i
    \sqrt{1-k^2}\rt),
  \end{align}
  from which
  \begin{align}
    \zeta_w(\omega_2/2- \omega_3) &= \frac12 \lt(\zeta_w(\omega_2) - k - i
    \sqrt{1-k^2}\rt) - \eta_3.
  \end{align}
  Using the parity and periodicity of $\wp'(z)$, and the quasi-periodicity of
  $\zeta_w(z)$ we arrive at
  \begin{subequations}
  \begin{align}
    \wp'(\alpha(\zeta_1)) &= \sigma_1\wp'(\omega_2/2), \\
    \wp'(\alpha(\zeta_2)) &= \sigma_1\wp'(\omega_2/2 - \omega_3),\\
    \zeta_w(\alpha(\zeta_1)) &= \sigma_1\zeta_w(\omega_2/2) +
    2n\eta_1 + 2m\eta_3,\\
    \zeta_w(\alpha(\zeta_2)) &= \sigma_1\zeta_w(\omega_2/2-\omega_3) +
    2n\eta_1 + 2m\eta_3,
  \end{align}
  \end{subequations}
  where $\sigma_1$ is either $\pm1$.  Substituting the above quantities
  into \eqref{eqn:Izeta} results in
  \begin{subequations}
  \begin{align}
    I(\zeta_1) &= I(\zeta_3) = \sigma_1 \frac{i\pi}{2} + 2\pi m,\\
    I(\zeta_2) &= I(\zeta_4) = 3\sigma_1 \frac{i\pi}{2} + 2\pi m.
  \end{align}
  \end{subequations}
  Therefore
  \begin{subequations}
  \begin{align}
    M(\zeta_1) &= M(\zeta_3) = \sigma_1 \pi + 4\pi m + \pi = 0 \mod 2\pi, \\
    M(\zeta_2) &= M(\zeta_4) = 3\sigma_1 \pi + 4\pi m + \pi = 0 \mod 2\pi.
  \end{align}
  \end{subequations}

\end{proof}

\begin{proof}[\textbf{Proof of Lemma~\ref{lem:noCrossingscn}}]
  Without loss of generality, let $\zeta = \zeta_r+ i \zeta_i$ with $\zeta_r<0$.
  The computation is the same for $\zeta_r>0$ by symmetry of the Lax spectrum.
  Consider the curve in the left half plane defined by $\IM(\Omega^2)=0,~
  \RE(\Omega^2<0)$ \eqref{eqn:ReImOmega}. For $\zeta_i\neq 0$, this curve is
  defined by
  \begin{align} \zeta_i^2 =Q(\zeta_r)= \zeta_r^2 -
    \frac14(1-2k^2)\label{eqn:cnParamCurve}, \qquad \text{for}\qquad  \zeta_r
    \in [-\sqrt{1-k^2}/2,0).
  \end{align}
  The above parameterization is valid only when $k\geq 1/\sqrt{2}$. For
  $k<1/\sqrt{2}$, $\zeta_r$ is restricted to a smaller range so that
  $\zeta_i\in \R$.

  Let $G(\zeta_r) = I(\zeta_r + i\zeta_i(\zeta_r))$ where $\zeta_i(\zeta_r)$ is
  defined with either sign of the square root in \eqref{eqn:cnParamCurve}. If we
  can show that $\RE(G(\zeta_r))>0$ for $\zeta_r \in (-\sqrt{1-k^2}/2,0),$ then
  we have shown that $\RE(I(\zeta))\neq 0$ when $\Omega(\zeta)\in
  i\R\setminus\{0\}$. We compute
  \begin{align}
    4\Omega_i \sqrt{Q(\zeta_r)} \DD{\RE(G)}{\zeta_r} = \zeta_r
    P_2(\zeta_r),
  \end{align}
  where
  \begin{align} P_2(\zeta_r) &:= -16 K(k)
    \zeta_r^2 + 4(E(k) - k^2K(k)),
  \end{align}
  and
  \begin{align}
    \Omega_i &:= \frac12\sqrt{(4\zeta_r^2+k^2-1)(k^2+4\zeta_r^2)},
  \end{align}
  the imaginary part of $\Omega$. Here we take $\Omega_i\sqrt{Q(\zeta_r)}>0$
  without loss of generality ($\Omega_i\sqrt{Q(\zeta_r)}<0$ corresponds to a
  different sign for $\zeta_i$ or $\Omega_i$ or both and is a nontrivial but
  straightforward extension of what follows). $P_2(\zeta_r)$ and $\d \RE(G)/\d
  \zeta_r$ have opposite signs since $\zeta_r<0$. Since
  $\RE(G(-\sqrt{1-k^2}/2))=0$ and $P_2(-\sqrt{1-k^2}/2)<0$, it suffices to show
  that $\d \RE(G)/\d\zeta_r>0$. Indeed, if this is true, then $\RE(G)>0$ when
  $\Omega(\zeta)\in i\R\setminus\{0\}$. There are three cases to consider.

  \begin{enumerate}[font={\bfseries}]
  \item \textbf{Case 1: $\bm{P_2(\zeta_r)}$ has no negative roots or one root at
    $\bm{\zeta_r=0}$.}

      If $P_2(\zeta_r)$ is always negative, then we are done since $\RE(G)$ is
      increasing on $(-\sqrt{1-k^2}/2,0)$. This is the case if $E(k)
      -k^2K(k)\leq 0$, which is true for $k \geq \kappa$ where $\kappa \approx
      0.799879$.

    \item \textbf{Case 2: $\bm{P_2(\zeta_r)}$ has one negative root and
      $\bm{Q(\zeta_r)}$ has no negative roots or a double negative root.}

      Let $\hat\zeta$ be such that $P_2(\hat\zeta)=0$. Then $\RE(G)$ is
      increasing on $(-\sqrt{1-k^2}/2,\hat\zeta)$ and decreasing on
      $(\hat\zeta,0)$. This can only occur for $1/\sqrt{2}< k<\kappa$. Since
      \begin{align}
        \DD{\RE(I)}{\zeta_i} = -\IM\lt(\DD{I}{\zeta}\rt),
      \end{align}
      $\d\RE(I)/\d\zeta_i>0$ for $\zeta = i\zeta_i$.  Since $\RE(G(\zeta_r))$
      must be minimized in the limit $\zeta_r\to 0^-$, it follows from
      continuity and the fact that $\RE(G)>0$ on the imaginary axis that
      $\RE(G)>0$ for $\zeta_r\in(-\sqrt{1-k^2}/2,0)$.

    \item \textbf{Case 3: $\bm{P_2(\zeta_r)}$ and $\bm{Q(\zeta_r)}$ both have
    one negative root.}

    Let $\hat \zeta$ be as above and let $\hat \xi$ be the negative root of $Q$.
    Then $\RE(G)$ is increasing on $(-\sqrt{1-k^2}/2,\hat\zeta)$ and decreasing
    on $(\hat\zeta,\hat\xi)$ at which $\RE(G(\hat\xi))=0$. Since the
    parameterization is not valid on $(\hat\xi,0)$, $\RE(G)>0$ for
    $\zeta_r\in(-\sqrt{1-k^2}/2,\hat\xi)$ which are all allowed $\zeta$ values
    for which $\zeta\not\in\R\cup i\R$.
  \end{enumerate}
  It follows that $\RE(G)>0$ when $\Omega(\zeta)\in i\R\setminus\{0\}$.
\end{proof}

\begin{proof}[\textbf{Proof of Lemma~\ref{lem:0atroots}}]
  We establish that $M_j = 0 \mod 2\pi$ on the boundary of the parameter space
  by establishing this fact for the Stokes waves ($k=0$) and using Lemmas~
  \ref{lem:M0dn} and \ref{lem:M0cn}.

  Setting $\lam = 0$ in \eqref{eqn:StokesEvals} shows that $\mu = -2n$. Since
  $T(k) = \pi$ for Stokes waves, $T(k)\mu = 0 \mod 2\pi$ whenever
  $\Omega = 0$. Next, we compute directly that $\partial_b M_j =0$ for the nontrivial-phase
  solutions. In what follows we use that
  \begin{equation}
    \zeta_j = \frac12 \lt(\sigma_1 \sqrt{1-b} + i \sigma_2 \lt(\sqrt{b} -
    \sigma_1 \sqrt{b-k^2}\rt)\rt),
  \end{equation}
  so that $\zeta_1,~\zeta_2,~\zeta_3,$ and $\zeta_4$ correspond to
  $(\sigma_1,\sigma_2) = (1,1),~(-1,1),~(-1,-1),~(1,-1)$ respectively. We define
  \begin{align}
    e_{p,j} &= \wp(\alpha_j) - e_0 = -2\zeta_j^2 + \omega,
  \end{align}
  where $e_0$ is defined in \eqref{eqn:alpha0}, and use
  \begin{subequations}
  \begin{align}
    \pd{\zeta_j}{b} &= \frac{e_{p,j}}{4c},\\
    \pd{\alpha_0}{b} &= \frac{1}{\wp'(\alpha_0)} = -\frac{i}{2c},\\
    \pd{\alpha_j}{b} &= -\frac{c + 2\zeta_j e_{p,j}}{2c\wp'(\alpha_j)} =
    \frac{4\zeta_j^3 - 2\zeta_j\omega -c}{2c\wp'(\alpha_j)}.
  \end{align}
  \end{subequations}
  From the definition of $\Gamma$ and \eqref{eqn:wpalphasquared},
  \begin{align}
    \frac{(4\zeta_j^3 - 2\zeta_j\omega - c)\Gamma}{\wp'(\alpha_j)} &= \frac{2i (4\zeta_j^3 - 2\zeta_j
    \omega-c)^2}{\wp'(\alpha_j)^2} = -\frac{i}{2}.
  \end{align}
  Using the above calculations, the expression \eqref{eqn:IzetaGamma}, and
  \eqref{eqn:muTk} with $\theta(T(k))$ defined in \eqref{eqn:thetaW}, we compute
\begin{align}
  \begin{split}
  \pd{}{b} M_j &= -2i \lt( \frac{\partial I(\zeta_j)}{\partial b} +
  \pd{}{b}(\alpha_0 \eta_1 - \omega_1 \zeta_w(\alpha_0)) \rt)\\
  &= -2i \lt( -\frac{i}{2c} e_{p,j}\omega_1 - \frac{(4\zeta_j^3 - 2\zeta_j\omega
  - c)\Gamma}{c\wp'(\alpha_j)}(\eta_1 + \omega_1(e_{p,j} + e_0)) -
  \frac{i}{2c}(\eta_1 + \omega_1e_0)\rt) = 0
  \end{split}
\end{align}
  by direct computation. Since $M_j = 0\mod 2\pi$ along the boundaries of the
  parameter region (Figure \ref{fig:parameterSpace}) and $\partial_b M_j = 0$ on
  the interior of the parameter space, it follows that $M_j$ is constant $(0
  \mod 2\pi)$ in the whole parameter space.

  \end{proof}

\section{Necessity of stability condition (\ref{eqn:mainStabilityBound}), proof
of Lemma~\ref{lem:NoDoublePointsSpine}, and proof of Theorem
\ref{thm:imaginaryNTPStabilityResult}} \label{Appendix}

In this appendix we present progress made towards showing that
\eqref{eqn:mainStabilityBound} is not only a sufficient but also a necessary
condition for \te{spectral} stability. We introduce a theorem which shows that
$\abs{\RE(\lam)}>0$ on the complex bands of the spectrum. For part of the
parameter space, the proof of this theorem is complete.  For a different part of
parameter space, the proof relies upon a numerical check over a bounded region
of parameter space (see Figure \ref{fig:whenToDoNumericalCheck}). The numerical
check consists of finding a root of a degree-six polynomial and evaluating
Weierstrass elliptic functions at that root. Numerical checks of this kind are
not uncommon (see \eg the non-degeneracy condition for focusing NLS in
\cite{gallay_haragus_OrbitalStability}). We use similar arguments as used in
Lemma~\ref{lem:noCrossings} to prove Theorem
\ref{thm:imaginaryNTPStabilityResult} and Lemma~\ref{lem:NoDoublePointsSpine}.

\begin{lemma}\label{lem:noCrossings}
Let $c\neq 0$ and $\zeta\in (\C^-\cap\sigma_L)\setminus(\R\cup i\R \cup
\{\zeta_2,\zeta_3\})$ where $\C^-$ is the left half plane. Then
$\Omega(\zeta)\notin i\R$.
\end{lemma}

\begin{proof}
    Let $c\neq0$ and $\zeta = \zeta_r + i\zeta_i$ with $\zeta_r<0$. Consider the
    curve in the left half plane defined by $\IM(\Omega^2) = 0$. For
    $\zeta_i\neq0,$ this curve is defined by
    \begin{align}
      \zeta_i^2 = \zeta_r^2 - \frac{\omega}2 -
      \frac{c}{4\zeta_r}\label{eqn:parameterizedOmegaLT0}.
    \end{align}
    The condition $\RE(\Omega^2)\leq 0$ implies $\abs{\zeta_r}\leq
    \sqrt{1-b}/2$ with equality attained at the roots of $\Omega^2$. Let
    \begin{align}
      Q(\zeta_r) := 4\zeta_r^3 - 2\omega\zeta_r - c.\label{eqn:Qzr}
    \end{align}
    We have that $\zeta_i\in \R$ only if $Q(\zeta_r) \leq 0$ and $Q(\zeta_r)$
    has two roots with negative real part. If both roots are complex or there is
    a double root, then the parameterization \eqref{eqn:parameterizedOmegaLT0}
    is valid for all $-\sqrt{1-b}/2\leq \zeta_r<0$. This is the case if the
    discriminant of $Q$ is nonpositive, which is true when
    \begin{align}
      b \geq \begin{cases}
        k^2, & k>1/\sqrt{2},\\
        F(k), & k\leq 1/\sqrt{2},
      \end{cases}\label{eqn:Qdiscriminant}
    \end{align}
    with
    \begin{align}
      F(k) := \frac{(1+k^2)^3}{9(1-k^2+k^4)}.\label{eqn:bltFk}
    \end{align}
    It is interesting to note that the condition $b<F(k)$ is the same condition
    as \cite[equation (85)]{deconinck2017stability} which determines when the
    imaginary $\Omega$ axis is quadruple covered by the map $\Omega(\zeta)$.

    Define $G(\zeta_r) = I(\zeta_r + i\zeta_i(\zeta_r)),$ where
    $\zeta_i(\zeta_r)$ is defined with either sign of the square root in
    \eqref{eqn:parameterizedOmegaLT0}. The goal is to show that $\RE G(\zeta_r)
    = 0$ only when $\zeta_i=0$ or $\zeta_r = -\sqrt{1-b}/2$, which corresponds
    to one of the roots of $\Omega^2$. Along the solutions of
    \eqref{eqn:parameterizedOmegaLT0},
    \begin{align}\label{eqn:derivativeRelation}
      \Omega_i \zeta_r\sqrt{Q(\zeta_r)}\DD{\RE G}{\zeta_r} = P_6(\zeta_r),
    \end{align}
    where
    \begin{align}
      \Omega_i = \pm \frac{1}{4\abs{\zeta_r}} \sqrt{(4\zeta_r^2 +
      b-1)(b+4\zeta_r^2)(b-k^2+4\zeta_r^2)},
    \end{align}
    the imaginary part of $\Omega$ (and $\Omega = \Omega_i$
    because of the parameterization). The polynomial $P_6$ is given by
    \begin{align}
      P_6(x) = -64 K(k) x^6 + 16(E(k) + (k^2-2b)K(k))x^4 + 8cK(k)x^3 + 2c
      (E(k) + (b-1)K(k)) x - c^2K(k).
    \end{align}
    We let $\Omega_i\sqrt{Q(\zeta_r)}>0$, without loss of generality.
    ($\Omega_i\sqrt{Q(\zeta_r)}<0$ corresponds to a different sign for $\zeta_i$
    or $\Omega_i$ or both and is a nontrivial but straightforward extension of
    what follows). Therefore, $P_6(\zeta_r)$ has the opposite sign of
    $\d\RE(G)/\d\zeta_r$ and $\RE(G(\zeta_r))\to+\infty$ as $\zeta_r\to
    0^-$ since $\Omega_i\zeta_r\sqrt{Q(\zeta_r)}\to 0^-$ and $P_6\to-c^2K(k)<0$.
    Since $\zeta_r = -\sqrt{1-b}/2$ corresponds to a root of $\Omega$ and the
    roots of $\Omega$ are in the Lax spectrum, $\RE G(-\sqrt{1-b}/2) = 0$.  We
    wish to show that $\d\RE(G)/\d\zeta_r\geq 0$, which guarantees that
    $\RE(G(\zeta_r))=0$ only when $\Omega(\zeta_r)=0$.

    Consider the polynomial
    \begin{align}
      \tilde P_6(x) = P_6(-x) = a_6 x^6 + a_4x^4 + a_3x^3 + a_1 x +
      a_0.\label{eqn:descartes}
    \end{align}
    It is clear that $a_6<0, a_3<0, a_0<0$ and $a_4$ changes sign depending on
    $b$ and $k$. We have
    \begin{align}
      a_1 = -2c\lt(E(k) + (b-1)K(k)\rt) \leq -2c\lt(E(k) +
      (k^2-1)K(k)\rt) = -2c\DD{K(k)}{k} \leq 0.
    \end{align}
    By Descartes' sign rule, an upper bound on the number of negative roots of
    $P_6$ is either 2 or 0, depending on the sign of $a_4$. Since
    $P_6(\zeta_r)\to -\infty$ as $\zeta_r\to -\infty$ and $P_6(0) <0$,
    $P_6(\zeta_r)$ has an even number of negative roots, either 2 or 0.

    We consider four cases.
    \begin{enumerate}[font={\bfseries}]
      \item \textbf{Case 1: $\bm{P_6(\zeta_r)}$ has no negative roots or a
      double negative root.}

        If $P_6(\zeta_r)$ has no negative roots or a double negative root, then
        $P_6(\zeta_r)\leq 0$ and $\RE(G(\zeta_r))>0$ so  $\RE(G(\zeta_r))$ is
        bounded away from 0 (see Figure \ref{fig:Case1}).

      \item \textbf{Case 2: $\bm{P_6(\zeta_r)}$ has two distinct negative roots,
        $\bm{Q(\zeta_r)}$ has no negative roots.}

        Let $\xi_1$ and $\xi_2$ be the two roots of $P_6$ with $\xi_1<\xi_2<0$
        (see Figure \ref{fig:Case2}).  Then $\RE(G)$ is increasing on
        $(-\sqrt{1-b}/2, \xi_1)$, decreasing on $(\xi_1,\xi_2)$, and increasing
        again on $(\xi_2,0)$. If $\RE(G( \xi_2))>0$, then $\RE(G)$ is bounded
        away from $0$ and we are done.  We do not know how to verify this
        condition analytically, so we check it numerically. It is found to
        always hold.

      \item \textbf{Case 3: $\bm{P_6(\zeta_r)}$ has two distinct negative roots,
        $\bm{Q(\zeta_r)}$ has a double negative root.}

        Let $\xi_1$ and $\xi_2$ be as above and let $\zeta_1$ be the negative
        double root of $Q$ (see Figure \ref{fig:Case3}). It must be the case
        that $\zeta_1>\xi_1$ since $\RE(G)$ is initially increasing and we know
        that $\RE(G)\to 0$ as $\zeta \to \zeta_1$.  However, since $\zeta_1$ is
        a double root of $Q$, it is also a root of $\RE(G)$ so it must be that
        $\zeta_1 = \xi_2$. This means that $\RE(G)$ is tangent to 0 at $\zeta =
        \zeta_1$. This corresponds to $\zeta \in \R$.
      \item \textbf{Case 4: $\bm{P_6(\zeta_r)}$ and $\bm{Q(\zeta_r)}$ have two
      distinct negative roots}

        Let $\xi_1$ and $\xi_2$ be as before and let $\zeta_1$ and $\zeta_2$ be the two
        negative roots of $Q$ with $\zeta_1<\zeta_2$. As before, it must be that
        $\xi_1$ is smaller than each of $\xi_2,~\zeta_1$, and $\zeta_2$. The
        next largest root may be either $\xi_2$ or $\zeta_1$.

        \begin{itemize}
          \item An illustration of this case is found in Figure
            \ref{fig:Case4Option1}. If $\xi_2$ is the next largest root, then
            there is a $\hat \zeta \in (\xi_1,\xi_2)$ such that
            $\RE(G(\hat\zeta))=0$. For $\zeta_r$ greater than $\xi_2$, $\RE(G)$
            increases to $0$ at $\zeta_r = \zeta_1$. For $\zeta \in
            (\zeta_1,\zeta_2)$, nothing can be said about $\RE(G)$ since the
            parameterization is not valid. For $\zeta \in (\zeta_1,0)$,
            $\RE(G)>0$ is increasing since $P_6(\zeta)<0$ in this range.  Thus
            if the ordering is $\xi_1<\xi_2<\zeta_1<\zeta_2$, there is a $\hat
            \zeta \in \sigma_L$ such that $\RE(G(\hat \zeta))=0$ and
            $\Omega(\hat \zeta) \in i\R$.

            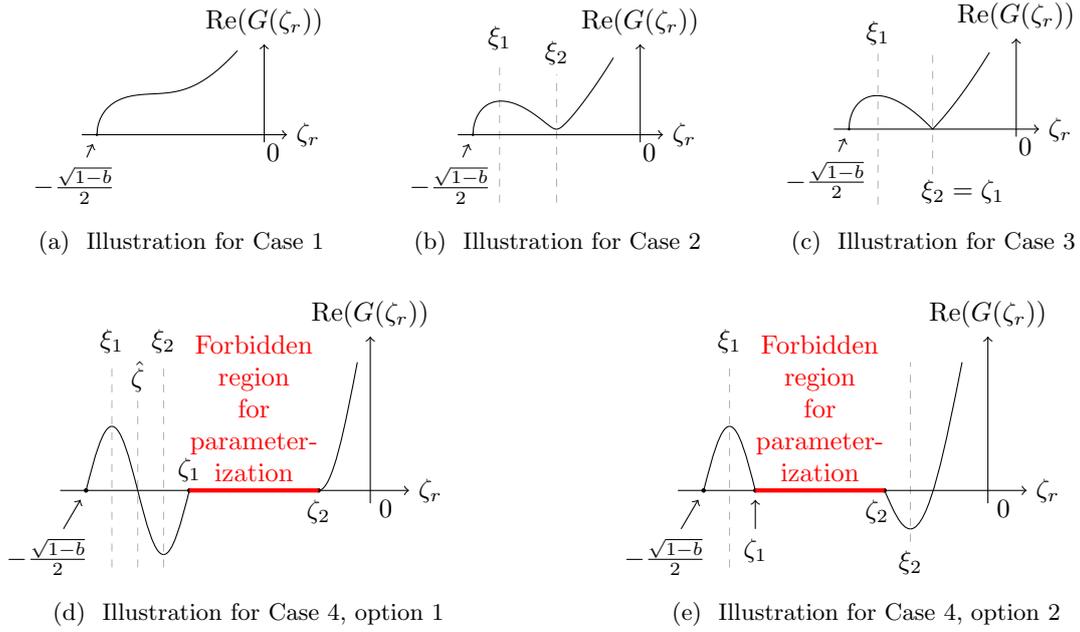
\begin{figure}
              \begin{subfigure}[b]{0.25\textwidth}
                \begin{tikzpicture}[scale=1]
                    \node[below] (origin) at (0.12,0) {0};
                    \draw[->] (-2.4,0)--(0.3,0) node[right] {$\zeta_r$};
                    \draw[->] (0,-0.1) -- (0,1.2) node[above] {$\RE(G(\zeta_r))$};
                    \node (LHSanchor) at (-2.2,0) {};
                    \fill (LHSanchor) circle [radius=0.5pt];
                    \node[below] (LHS) at (-2.5,-0.3) {$-\frac{\sqrt{1-b}}2$};
                    \draw[->] (LHS) -- (LHSanchor);
                    \draw[domain=-2.2:-0.35, smooth, variable={\x},black,
                    samples=300] plot
                    ({\x},{0.8*sqrt{((\x*\x*\x)+(4.4*\x*\x)+6.54*\x+3.7401)}});
                \end{tikzpicture}
                \caption{\label{fig:Case1} Illustration for Case 1}
              \end{subfigure}
              \qquad
               \begin{subfigure}[b]{0.25\textwidth}
                \begin{tikzpicture}[scale=1]
                    \node[below] (origin) at (0.12,0) {0};
                    \draw[->] (-2.4,0)--(0.3,0) node[right] {$\zeta_r$};
                    \draw[->] (0,-0.1) -- (0,1.2) node[above] {$\RE(G(\zeta_r))$};
                    \node (LHSanchor) at (-2.2,0) {};
                    \fill (LHSanchor) circle [radius=0.5pt];
                    \node[below] (LHS) at (-2.5,-0.3) {$-\frac{\sqrt{1-b}}2$};
                    \draw[->] (LHS) -- (LHSanchor);
                    \draw[domain=-2.2:-0.35, smooth, variable={\x},black,
                    samples=300] plot
                    ({\x},{sqrt{((\x*\x*\x)+(4.4*\x*\x)+6.0549*\x+2.6728)}});
                    \draw[dashed, opacity=0.3] (-1.85, -0.9) -- (-1.85,1);
                    \node[above] () at (-1.85, 1) {$\xi_1$};
                    \draw[dashed, opacity=0.3] (-1.1, -0.9) -- (-1.1,0.8);
                    \node[above] () at (-1.1, 0.8) {$\xi_2$};
                \end{tikzpicture}
                \caption{\label{fig:Case2} Illustration for Case 2}
              \end{subfigure}
              \qquad
               \begin{subfigure}[b]{0.25\textwidth}
                \begin{tikzpicture}[scale=1]
                    \node[below] (origin) at (0.12,0) {0};
                    \draw[->] (-2.4,0)--(0.3,0) node[right] {$\zeta_r$};
                    \draw[->] (0,-0.1) -- (0,1.2) node[above] {$\RE(G(\zeta_r))$};
                    \node (LHSanchor) at (-2.2,0) {};
                    \fill (LHSanchor) circle [radius=0.5pt];
                    \node[below] (LHS) at (-2.5,-0.3) {$-\frac{\sqrt{1-b}}2$};
                    \draw[->] (LHS) -- (LHSanchor);
                    \draw[domain=-2.2:-0.35, smooth, variable={\x},black,
                    samples=300] plot
                    ({\x},{sqrt{((\x+2.2)*(\x+1.1)*(\x+1.1))}});
                    \draw[dashed, opacity=0.3] (-1.82, -1) -- (-1.82,1);
                    \node[above] () at (-1.82, 1) {$\xi_1$};
                    \draw[dashed, opacity=0.3] (-1.1, -.6) -- (-1.1,1);
                    \node[below] () at (-0.7, -0.55) {$\xi_2 = \zeta_1$};
                \end{tikzpicture}
                \caption{\label{fig:Case3} Illustration for Case 3}
              \end{subfigure}
              \par\bigskip
              \begin{subfigure}[b]{0.4\textwidth}
                \begin{tikzpicture}[scale=1.7]
                    \node[below] (origin) at (0.12,0) {0};
                    \draw[->] (-2.4,0)--(0.3,0) node[right] {$\zeta_r$};
                    \draw[->] (0,-0.1) -- (0,1.2) node[above] {$\RE(G(\zeta_r))$};
                    \node (LHSanchor) at (-2.2,0) {};
                    \fill (LHSanchor) circle [radius=0.5pt];
                    \node[below] (LHS) at (-2.5,-0.3) {$-\frac{\sqrt{1-b}}2$};
                    \draw[->] (LHS) -- (LHSanchor);
                    \draw[black] (-2.2,0) sin (-2.0,0.5);
                    \draw[black] (-2.0,0.5) cos (-1.8,0);
                    \draw[black] (-1.8, 0) sin (-1.6, -0.5);
                    \draw[black] (-1.6, -0.5) cos (-1.4, 0);
                    \node[above] (zeta1) at (-1.4,0) {$\zeta_1$};
                    \fill (zeta1.south) circle [radius=0.5pt];
                    \draw[dashed, opacity=0.3] (-2.0, -0.6) -- (-2.0,1);
                    \node[above] () at (-2.0, 1) {$\xi_1$};
                    \draw[dashed, opacity=0.3] (-1.8, -0.6) -- (-1.8,0.7);
                    \node[above] () at (-1.8, 0.7) {$\hat\zeta$};
                    \draw[dashed, opacity=0.3] (-1.6, -0.6) -- (-1.6,1);
                    \node[above] () at (-1.6, 1) {$\xi_2$};
                    \node[below] (zeta2) at (-0.4,0) {$\zeta_2$};
                    \fill (zeta2.north) circle [radius=0.5pt];
                    \draw[ultra thick,red] (zeta1.south) -- (zeta2.north);
                    \draw[-] (-0.4,0) cos (-0.1,1);
                    \node[above, red, align=center] () at (-0.9,0) {Forbidden\\ region \\
                        for\\ parameter-\\ization};
                \end{tikzpicture}
                \caption{\label{fig:Case4Option1} Illustration for Case 4, option 1}
              \end{subfigure}
              \qquad \qquad
              \begin{subfigure}[b]{0.4\textwidth}
                \begin{tikzpicture}[scale=1.7]
                    \node[below] (origin) at (0.12,0) {0};
                    \draw[->] (-2.4,0) -- (0.3,0) node[right] {$\zeta_r$};
                    \draw[->] (0,-0.1) -- (0,1.2) node[above] {$\RE(G(\zeta_r))$};
                    \node (LHSanchor) at (-2.2,0) {};
                    \fill (LHSanchor) circle [radius=0.5pt];
                    \node[below] (LHS) at (-2.5,-0.3) {$-\frac{\sqrt{1-b}}2$};
                    \draw[->] (LHS) -- (LHSanchor);
                    \draw[black] (-2.2,0) sin (-2.0,0.5);
                    \draw[black] (-2.0, 0.5) cos (-1.8, 0);
                    \node (zeta1anchor) at (-1.8,0) {};
                    \fill (zeta1anchor) circle [radius=0.5pt];
                    \node[below] (zeta1) at (-1.8, -0.3) {$\zeta_1$};
                    \draw[->] (zeta1) -- (zeta1anchor);
                    \draw[dashed, opacity=0.3] (-2.0, -0.6) -- (-2.0,1);
                    \node[above] () at (-2.0, 1) {$\xi_1$};
                    \node[below] (zeta2) at (-0.8,0) {\vspace{1mm}\hspace{-2mm}$\zeta_2$};
                    \fill (zeta2.north) circle [radius=0.5pt];
                    \draw[ultra thick,red] (zeta1anchor.center) -- (zeta2.north);
                    \draw[-] (-0.8,0) sin (-0.6,-0.3);
                    \draw[-] (-0.6,-0.3) cos (-0.2,1);
                    \draw[dashed, opacity=0.3] (-0.6, -0.4) -- (-0.6,1);
                    \node[below] () at (-0.6, -0.4) {$\xi_2$};
                    \node[above, red, align=center] () at (-1.3,0) {Forbidden\\ region \\
                        for\\ parameter-\\ization};
                \end{tikzpicture}
                \caption{\label{fig:Case4Option2} Illustration for Case 4, option 2}
              \end{subfigure}
              \caption{Illustrations of $\zeta_r$ \vs $\RE(G(\zeta_r))$ for the
                four cases in the proof of Lemma \ref{lem:noCrossings}.}
            \end{figure}

          \item An illustration of this case is found in Figure
            \ref{fig:Case4Option2}.  If $\zeta_1$ is the next largest root,
            there are no zeros on $(-\sqrt{1-b}/2,\zeta_1)$. If there were,
            there would be another zero of $P_6$ in $(\xi_1,\zeta_1)$ (so that
            $\RE(G)$ can increase back to zero) but there is not, by assumption.
            For $\zeta_r \in (\zeta_1,\zeta_2)$, the parameterization is not
            valid. $\RE(G( \zeta_2))=0$ and is increasing if $\xi_2<\zeta_2$
            and is decreasing if $\xi_2>\zeta_2$. If $\RE(G)$ is increasing at
            $\zeta_2$, we are done.  If $\RE(G)$ is decreasing at $\zeta_2$,
            then since $\RE(G)\to \infty$ as $\zeta_r \to 0$, there must be
            another zero of $\RE(G)$ in $(\zeta_2,0)$.
        \end{itemize}

    \end{enumerate}
    In either of the two subcases of Case 4, there can be at most one $\zeta_r =
    \hat \zeta_r$ with $\RE G(\hat \zeta_r) = 0$.  However,
    Lemma~\ref{lem:numOfIntersections} below shows that there must be an even
    number of zeros of $\RE(G(\zeta_r))$ for $\zeta_r<0$. It follows that there
    must be 0 intersections and Case 4 is eliminated.  Since Case 1 and Case 3
    also do not pose any problems, we are left with verifying Case 2 only. This
    check is done numerically for some parameters, which completes the proof of
    Lemma \ref{lem:noCrossings}.
  \end{proof}

  \begin{remark}
    The numerical search required for Lemma \ref{lem:noCrossings} need not take
    place over the whole parameter space. Case 2 can only occur when
    \eqref{eqn:Qdiscriminant} holds with strict inequality ($b=F(k)$ corresponds
    to Case 3). Thus our search region covers only those $b$ values satisfying
    $b>\max(k^2,F(k))$.  The search space is shrunk further by looking only for
    those $(b,k)$ pairs satisfying $a_4>0$ in \eqref{eqn:descartes}. $a_4\leq 0$
    if and only if $b\geq G(k)$, where
    \begin{align}
      G(k) := \frac{ E(k) + k^2 K(k)}{2 K(k)} \label{eqn:bgtGk}.
    \end{align}
    The search region is further shrunk by first checking whether or not $P_6$
    has two negative roots, counted with multiplicity. This check needs to be
    done numerically since the roots cannot be found analytically. The search
    region shown in Figure \ref{fig:whenToDoNumericalCheck} indicates where
    $P_6$ has two negative roots. From our numerical tests, fewer than 4\% of
    the grid points in the search region give rise to $P_6$ with negative roots,
    independent of grid spacing. Therefore, fewer than 4\% of the points are
    checked to satisfy $\RE(G(\xi_2))>0$.  Representative plots of
    $\RE(G(\zeta_r))$ near $b=F(k)$ are shown in Figure \ref{fig:bNearFk}. It is
    verified that, for a grid spacing of $10^{-10}$, the condition
    $\RE(G(\xi_2))>0$ is satisfied in the necessary domain. The numerical check
    can be removed if it can be shown that the minimum of $\RE(G(\zeta_r))$ at
    $\xi_2$ is monotonically increasing as $b$ increases from $F(k)$. We are
    not, however, able to prove that at this time.
   \end{remark}

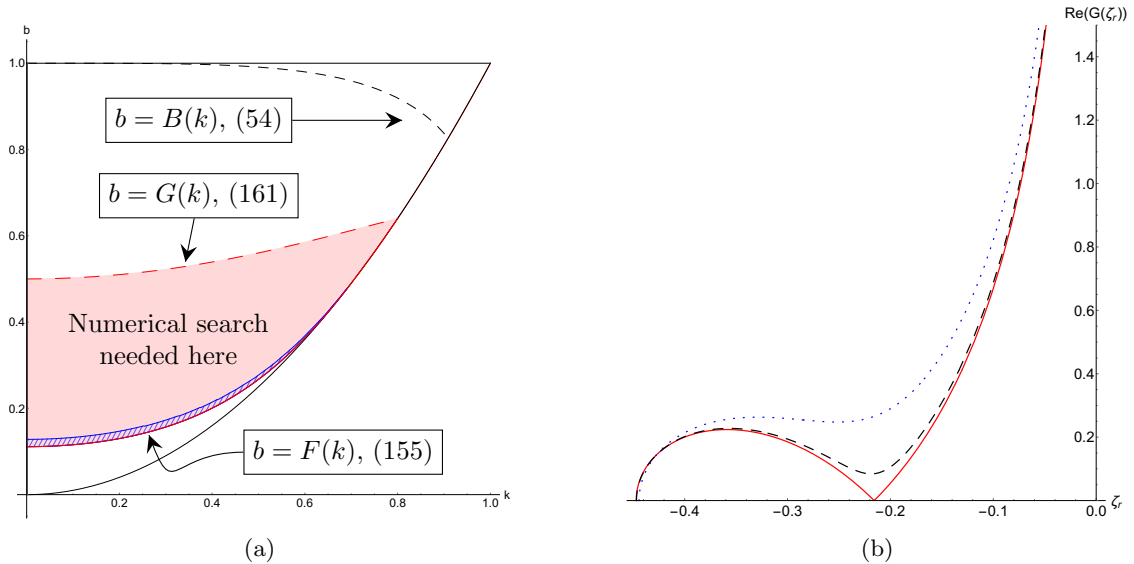
\begin{figure}
  \begin{subfigure}[t]{0.4\textwidth}
    \begin{tikzpicture}[scale=1]
      \node[inner sep=0pt] (test) at (0,0)
      {\includegraphics[width=\textwidth]{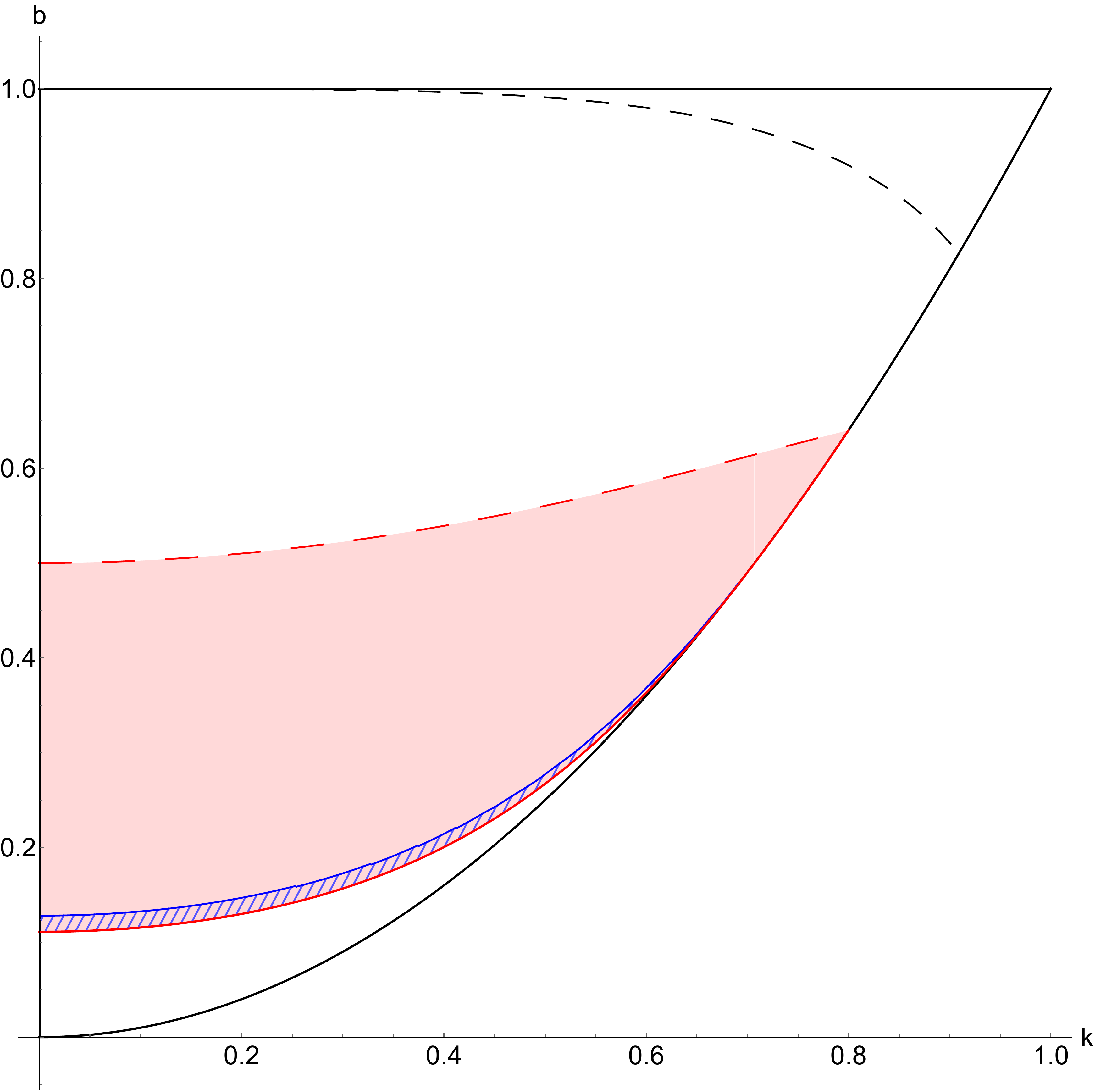}};
      \node (imagCurveAnchor) at (2,2.) {};
      \node[draw] (imagCurve) at (-0.8,2.) {$b=B(k)$, \eqref{eqn:Bk}};
      \draw[decoration={markings,mark=at position 1 with {\arrow[scale=2,
      >=stealth]{>}}}, postaction={decorate}] (imagCurve) -- (imagCurveAnchor) ;
      \node (GkCurveAnchor) at (-1, 0.0) {};
      \node[draw] (GkCurve) at (-0.8, 1.0) {$b = G(k)$, \eqref{eqn:bgtGk}};
      \draw[decoration={markings,mark=at position 1 with {\arrow[scale=2,
      >=stealth]{>}}}, postaction={decorate}] (GkCurve) -- (GkCurveAnchor) ;
      \node (FkCurveAnchor) at (-1.5,-2.05) {};
      \node[draw] (FkCurve) at (1.1, -2.4) {$b = F(k)$, \eqref{eqn:bltFk}};
      \draw[decoration={markings,mark=at position 1 with {\arrow[scale=2,
      >=stealth]{>}}}, postaction={decorate}] (FkCurve) to [out=-180, in=-70,
      looseness=2] (FkCurveAnchor) ;
      \node[align=center, inner sep=0.8mm] (numerical) at (-1.2,-0.9)
            {Numerical search \\ needed here};
    \end{tikzpicture}
    \caption{\label{fig:whenToDoNumericalCheck}}
    \end{subfigure}
    \qquad
    \qquad
    \begin{subfigure}[t]{0.4\textwidth}
    \includegraphics[width=\textwidth]{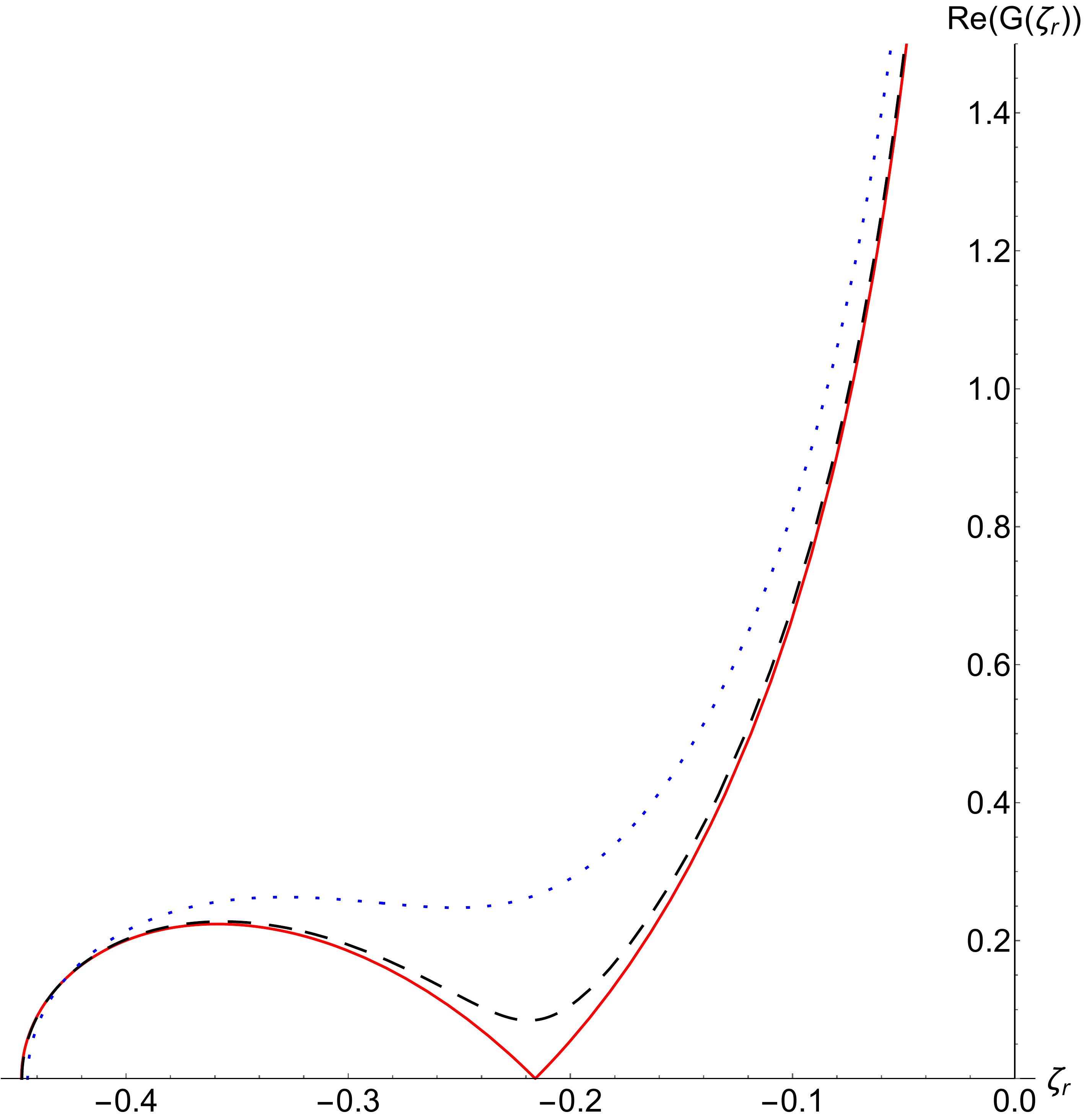}
    \caption{\label{fig:bNearFk}}
    \end{subfigure}
    \caption{(a) The parameter space with curves indicating when a numerical
      check to show that the condition \eqref{eqn:mainStabilityBound} in Theorem
      \ref{thm:mainNTPStabilityResult} is both necessary and sufficient.  For
      more details, see Lemma~\ref{lem:noCrossings}. The dashed blue region just
      above the line $b=F(k)$ indicates where $P_6$ has either 1 or 2 negative
      roots and hence where the numerical check takes place. (b)
      Plots of $\zeta_r$ \vs $\RE(G(\zeta_r))$ near $b=F(k)$ for $k=0.4$.
      The curve $b=F(k)$ is shown in solid red, $b=F(k)+0.001$ in dashed black,
      and $b=F(k) + 0.01$ in dotted blue. See Cases 2 and 3 in the proof of
      Lemma~\ref{lem:noCrossings}. The numerical check in case 2 is to determine
      whether $\RE(G(\zeta_r))=0$ anywhere for $b>F(k)$.
      \label{fig:numericalCheck}}
\end{figure}
\subsection{$\zeta_c\in \R$: an extension of Theorem
  \ref{thm:mainNTPStabilityResult} \label{Appendix:real}}
We first look at cases when $b\leq B(k)$ \eqref{eqn:Bk} so that $\zeta_c
\in \R$.

\begin{lemma}\label{lem:numOfIntersections}
  Let $b\leq B(k)$ so that $\zeta_c\in \R$. Then for $\zeta \in (\C^-\cap
  \sigma_L)\setminus\R$, $\Omega(\zeta)$ has an even number of intersections
  with the imaginary $\Omega$ axis.
\end{lemma}

\begin{proof}
  We note that for $\zeta \in (\C^- \cap \sigma_L)\setminus \R$, $\Omega(\zeta)$
  has 0, 1, or 2 intersections with the imaginary axis by Lemma~
  \ref{lem:noCrossings}. The tangent line to $\sigma_{\cL}$ at the origin is
  given by \cite[equation (104)]{deconinck2017stability},
  \begin{align}
    \DD{\Omega_i}{\Omega_r} &= \pm \frac{(2c -
      \sqrt{1-b}(k^2-2b))E(k)}{(\sqrt{b-k^2}+\sqrt{b})(1 +
      \sqrt{b(b-k^2)}-b)E(k) + (1-k^2)K(k)},
  \end{align}
  with $+$ corresponding to $\zeta_3$ and $-$
  corresponding to $\zeta_2$. It follows that for $\zeta$ near $\zeta_3$ on
  $\sigma_L$, the stability spectrum enters the 1st quadrant of the $\lam$
  plane. For $\zeta\in \sigma_L\setminus\R$ near $-\zeta_c\in\R$, $\zeta =
  -\zeta_c + i\delta_i + \0{\delta_i^2}$ where $\delta_i\in\R$ is a small
  perturbation parameter \cite[equation (150)]{deconinck2017stability}. A short
  calculation gives
  \begin{align}
    \Omega(-\zeta_c + i\delta_i) &= \Omega(-\zeta_c) +
    \frac{i}{2}\frac{\delta_i}{\Omega(-\zeta_c)}(4\zeta_c^3 - 2\omega\zeta_c +
  c),
  \end{align}
  where $\Omega(-\zeta_c)\in i\R$. Then
  \begin{align}
    \begin{split}
    4 \zeta_c^3 - 2\omega\zeta_c + c &=
    \sqrt{\frac{2E(k) - (b-k^2+1)K(k)}{2K^3(k)}}
    \lt( 4E(k) + K(k)(b+k^2-3)\rt)\\
    &= \sqrt{\frac{2E(k) - (b-k^2+1)K(k)}{2K^3(k)}}\lt( k(k')^2\DD{K(k)}{k} +
    2E(k) - K(k)\rt) \geq 0,
    \end{split}
  \end{align}
  since $b<B(k)$. Since $\sigma_{\cL}$ enters the first quadrant from the origin
  and enters the imaginary axis from the first quadrant, it must have an even
  number of crossings with the imaginary axis.  In particular there must be
  either 0 or 2 crossings.
\end{proof}

\noindent Using Theorem \ref{thm:mainNTPStabilityResult}, Lemmas~\ref{lem:noCrossings} and
\ref{lem:numOfIntersections} imply that the condition
\eqref{eqn:mainStabilityBound} is both a necessary and sufficient condition for
\te{spectral} stability when $2E(k)-(1+b-k^2)K(k)\geq0$ by following the exact same proof as
for Theorem~\ref{thm:cnStabilityResult}.

\begin{theorem}
  The sufficient condition for \te{spectral} stability \eqref{eqn:mainStabilityBound} given in
  Theorem \ref{thm:mainNTPStabilityResult} is also necessary.
\end{theorem}
\begin{proof}
  Using Lemma~\ref{lem:noCrossings} we see that $\Omega(\zeta)\in i\R$ for
  $\zeta \in \sigma_L\cap \C^-$ if and only if $\zeta \in \R
  \cup\{\zeta_1,\zeta_2\}$. This means that the bound
  \eqref{eqn:mainStabilityBound} is a necessary and sufficient condition for
  \te{spectral} stability. When $\max(k^2,F(k))< b < G(k)$, Lemma~ \ref{lem:noCrossings}
  relies upon a numerical check.
\end{proof}

\begin{remark}
    If one is not pleased working with the numerical check, then the results in
    this appendix only change in the following manner. The bound
    \eqref{eqn:mainStabilityBound} still determines which solutions are
    spectrally stable with respect to perturbations of period $PT(k)$. It still
    follows that if $Q<P$ and a solution is stable with respect to perturbations
    of period $PT(k)$, then this solution is also \te{spectrally} stable with respect to
    perturbations of period $QT(k)$. The results in the appendix are only needed
    to rule out \te{spectral} stability with respect to other perturbations, \eg perturbations
    with period $RT(k)$ for $R>P$.
\end{remark}

\subsection{A proof of Theorem \ref{thm:imaginaryNTPStabilityResult}, $\zeta_c
\in i\R$ \label{Appendix:imaginary}}
In this subsection we present the details needed to establish Theorem
\ref{thm:imaginaryNTPStabilityResult}.
\begin{lemma}\label{lem:RHSCrossings}
    Let $c\neq 0$, $\zeta_c \in i\R$, and $\zeta\neq \zeta_1$ be in the open
    first quadrant. Then $\Omega(\zeta)\in i\R$ for at most one value of
    $\zeta\in \sigma_L$.
\end{lemma}

\begin{proof}
    The proof is similar to that of Lemma~\ref{lem:noCrossings} with the
    following changes. Here $Q(\zeta_r)$ always has one zero for $\zeta_r>0$.
    Call this zero $\hat\zeta$. Then the parameterization
    \eqref{eqn:parameterizedOmegaLT0} is valid for
    $\zeta_r\in[\hat\zeta,\sqrt{1-b}/2]$. We find that $P_6(\zeta_r)$ has at
    most two positive zeros by Descartes' sign rule. Since $P_6(\zeta_r)$ has at
    most two positive zeros and we know that $\RE(G( \hat \zeta)) =
    \RE(G(\sqrt{1-b}/2))=0$, it follows that there is at most one other
    $\zeta_r$ value at which $\RE(G)=0$.
\end{proof}

\begin{proof}[\textbf{Proof of Theorem~\ref{thm:imaginaryNTPStabilityResult}}]
  We note that if $2E(k) - (1+b-k^2)K(k)<0$, then $\zeta_c \in i\R$ and it must
  be that $\sigma_L$ intersects $i\R \setminus\{0\}$
  (Lemma~\ref{lem:laxSpectrumTopology}, see Figure
  \ref{fig:LaxSpectrumPics}(iv)). Let $\hat\zeta\in i\R\setminus\{0\}$ be the
  intersection point of $\sigma_L$ and $i\R\setminus\{0\}$. Since
  $\RE(\hat\zeta)=0$ and $\IM(\hat\zeta)\neq0$, \eqref{eqn:ReImOmega} implies
  that $\Omega(\hat\zeta)\notin i\R$. By \eqref{eqn:vfsymmetry}, $M(\zeta)$ is
  increasing on $(\zeta_2,\zeta_1)$ except perhaps at $\zeta_c$ if $\zeta_c\in
  \sigma_L$. In any case, since $M(\zeta_2) = M(\zeta_1)=0\mod{2\pi}$,
  $M(\zeta)$ traces out all of $T(k)\mu \in (0,2\pi)$. By Lemma~
  \ref{lem:noCrossings}, $\RE(\lam)>0$ for $\zeta \in (\zeta_2,\hat\zeta]$. By
  Lemma~\ref{lem:RHSCrossings}, $\RE(\lam)=0$ at most at one point in the band
  connecting $\zeta_2$ to $\zeta_1$. Since we need $\RE(\lam)=0$ for $P-1$
  different $\mu$ values different from $0$ for stability by \eqref{eqn:muDefn},
  it follows that there can be stability at most for $P=2$. Since $P=2$
  corresponds to perturbations of twice the period, we have arrived at the
  desired result.

  Finally, we note that the above proof does not rely on the numerical check in
  Lemma~\ref{lem:noCrossings} since the curve $b=B(k)$ \eqref{eqn:Bk}
  always lies above the curve $b=G(k)$ \eqref{eqn:bgtGk} for $k^2<b<1$. To see
  this, we note that $B(k) > G(k)$ if and only if
  \begin{align}
    3E(k) - 2(k')^2 K - k^2 K(k) >0.
  \end{align}
  But
  \begin{align}
    3 E(k) - 2 (k')^2 K(k) -k^2 K(k)> \frac{\pi k^2}4\frac{ 2(1-k^2)^{-3/8} -
    (1-k^2/4)^{-1/2}}{(1-k^2/4)^{1/2}(1-k^2)^{3/8}} > 0
  \end{align}
  for $0<k<\tilde k \approx 0.941952$, where all estimates are found in
  \cite[Section 19.4]{NIST:DLMF}. It can be verified that both $B(\tilde
  k)<\tilde k^2$ and $G(\tilde k)<\tilde k^2$, so we have $B(k)>G(k)$ everywhere
  in the domain $k^2<b<1$, hence no numerical check is needed for solutions
  satisfying $b>B(k)$.
\end{proof}

\subsection{A proof of Lemma~
  \ref{lem:NoDoublePointsSpine}}\label{Appendix:NoDoublePointsProof}

\begin{proof}[\textbf{Proof of Lemma~\ref{lem:NoDoublePointsSpine}}]
  For the cn solutions and the NTP solutions with $b\leq F(k)$ \eqref{eqn:bltFk}
  or $b\geq G(k)$ \eqref{eqn:bgtGk}, Lemmas~\ref{lem:noCrossingscn}
  and~\ref{lem:noCrossings} imply that every $\zeta \in (\C^-\cap\sigma_L)
  \setminus \R$ gives rise to an unstable eigenvalue $\lam(\zeta)$. By
  \cite{gallay_haragus_OrbitalStability}, the elliptic solutions are
  \te{spectrally} stable with
  respect to coperiodic perturbations. Since coperiodic perturbations correspond
  to $T(k)\mu = 0 \mod{2\pi}$, we conclude that in the cases above $M(\zeta)\neq
  0$ for $\zeta$ on the complex bands of the Lax spectrum in the left half
  plane. It is left to show that the same result holds for the NTP solutions
  with $F(k)< b < G(k)$.

  By continuity, an eigenvalue with $T(k)\mu = 0 \mod{2\pi}$ (hereafter called a
  \textit{periodic eigenvalue}) can only enter a complex band by passing through
  the intersection of the complex band with the real axis. Since a periodic
  eigenvalue has $\RE(\Omega(\zeta)) = 0$ by
  \cite{gallay_haragus_OrbitalStability}, it must be the case that the curve
  \eqref{eqn:parameterizedOmegaLT0} intersects the complex band at a periodic
  eigenvalue. Since the intersection of \eqref{eqn:parameterizedOmegaLT0} and the
  complex band must occur immediately upon the periodic eigenvalue entering the
  band, it must be that the curve \eqref{eqn:parameterizedOmegaLT0} and the
  complex band intersect the real axis at the same location, $\zeta =- \zeta_c$
  \eqref{eqn:zetac}. The curve \eqref{eqn:parameterizedOmegaLT0} intersects the
  real axis when $Q(\zeta_r) = 0$ \eqref{eqn:Qzr}. But $Q(\zeta_r)=0$ only at
  the boundary of the region $F(k)< b < G(k)$, when $b=F(k)$. Therefore, in
  order to establish that no periodic eigenvalues enter the complex band, we
  must establish that the zero of $Q(\zeta_r)$ mentioned
  above is not equal to $- \zeta_c$.

  When $b=F(k),$ $Q(\zeta_r)$ has a double zero at $\zeta_r = \tilde\zeta_1<0$:
  \begin{align}
    Q(\zeta_r) &= 4(\zeta_r - \tilde \zeta_1)^2(\zeta_r - \tilde\zeta_2).
  \end{align}
  Comparing the above expression to \eqref{eqn:Qzr}, we find that $\tilde
  \zeta_1^2 = \omega/6$. But
  \begin{align}
    \zeta_c^2 - \tilde \zeta_1^2 &= 2 \lt( E(k) - \frac13(2-k^2)K(k)\rt)\\
    &\geq E(k) - \frac23 K(k) > \sqrt{1-k^2}K(k) - \frac23 K(k) >0,
  \end{align}
  for $k^2<5/9$ (the inequality used for $E(k)$ comes from \cite[Section
  19.9]{NIST:DLMF}).  Since $ b = F(k) < k^2$ only when $k^2 < 1/2 < 5/9$, we
  find that the intersection of $Q(\zeta_r)$ with the real line is well
  separated from the intersection of the complex band with the real line for all
  allowed $k$. It follows that no periodic points can enter the complex band in
  the left half plane. We finish the proof by noting that since $2\pi >
  M(-\zeta_c)>M(\zeta_c)$, periodic points also can not enter the complex band
  in the right half plane.

\end{proof}

\te{
\subsection{A proof of Theorem \ref{thm:squaredEigCompleteness} \label{Appendix:SQEFCompleteness}}

\begin{proof}[\textbf{Proof of Theorem~\ref{thm:squaredEigCompleteness}}]
  The proof is similar to the proof of \cite[Theorem~2]{bottman2011elliptic}. We
  provide details omitted there.

    For every $\lam \in \C$, \eqref{eqn:spectralProblem} can be written as a
    four-dimensional first-order system of ordinary differential equations. For
    each $\lam \in \C$, one value of $\Omega$ is obtained through $\Omega =
    \lam/2$. Defining
    \begin{align}
      \tilde Q_4(\zeta) := -\zeta^4 + \omega\zeta^2 + c\zeta -
      \frac1{16}\lt(4\omega b + 3b^2 + (1-k^2)^2\rt),
    \end{align}
    and
    \begin{align}
       Q_4(\Omega,\zeta) := \Omega^2 - \tilde Q_4(\zeta),
    \end{align}
    we let
    \begin{align}
      \cB := \{ \lam \in \C : \text{ the discriminant of
        $Q_4$ with respect to $\zeta$ vanishes}\}.
    \end{align}
    For $\lam \in \C\setminus \cB$, the zeros of $Q_4(\Omega,\zeta)$ give four
    values of $\zeta\in\C$. It is not necessary that each of these four values
    of $\zeta$ are in the Lax spectrum since this counting argument is
    independent of the Lax spectrum. The squared-eigenfunction connection
    \eqref{eqn:squaredEF} gives a solution to \eqref{eqn:spectralProblem} for
    each of the four $\zeta \in \C$. Therefore, \eqref{eqn:squaredEF} gives four
    solutions of the fourth-order problem \eqref{eqn:spectralProblem} for each
    $\lam \in \C \setminus \cB$. We first show that the four solutions obtained
    through \eqref{eqn:squaredEF} are linearly independent for $\lam \in
    \C\setminus \cB$, then later we will look at $\lam \in \cB$.

    Using the fact that
    \begin{align}
      B_x &= 2(-i\zeta B -\phi A),
    \end{align}
    the eigenfunctions \eqref{eqn:gamma(x)} may be written as
    \begin{align}
      \begin{split}
      \chi(x,t) &= e^{\Omega t} \begin{pmatrix} -B \\
        A-\Omega\end{pmatrix} \gamma_0 \exp\lt( -\int \lt( \frac{B_x}{2B} +
      \frac{\phi \Omega}{B}\rt)~\d x\rt)\\
      &= e^{\Omega t}\begin{pmatrix}-B \\ A-\Omega\end{pmatrix}
      \frac{\gamma_0}{B^{1/2}} \exp\lt( -\int \frac{\phi\Omega}{B}~\d x\rt).
    \end{split}
    \end{align}
    When $\lam \in \C\setminus(\cB \cup \{0\})$, the above gives four
    eigenfunctions, one for each $\zeta$. The four eigenfunctions have
    singularities at the zeroes of $B$. Since the zeros of $B$ depend on
    $\zeta$, the four eigenfunctions have different singularities in the complex
    $x$ plane for the four different values of $\zeta$. When $\Omega = 0$, there
    exist two bounded eigenfunctions \cite[Proposition
    3.2]{gallay_haragus_OrbitalStability}. Only one of these is obtained through
    \eqref{eqn:gamma(x)}.

    We now consider the six values of $\lam \in \cB$.
    The discriminant can only vanish in one of the following cases:
    \begin{enumerate}
      \item $Q_4 = (\zeta - \hat \zeta_1)(\zeta - \hat
        \zeta_2)(\zeta - \hat \zeta_3)^2 = 0,$
      \item $Q_4 = (\zeta - \hat \zeta_1)^2 (\zeta - \hat \zeta_2)^2 = 0,$
      \item $Q_4 = (\zeta - \hat \zeta_1) (\zeta - \hat \zeta_2)^3 = 0$, or
      \item $Q_4 = (\zeta - \hat \zeta_1)^4 = 0.$
    \end{enumerate}
    The zeros of $Q_4$ come from level sets of $\tilde Q_4(\zeta)$. Case 4 can
    only occur when the graph of $\tilde Q_4(\zeta)$ has one maximum.
    However, since we know from \eqref{eqn:zetaRoots} that all four roots of
    $\tilde Q_4(\zeta)$ cannot be equal, case 4 is not possible. Case 3 can also
    be ruled out since the four roots \eqref{eqn:zetaRoots} of $\tilde
    Q_4(\zeta)$ are real. Case 2 can only occur when two roots of
    \eqref{eqn:zetaRoots} collide, which can only occur for the cn or dn
    solutions. The stability of these cases has been determined
    \cite{gustafson2017stability} so they are not a concern here. Finally, case
    1 is possible. In case 1, only three values of $\zeta$ are determined from
    $\Omega$. In such a case, three linearly independent solutions of
    \eqref{eqn:spectralProblem} are found. The fourth is obtained using
    reduction of order and introduces algebraic growth so it is not an
    eigenfunction. Therefore in this case, all eigenfunctions are found using
    the squared-eigenfunction connection.

  \end{proof}
}

\end{appendices}

{\footnotesize
\bibliographystyle{siam}
\bibliography{mybib}}

\end{document}